\documentclass[draft]{ecta_like_modified}

\RequirePackage[OT1]{fontenc}
\RequirePackage[dvipdfmx]{graphicx}
\RequirePackage{amsthm}
\RequirePackage[cmex10]{amsmath}
\RequirePackage[longnamesfirst]{natbib}
\RequirePackage[colorlinks,citecolor=blue,urlcolor=blue]{hyperref}
\RequirePackage{hypernat}
\usepackage{colortbl}

\usepackage{algorithm}
\usepackage{algpseudocode}
\usepackage[top=30truemm,bottom=30truemm,left=30truemm,right=30truemm]{geometry}

\usepackage[utf8]{inputenc} 
\usepackage{url}            
\usepackage{booktabs}       
\usepackage{amsfonts}       
\usepackage{amssymb}
\usepackage{nicefrac}       
\usepackage{microtype}      
\usepackage{lipsum}
\usepackage{comment}
\usepackage{setspace}
\usepackage{float}
\graphicspath{ {./images/} }
\usepackage{quoting}
\usepackage{wrapfig}
\quotingsetup{font={itshape}}
\usepackage{multirow}
\usepackage{bm}
\usepackage{longtable}

\startlocaldefs
\numberwithin{equation}{section}
\theoremstyle{definition}

\newtheorem{theorem}{Theorem}[section]

\newtheorem{lemma}{Lemma}[section]
\newtheorem{remark}{Remark}[section]
\newtheorem{assumption}{Assumption}[section]

\renewcommand{\thetable}{\arabic{section}.\arabic{table}}
\renewcommand{\thefigure}{\arabic{section}.\arabic{figure}}
\endlocaldefs

\begin{document} 	

\begin{frontmatter}

\title{Optimal estimation for regression discontinuity design \\ 
with binary outcomes\protect\thanksref{T1}\protect\thanksref{T2}}

\runtitle{Optimal estimation for regression discontinuity design \\ 
with binary outcomes}
\thankstext{T1}{
This study was supported by JSPS KAKENHI Grant Numbers JP22K13373 (Ishihara) and JP21K13269 (Sawada). We thank Yu-Chang Chen, Atsushi Inoue, Timothy Neal, Michal Koles\'ar, Soonwoo Kwon, Tomasz Olma, and Ke-Li Xu, as well as seminar participants at the Japanese Joint Statistical Meeting, Hitotsubashi University, Kansai Keiryo Keizaigaku Kenkyukai, the Tohoku-NTU Joint Seminar, the Econometric Society World Congress 2025, and LMU-Todai Econometrics Workshop, for their insightful comments.
}
\thankstext{T2}{First Version: September 23, 2025; Current Version: \today}
\begin{aug}
\author{\fnms{Takuya} \snm{Ishihara}\thanksref{a}\ead[label=e1]{takuya.ishihara.b7@tohoku.ac.jp}}
\author{\fnms{Masayuki} \snm{Sawada}\thanksref{b}\ead[label=e2]{masayuki.sawada@r.hit-u.ac.jp}}
\author{\fnms{Kohei} \snm{Yata}\thanksref{c}\ead[label=e3]{yata@wisc.edu}}

\address[a]{Tohoku University, Graduate School of Economics and Management
}

\address[b]{Hitotsubashi University, Institute of Economic Research
}

\address[c]{The University of Wisconsin--Madison, Department of Economics
}

\end{aug}

\begin{abstract}
We develop a finite-sample optimal estimator for regression discontinuity design when the outcomes are bounded, including binary outcomes as the leading case. Our estimator achieves minimax mean squared error among linear shrinkage estimators with nonnegative weights when the regression function lies in a Lipschitz class. Although the original minimax problem involves an iterative noncovex optimization problem, we show that our estimator is obtained by solving a convex optimization problem. A key advantage of the proposed estimator is that the Lipschitz constant is its only tuning parameter. We also propose a uniformly valid inference procedure without a large-sample approximation. In a simulation exercise for small samples, our estimator exhibits smaller mean squared errors and shorter confidence intervals than those of conventional large-sample techniques. In an empirical multi-cutoff design in which the sample size for each cutoff is small, our method yields informative confidence intervals, in contrast to the leading large-sample approach.
\end{abstract}


\end{frontmatter}

\section{Introduction}\label{sec:intro}

A large-sample approximation is the basis for the leading estimators for regression discontinuity (RD) designs \citep[for example]{Imbens.Kalyanaraman2012,Calonico.Cattaneo.Titiunik2014}. 
RD designs involve the estimation of conditional expectation functions at a cutoff point on the support of a running variable. Hence, effective observations are limited to the neighborhood of the cutoff, and the number of these observations can be small, even if the total sample size is large \citep{Cattaneo.Frandsen.Titiunik2015,Canay.Kamat2017}. For example, an effective sample can be small for designs with multiple cutoffs, with a cutoff at the tail of the distribution, or with subgroup analyses. In small samples, the large-sample asymptotics may not provide good approximations of the behaviors of the existing estimators; hence, their desirable properties may be lost.

Some studies have considered finite-sample minimax estimators for RD designs.\footnote{Throughout the manuscript, we compare our estimator with existing finite-sample minimax estimators. Another notable approach is finite-sample valid estimation and inference based on the local randomization of the RD design \citep{Cattaneo.Frandsen.Titiunik2015, cattaneoInferenceRegressionDiscontinuity2016, cattaneoComparingInferenceApproaches2017}. The local randomization approach is based on the assumption that the running variable is randomly assigned with a constant regression function within a given small window around the threshold \citep{Cattaneo_Idrobo_Titiunik_2024}, whereas we consider a smooth but nonconstant regression function within the window.} For example, \cite{Armstrong.Kolesar2018} and \cite{Imbens.Wager2019} propose finite-sample minimax linear estimators under the smoothness of the regression function.
However, these minimax estimators require knowledge of the conditional variance function, which is generally unavailable in practice. Although the variance can be estimated, we cannot guarantee the theoretical validity of the plug-in estimators with the estimated variance in finite samples.
Furthermore, the construction of finite-sample valid confidence intervals based on these estimators additionally requires the normality of the regression errors.

In this study, we propose finite-sample estimation and inference methods for RD designs with binary outcomes. 
For a binary dependent variable, all features of its conditional distribution, including its conditional variance, are \textit{known} functions of its conditional mean function. We establish the finite-sample validity of our methods under a smoothness restriction on the conditional mean function by considering the implicit restrictions imposed on the entire conditional distribution. Therefore, our procedure is both feasible and theoretically valid without knowledge or estimation of the conditional variance or, more generally, any features of the conditional distribution, except for the smoothness of the conditional mean.

Specifically, we consider a minimax optimal estimator among a class of {\it linear shrinkage} estimators for the regression function at a boundary point, under the assumption that the regression function satisfies Lipschitz continuity. The class of linear shrinkage estimators has the form $\sum_{i=1}^nw_i(Y_i-1/2)+1/2$ with $\sum_{i=1}^nw_i\le 1$ and $w_i\ge 0$, where $Y_1,...,Y_n$ are the observed outcomes on either side of the boundary. The shrinkage toward $1/2$ is motivated by the fact that the regression function is bounded and takes values in $[0,1]$, leading to a scope of efficiency gain by shrinkage. Given this class of linear shrinkage estimators, we derive a linear shrinkage estimator that minimizes the maximum mean squared error (MSE) under Lipschitz continuity with a known Lipschitz constant. In other words, we assume the researcher's a priori knowledge of the bound of how much the function value can change if the running variable is changed by one unit. We emphasize that the Lipschitz constant is the only tuning parameter used. Furthermore, we show that the minimax estimator is a solution to a convex optimization problem that is computationally feasible. Thus, we provide a practical estimator that achieves finite-sample optimality in the binary-outcome setting.

Our estimator is applicable to many practical RD designs. Binary outcomes are among the most common types of outcomes in empirical applications. For example, the following outcome variables are all binary: an indicator for winning the next election in the famous U.S. House election study by \cite{Lee2008}; a corruption indicator in \cite{Brollo.Nannicini.Perotti.Tabellini2013}; a mortality indicator in \cite{card_does_2009}; and indicators for student's enrollment and dropout in \cite{Melguizo2016} and \cite{Cattaneo2021multi}. Furthermore, the first stage in fuzzy RD designs often involves treatment status as a binary dependent outcome. Moreover, the minimax optimality of our estimator for binary outcomes immediately extends to that for bounded outcomes because the variance of any linear estimator is maximized when the outcomes are Bernoulli, given the conditional mean function. Hence, our estimator can be applied not only to binary outcomes but also to bounded outcomes, which are also frequently used in the RD design. 

Our method also complements existing minimax estimators. We compare our estimator to a version of the existing minimax estimators \citep{Armstrong.Kolesar2018,Imbens.Wager2019} and demonstrate that our method has better finite-sample performance than the existing approach, although their asymptotic behaviors are similar. Specifically, we consider a minimax linear estimator obtained under a misspecified model in which the conditional mean and variance are unrelated, the variance is known, and the regression function lies in a Lipschitz class with no bounds on the function values. This estimator is not directly feasible in our binary-outcome setting, in which the variance is unknown. As a feasible version of this estimator, we consider one obtained under the assumption of a constant variance of $1/4$, which is the maximum possible variance of a binary variable. For binary outcomes, we theoretically show that the efficiency gain from our estimator, relative to the alternative estimator above, tends to vanish as the sample size increases. Nevertheless, for small samples, we numerically demonstrate that the alternative method can result in a $5\%$ to $20\%$ increase in the worst-case root MSE owing to model misspecification. Hence, our method supplements the existing minimax estimators with better finite-sample performance and similar asymptotic behaviors in a binary-outcome setting.

We also propose confidence intervals with correct coverage in finite samples uniformly over the Lipschitz class. We construct confidence intervals by inverting one- or two-sided uniformly valid tests that use a linear estimator as the test statistic. To construct a uniformly valid test, we propose a simulation-based approximation to the distribution of the test statistic by drawing samples from a multivariate Bernoulli distribution that satisfies the null restriction. We then numerically optimize the critical value so that the worst-case rejection probability is equal to or smaller than the significance level. A computational challenge with this approach is the calculation of the worst-case rejection probability, which involves the optimization of an $(n+1)$-dimensional parameter.
We overcome this challenge by deriving a simple characterization of the worst-case rejection probability under Lipschitz continuity, which significantly reduces the computational burden.
We also emphasize that our confidence intervals are valid in finite samples for binary outcomes.
This contrasts with existing inference methods, which are based on either a large-sample approximation or a restrictive assumption of Gaussian errors with a known variance.

The same inference approach does not apply to bounded outcomes because the simple characterization of the worst-case rejection probability relies on the fact that the outcome is binary. For bounded outcomes, we provide an alternative finite-sample inference procedure based on a uniform bound on the rejection probability obtained using Hoeffding's inequality. The resulting confidence intervals have correct coverage in finite samples but can be conservative, similar to Hoeffding's inequality-based confidence intervals in other contexts.

We demonstrate the performance of our methods through simulations and an empirical application. In the simulations, our estimator achieves substantially smaller MSEs relative to the leading large-sample estimators when the sample size is small. Furthermore, our estimator behaves similarly to the large-sample estimators when the sample size is large; the differences in MSE decrease as the number of observations increases. Our proposed inference method also achieves guaranteed coverage rates with shorter confidence intervals when the sample size is small. Hence, our methods, while theoretically valid, are also useful in practice.

We illustrate our methods by revisiting \cite{Brollo.Nannicini.Perotti.Tabellini2013}, who estimate the impact of additional government revenue on corruption. They exploit the regional fiscal rule in Brazil, where federal transfers to municipal governments change exogenously at given population thresholds.
This is a multi-cutoff RD design with a small sample size near each cutoff point.
We demonstrate that our estimates are similar to the conventional estimates for large-sample pooling of multiple cutoffs. Nevertheless, our inference method provides much shorter confidence intervals than conventional methods when we focus on a small sample near each cutoff value. Consequently, our estimates provide more informative results than conventional methods.

Both the simulation and application results indicate that the small-sample estimations are generally challenging, whereas our estimator has the potential to provide informative estimates. Hence, our estimator is a practical last resort for an empirical researcher facing a research question with a small effective sample size for an RD design.

In addition to contributing to estimation in RD design, 
we contribute to the vast literature on minimax estimation.
\cite{donoho1994} considers minimax affine estimation and inference on linear functionals in nonparametric regression models with Gaussian errors. Recently, this framework has been applied to the estimation and inference of treatment effects in various settings, including RD designs \citep{Armstrong.Kolesar2018,Armstrong2021ATE,Gao2018,Imbens.Wager2019,Kwon.Kwon2020,Chaisemartin2021,rambachan2023parallel}.
We complement these studies by examining nonparametric regression models with Bernoulli dependent variables, which are not covered by their frameworks.
To the best of our knowledge, no general minimax estimator under the squared error loss has been established for the problem of estimating linear functionals in this setting.\footnote{\cite{DeRouen1974} derive a $\Gamma$-minimax estimator for a linear combination of the success probabilities of multiple independent binomial variables when the class of prior distributions consists of distributions with the same, known means.} No solution is known, even for the estimation of the difference in the success probability between two independent binomial variables with unequal numbers of trials \citep[Example 5.1.9]{Lehmann1998}.\footnote{For the estimation of the success probability of a single binomial variable, a linear shrinkage (toward $1/2$) estimator is minimax among all estimators \citep[Example 5.1.7]{Lehmann1998}.
\cite{Marchand2000} consider this problem with a restricted parameter space. They show that, when the success probability is known to lie in a symmetric interval around $1/2$, a linear shrinkage estimator is minimax among all linear estimators.}
We contribute to this underexplored literature by developing a minimax estimator for a regression function at a point within the class of linear shrinkage estimators under the Lipschitz continuity of the regression function.

\section{Our minimax estimator and its properties}\label{sec:main}

RD designs exploit a discontinuous change in treatment status when a running variable exceeds a cutoff point. For example, \cite{Brollo.Nannicini.Perotti.Tabellini2013} exploit discontinuous increases in the amount of central government subsidies for a local government when its population equals or exceeds a threshold level. The target parameter of the RD design is the average treatment effect at the cutoff point, which is identified as the difference between the conditional expectation functions evaluated at the cutoff point. Hence, its estimation involves the nonparametric estimation of the conditional mean functions at their boundary points.

\subsection{Setting}\label{sec:settings}
Suppose that we have a random sample $\{Y_i,D_i,R_i\}_{i=1}^N$, where $R_i \in \mathbb{R}^{d_r}$ is a $d_r (\geq 1)$-dimensional vector of running variables, $Y_i$ is a binary outcome, $D_i$ is a binary treatment assigned as $D_i = 1\{R_i \in \mathcal{T}\}$, and $\mathcal{T} \subset \mathbb{R}^{d_r}$ is a known treated region. The leading case is that in which $R_i$ is univariate ($d_r = 1$) and $\mathcal{T} = [c,\infty)$ for some known cutoff $c$; however, the following arguments also apply to a multidimensional case (i.e., $d_r > 1$). Suppose
\begin{equation}
    Y_i = f(D_i,R_i) + U_i, \ \ E[U_i|D_i,R_i]=0, \label{model} \nonumber
\end{equation}
for some unknown function $f:\{0,1\}\times\mathbb{R}^{d_r}\rightarrow[0,1]$.
Let $R_0$ be a fixed boundary point in the treatment region $\mathcal{T}$.
When $f(d,r)$ represents the conditional expectation function of the underlying potential outcome $Y_{i}(d)$ conditional on $R_i = r$ for each $d \in \{0,1\}$, $f(1,R_0) - f(0,R_0)$ is interpreted as the average treatment effect at the boundary point $R_0$ \citep{hahnIdentificationEstimationTreatment2001}. The data $\{Y_i,D_i,R_i\}_{i=1}^N$ can be divided into $\{ Y_{i,+}, R_{i,+} \}_{i=1}^{n_{+}}$ and $\{ Y_{i,-}, R_{i,-} \}_{i=1}^{n_{-}}$, where the former is the data from the treatment group and the latter is the data from the control group. We use the two samples separately to estimate $f(1,R_0)$ and $f(0,R_0)$, respectively.

Without loss of generality, we consider the estimation of $f(1,R_0)$ throughout this section, except in Remark \ref{rem:ATE} at the end of this section, where we discuss the estimation of $f(1,R_0) - f(0,R_0)$. To simplify the notation, we use $\{Y_i,R_i\}_{i=1}^n$ to denote $\{ Y_{i,+}, R_{i,+} \}_{i=1}^{n_{+}}$, so that $R_i\in\mathcal{T}$ for all $i=1,...n$.
Furthermore, we use $f(\cdot)$ to denote $f(1,\cdot)$.
Additionally, our analysis conditions on the realization of $\{R_i\}_{i=1}^n$, and we treat $\{R_i\}_{i=1}^n$ as deterministic, so that $P(Y_i=1)=f(R_i)$ for all $i=1,\ldots,n$.
Let $p_i \equiv f(R_i)$ for $i=0, 1, \ldots, n$ and $\bm{p} \equiv (p_0,p_1, \ldots, p_n)' \in [0,1]^{n+1}$. Without loss of generality, we assume that $R_0=0$ and $\|R_0\| \leq \|R_1\| \leq \cdots \leq \|R_n\|$, where $\|\cdot\|$ is a norm on $\mathbb{R}^{d_r}$. The following theoretical result holds for any norm. We focus on the Euclidean norm in our numerical exercises, simulations, and empirical application.

For the parameter of interest $p_0 = f(R_0) = f(0)$, we consider the following linear shrinkage estimator:
\begin{equation}
    \hat{p}_0(\bm{w}) \equiv \frac{1}{2} + \sum_{i=1}^n w_i \left( Y_i-\frac{1}{2} \right), \ \ \bm{w} \equiv (w_1, \ldots , w_n)' \in \mathcal{W}, \label{linear_estimator}
\end{equation}
where $\mathcal{W} \equiv \left\{ \bm{w} \in \mathbb{R}^n : \sum_{i=1}^n w_i \leq 1 \ \text{and} \ w_i \geq 0 \ \text{for all $i$}  \right\}$. When $\sum_{i=1}^n w_i =1$, $\hat{p}_0(\bm{w})=\sum_{i=1}^n w_iY_i$, and no shrinkage occurs. When $\sum_{i=1}^n w_i < 1$, $\hat{p}_0(\bm{w})$ is an estimator that shrinks toward $1/2$.

We assume that $f$ belongs to a Lipschitz class:
\begin{equation}
    \mathcal{F}_{\text{Lip}}(C) \equiv \left\{ f: \left| f(r) - f(r') \right| \leq C \| r-r' \| \ \text{and} \ f(r) \in [0,1]  \right\}, \label{Lipschitz_class}
\end{equation}
where $C$ denotes the Lipschitz constant and is known. This assumption implies that $\bm{p} \in [0,1]^{n+1}$ satisfies $|p_i-p_j| \leq C \|R_i - R_j\|$ for all $i$ and $j$. Conversely, if $|p_i-p_j| \leq C \|R_i - R_j\|$ for all $i$ and $j$, we can find a function $f\in\mathcal{F}_{\text{Lip}}(C)$ such that $f(R_i)=p_i$ for all $i$ \citep{Beliakov2006}. Hence, the parameter space of $\bm{p}$ can be expressed as follows:
\begin{equation}
    \mathcal{P} \equiv \left\{ \bm{p} \in [0,1]^{n+1} : |p_i-p_j| \leq C \|R_i - R_j\| \ \text{for all $i$ and $j$} \right\}. \label{P_class}\nonumber
\end{equation}

Since $Y_1,...,Y_n$ are independent binary variables, the mean squared error (MSE) of $\hat{p}_0(\bm{w})$ is given by
\begin{eqnarray*}
    \text{MSE}(\bm{w},\bm{p}) &\equiv & E\left[ \left( \hat{p}_0(\bm{w}) - p_0 \right)^2 \right] \\
    &=& \left\{ \frac{1}{2} + \sum_{i=1}^n w_i \left( p_i - \frac{1}{2} \right) - p_0 \right\}^2 + \sum_{i=1}^n w_i^2 p_i \left( 1 - p_i \right).
\end{eqnarray*}
We consider the linear shrinkage estimator whose corresponding weight vector solves the following problem:
\begin{equation}
    \min_{\bm{w} \in \mathcal{W}}  \max_{\bm{p} \in \mathcal{P}} \text{MSE}(\bm{w},\bm{p}). \label{minimax_problem_1}
\end{equation}
To simplify the expression in \eqref{minimax_problem_1}, we redefine $p_i$ as $\theta_i \equiv p_i-1/2$ for $i=0, 1, \ldots, n$ and let $\bm{\theta} \equiv (\theta_0, \theta_1, \ldots , \theta_n)'$; thus, the problem is
\begin{equation}
    \min_{\bm{w} \in \mathcal{W}} \, \max_{\bm{\theta} \in \Theta} \text{MSE}(\bm{w},\bm{\theta}), \label{minimax_problem_2}
\end{equation}
where $\Theta \equiv \{\bm{\theta} \in [-1/2,1/2]^{n+1} : |\theta_i - \theta_j| \leq C \| R_i-R_j \| \ \text{for all $i$ and $j$} \}$ and
\begin{eqnarray*}
    \text{MSE}(\bm{w},\bm{\theta}) & \equiv & \left( \sum_{i=1}^n w_i \theta_i - \theta_0 \right)^2 + \sum_{i=1}^n w_i^2 \left( \frac{1}{4} - \theta_i^2 \right).
\end{eqnarray*}
Hence, we aim to obtain a weight vector that minimizes the maximum MSE by solving \eqref{minimax_problem_2}.

\begin{remark}\label{rem:shrinkage}
The class of linear shrinkage estimators (\ref{linear_estimator}) eliminates linear estimators with negative weights. Hence, it excludes local polynomial estimators (except for local constant estimators), which are commonly employed in RD designs. Nevertheless, the linear minimax MSE estimator has nonnegative weights in related setups in which the outcome is nonbinary  (e.g., Gaussian outcomes) and its regression function lies in the Lipschitz class with a known conditional variance: see Section \ref{sec:gauss} and Appendix \ref{sec:gauss_weights}.
Hence, we focus on linear shrinkage estimators with nonnegative weights.
\end{remark}

\begin{remark}\label{rem:Lipschitz}
Shape restrictions on second derivatives are common in studies of honest inference in RD designs \citep[e.g.,][]{kolesar2018discrete,Imbens.Wager2019,noack_bias_aware_2024}. One example is imposing bounds on second derivatives, which aligns with local linear estimators. We focus on the Lipschitz class for two reasons. First, restrictions on the second derivatives are less transparent and more challenging to evaluate than the Lipschitz constraints, which bound the partial effects of the running variable on the outcome. Second, the bounded second derivative implies the bounded first derivative when the regression function is bounded.
To see this, suppose that the domain of $f$ is $\mathbb{R}$ and the absolute value of the second derivative $f''(x)$ is bounded by $C > 0$, so that $f'(x+u) \ge f'(x) - Cu$ for $u>0$. Then, we obtain $f(x+\delta)-f(x) = \int_0^\delta f'(x+u) du \geq f'(x) \delta - C \delta^2 / 2$ for any $\delta > 0$. If the range of $f$ is $[0,1]$, $f(x+\delta)-f(x)$ must be less than or equal to $1$. Consequently, the first derivative satisfies $f'(x) \leq \delta^{-1} + C \delta / 2$ for any $\delta > 0$, which implies that $f'(x) \leq \min_{\delta>0}(\delta^{-1} + C \delta / 2)=\sqrt{2C}$.
Similarly, we have $f'(x) \ge -\sqrt{2C}$.
In other words, the absolute value of the first derivative is bounded by $\sqrt{2C}$ when the absolute value of the second derivative $f''(x)$ is bounded by $C$ and the range of $f$ is $[0,1]$. Thus, the second-derivative restriction is closely related to the Lipschitz constraint for bounded outcomes.
\end{remark}

\begin{remark}\label{rem:bounded_outcome}
The solution to (\ref{minimax_problem_1}) is also a minimax linear shrinkage estimator for bounded outcomes. Consider the estimation of $p_0$ under the assumption that $P(0 \leq Y_i \leq 1) = 1$ and $\bm{p}\in \mathcal{P}$, where $p_i=E[Y_i]$. We impose no additional assumptions on $Y_i$. Then, the variance of $Y_i$ must be less than or equal to $p_i(1-p_i)$ because we have
\[
Var(Y_i) = E[Y_i^2] - E[Y_i]^2 \leq E[Y_i] - E[Y_i]^2 = p_i(1-p_i),
\]
where the inequality follows from $P(Y_i^2 \leq Y_i)=1$. Since the bias of a linear estimator is the same for bounded and binary outcomes, the worst-case MSE for bounded outcomes is equal to the worst-case MSE for binary outcomes. Hence, the solution to (\ref{minimax_problem_1}) is also a minimax linear shrinkage estimator when $Y_i \in [0,1]$ and $\bm{p} \in \mathcal{P}$.
\end{remark}

\begin{remark}\label{rem:holder}
Our theoretical analysis extends to a class of H\"{o}lder continuous functions:
\begin{align}
    \mathcal{F}_{\text{H\"{o}l},\gamma}(C) \equiv \left\{ f: \left| f(r) - f(r') \right| \leq C \| r-r' \|^\gamma \ \text{and} \ f(r) \in [0,1]  \right\}, \label{holder_class}
\end{align}
where $\gamma\in (0,1]$ and $C$ are known constants.\footnote{The corresponding space of $\bm{p}$ can be expressed as $\left\{ \bm{p} \in [0,1]^{n+1} : |p_i-p_j| \leq C \|R_i - R_j\|^\gamma \ \text{for all $i$ and $j$} \right\}$. This is because, as in the Lipschitz case, we can show that $|p_i-p_j| \leq C \|R_i - R_j\|^\gamma$ for all $i$ and $j$ if and only if there exists $f\in\mathcal{F}_{\text{H\"{o}l},\gamma}(C)$ such that $f(R_i)=p_i$ for all $i$, using arguments similar to those used in Theorem 4 of \cite{Beliakov2006}.} This class includes functions that are not smooth enough to satisfy Lipschitz continuity; a larger exponent $\gamma$ corresponds to a smoother class, with $\gamma=1$ corresponding to the Lipchitz class.
All of our main theoretical results, except for the asymptotic result in Theorem \ref{thm:asymptotic} in Section \ref{sec:compare_gauss}, also hold for the H\"{o}lder class $\mathcal{F}_{\text{H\"{o}l},\gamma}(C)$ by replacing $\|\cdot\|$ with $\|\cdot\|^{\gamma}$ in the results for the Lipschitz class.
This extension is possible because the proofs for the Lipschitz class rely only on the fact that every norm $\|\cdot\|$ is subadditive, nonnegative, and even, and these properties are also satisfied by $\|\cdot\|^{\gamma}$ for $\gamma\in (0,1)$ (i.e., $\|r+r'\|^\gamma\le \|r\|^\gamma+\|r'\|^\gamma$, $\|r\|^\gamma\ge 0$, and $\|-r\|^\gamma=\|r\|^\gamma$).
Furthermore, the asymptotic result in Theorem \ref{thm:asymptotic} can also be extended to the H\"{o}lder class: see the proof of Theorem \ref{thm:asymptotic} in Appendix \ref{sec:proofs} for details.
\end{remark}

\subsection{Computing the worst-case MSE of a linear shrinkage estimator}\label{sec:minimax}

Our goal is to obtain a weight vector $\bm{w}$ that minimizes the maximum MSE. First, we consider the maximization part of (\ref{minimax_problem_2}) for a given weight vector $\bm{w}\in \mathcal{W}$. Note that the objective function $\text{MSE}(\bm{w},\bm{\theta})$ is generally nonconcave in $\bm{\theta} = (\theta_0, \ldots, \theta_n)'$, as the squared bias $\left( \sum_{i=1}^n w_i \theta_i - \theta_0 \right)^2$ is convex in $\bm{\theta}$, whereas the variance $\sum_{i=1}^n w_i^2 \left( \frac{1}{4} - \theta_i^2 \right)$ is concave in $\bm{\theta}$.
Nevertheless, we show that this nonconvex ($n+1$)-dimensional optimization problem can be reduced to an optimization problem over a single parameter $\theta_0$, and is therefore computationally tractable.

Note first that $\Theta$ is centrosymmetric (i.e., $\bm{\theta}\in \Theta$ implies $-\bm{\theta}\in\Theta$) and that $\text{MSE}(\bm{w},\bm{\theta})=\text{MSE}(\bm{w},-\bm{\theta})$ for all $\bm{\theta}\in \Theta$.
Therefore, it suffices to consider maximizing the MSE over $\bm{\theta}\in\Theta$ such that $\theta_0\leq 0$.
In addition, the following lemma implies that it suffices to consider $\bm{\theta} = (\theta_0, \ldots, \theta_n)'$ satisfying $\theta_i \geq \theta_0$ for all $i$.

\begin{lemma}\label{lem:max_prob_1}
Suppose that $\bm{w} \in \mathcal{W}$. If $\bm{\theta}$ satisfies $\theta_0 \leq 0$, there exists $\tilde{\bm{\theta}} \equiv (\tilde{\theta}_0, \tilde{\theta}_1, \ldots , \tilde{\theta}_n)' \in \Theta$ such that $\text{MSE}(\bm{w},\bm{\theta}) \leq \text{MSE}(\bm{w}, \tilde{\bm{\theta}})$ and $\tilde{\theta}_i \geq \tilde{\theta}_0$ for all $i$.
\end{lemma}

The proofs of all theoretical results in the main text are provided in Appendix \ref{sec:proofs}.
In the proof of Lemma \ref{lem:max_prob_1}, we show that $\tilde{\bm{\theta}} = (\theta_0, \theta_1 + 2\cdot \max\{0, \theta_0 - \theta_1\}, \ldots, \theta_n + 2\cdot \max\{0, \theta_0 - \theta_n\})'$ satisfies $\text{MSE}(\bm{w},\bm{\theta}) \leq \text{MSE}(\bm{w}, \tilde{\bm{\theta}})$.
We construct $\tilde{\bm{\theta}}$ by increasing $\theta_i$ to $\theta_0+\theta_0-\theta_i$ for each $i$ if $\theta_i$ is less than $\theta_0$. The new value is larger than $\theta_0$ by $\theta_0-\theta_i$.
The change from $\bm{\theta}$ to $\tilde{\bm{\theta}}$ increases the variance while maintaining the Lipschitz constraint.
Furthermore, we can show that this change results in a positive bias whose absolute value is larger than that of the bias at the original $\bm{\theta}$.

In view of Lemma \ref{lem:max_prob_1}, we may consider the maximization of the MSE over $\bm{\theta}\in\Theta$ subject to the following restriction:
\begin{equation}
    \text{$\theta_0 \leq 0$ and $\theta_i \geq \theta_0$ for all $i$.} \label{theta_condition}
\end{equation}
By calculating the derivatives of the MSE, we can show that $\text{MSE}(\bm{w},\bm{\theta})$ is nondecreasing in $\theta_j$ under (\ref{theta_condition}).
To see this, note that
\begin{eqnarray}
    \frac{\partial}{\partial \theta_j} \text{MSE}(\bm{w},\bm{\theta}) &=& 2 w_j \left( \sum_{i \neq j} w_i \theta_i - \theta_0 \right), \hspace{0.3in} j = 1, \ldots , n. \label{derivative_j}
\end{eqnarray}
Because we have $\sum_{i \neq j} w_i \theta_i - \theta_0 \geq \left( \sum_{i \neq j} w_i  - 1 \right) \theta_0 \geq 0$ for all $\bm{w} \in \mathcal{W}$ under (\ref{theta_condition}), it follows from (\ref{derivative_j}) that $\text{MSE}(\bm{w},\bm{\theta})$ is nondecreasing in $\theta_j$ under (\ref{theta_condition}). This monotonicity of the MSE implies that $\text{MSE}(\bm{w},(\theta_0,\theta_1,\ldots,\theta_n)')$ is maximized by setting $\theta_1,\ldots,\theta_n$ to the largest possible values that satisfy the Lipschitz constraint for each fixed value of $\theta_0$.

\begin{figure}[h]
    \centering
    \includegraphics[width=12cm]{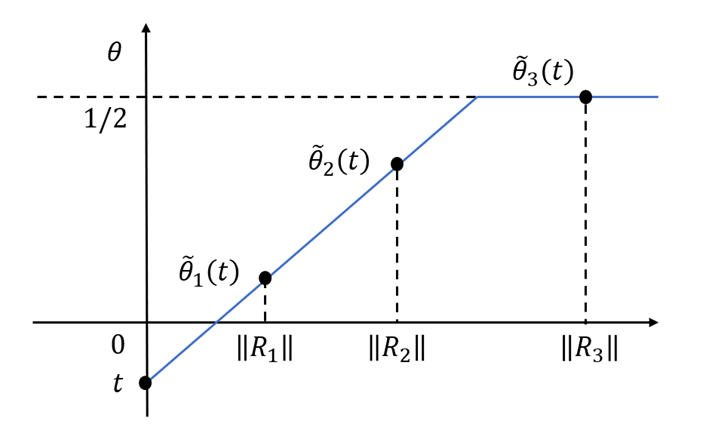}
    \caption{An illustration of the shape of $\tilde{\bm{\theta}}(t)$. The blue solid line denotes a function $r \mapsto \min \{t + C r, \frac{1}{2}\}$.}
    \label{fig:tilde_theta}
\end{figure}

Formally, we define the largest possible values of $\theta_0, \theta_1,\ldots, \theta_n$ given $\theta_0 = t$ as
\[
\tilde{\bm{\theta}}(t) \equiv \left( \tilde{\theta}_0(t), \tilde{\theta}_1(t), \ldots, \tilde{\theta}_n(t) \right)' \ \text{and} \ \tilde{\theta}_i(t) \equiv \min \{t + C \|R_i\|, 1/2 \} \ \text{for $i = 0,1, \ldots, n$,}
\]
as illustrated in Figure \ref{fig:tilde_theta}.
For any $\bm{\theta}=(\theta_0,\theta_1,\ldots,\theta_n)'\in\Theta$, we have $\theta_0=\tilde\theta_0(\theta_0)$ and $\theta_i \leq \tilde{\theta}_i(\theta_0)$ for $i=1,\ldots,n$. From (\ref{derivative_j}), if $\bm{\theta}\in\Theta$ satisfies (\ref{theta_condition}), we can increase the MSE by increasing $\theta_i$ to $\tilde{\theta}_i(\theta_0)$:
\begin{equation*}
    \text{MSE}(\bm{w},\bm{\theta}) \ \leq \ \text{MSE}(\bm{w},\tilde{\bm{\theta}}(\theta_0)) \ \ \text{for all $\bm{w} \in \mathcal{W}$.} 
\end{equation*}
Also, $\tilde{\bm{\theta}}(t) \in \Theta$ for any $t \in [-1/2,1/2]$ because $\tilde{\bm{\theta}}(t)$ satisfies $\tilde{\bm{\theta}}(t) \in [-1/2, 1/2]^{n+1}$ and
\begin{equation*}
    \left| \tilde{\theta}_i(t) - \tilde{\theta}_j(t) \right| \ \leq \ C \left| \|R_i\| - \|R_j\|  \right| \ \leq \  C \| R_i - R_j \|,
\end{equation*}
where the second inequality follows from the reverse triangle inequality.
Hence, we can reduce the ($n+1$)-dimensional maximization problem in (\ref{minimax_problem_2}) to a one-dimensional problem with a single parameter $\theta_0$, as in the following theorem.
\begin{theorem}\label{thm:max_problem}
Suppose that $\sum_{i=1}^n w_i \leq 1$ and $w_i \geq 0$ for all $i$. Then, we have
\begin{equation}
    \max_{\bm{\theta} \in \Theta} \text{MSE}(\bm{w},\bm{\theta}) \ = \ \max_{\theta_0 \in [-1/2, 0]} \text{MSE}(\bm{w},\tilde{\bm{\theta}}(\theta_0)). \label{max_problem}
\end{equation}
\end{theorem}

\subsection{The minimax linear shrinkage estimator}\label{sec:minimax_est}
Next, we derive a weight vector that minimizes the maximum MSE. The following two lemmas show that the optimal weight vector is nonincreasing and that the $i$th element of the optimal weight vector is zero if $R_i$ is sufficiently far away from $R_0$.

\begin{lemma}\label{lem:monotone}
We obtain
\begin{equation*}
    \min_{\bm{w} \in \mathcal{W}} \max_{\bm{\theta} \in \Theta} \text{MSE}(\bm{w},\bm{\theta}) \ = \ \min_{\bm{w} \in \mathcal{W}_0} \max_{\bm{\theta} \in \Theta} \text{MSE}(\bm{w},\bm{\theta}),
\end{equation*}
where $\mathcal{W}_0 \equiv \left\{ \bm{w} \in \mathcal{W} : w_1 \geq w_2 \geq \cdots \geq w_n \right\}$.
\end{lemma}

\begin{lemma}\label{lem:zero_weight}
We obtain
\begin{equation*}
    \min_{\bm{w} \in \mathcal{W}} \max_{\bm{\theta} \in \Theta} \text{MSE}(\bm{w},\bm{\theta}) \ = \ \min_{\bm{w} \in \mathcal{W}_1} \max_{\bm{\theta} \in \Theta} \text{MSE}(\bm{w},\bm{\theta}),
\end{equation*}
where $\mathcal{W}_1 \equiv \left\{ \bm{w} \in \mathcal{W}_0 : w_i = 0 \ \text{if $C\|R_i\| \geq 1/2$} \right\}$.
\end{lemma}

Lemma \ref{lem:monotone} shows that the optimal weight vector must be nonincreasing. In the proof of Lemma \ref{lem:monotone}, we show that if $\bm{w} \in \mathcal{W}$ satisfies $w_j < w_{j+1}$, then the maximum MSE can be reduced by swapping the positions of $w_j$ and $w_{j+1}$. By repeating this procedure until the weight vector becomes monotone, we can obtain $\tilde{\bm{w}} \in \mathcal{W}_0$ such that $\max_{\bm{\theta} \in \Theta} \text{MSE}(\tilde{\bm{w}},\bm{\theta}) \leq \max_{\bm{\theta} \in \Theta} \text{MSE}(\bm{w},\bm{\theta})$. Lemma \ref{lem:zero_weight} shows that the $i$th element of the optimal weight vector is zero if $C\|R_i\| \geq 1/2$. By calculating the derivative of $\text{MSE}(\bm{w}, \tilde{\bm{\theta}}(\theta_0))$ with respect to $w_{i}$, we can show that $\text{MSE}(\bm{w}, \tilde{\bm{\theta}}(\theta_0))$ is nondecreasing in $w_i$ when $C\|R_i\| \geq 1/2$, and hence, setting $w_i = 0$ is optimal.

These two lemmas allow us to restrict our search space for the optimal $\bm{w}$ to nonincreasing vectors that place no weight on the observations with $C\|R_i\| \geq 1/2$.
For notational simplicity, without loss of generality, we assume that our sample includes observations with $C \|R_i\| < 1/2$ only, so that $\mathcal{W}_0=\mathcal{W}_1$.
Theorem \ref{thm:max_problem} and Lemma \ref{lem:monotone} then imply that the minimax problem is reduced to
\begin{align}
    \min_{\bm{w} \in \mathcal{W}_0} \max_{\theta_0 \in [-1/2,0]}\text{MSE}(\bm{w},\tilde{\bm{\theta}}(\theta_0)),
    \label{minimax_problem_simple}
\end{align}
where
\begin{eqnarray*}
    \text{MSE}(\bm{w},\tilde{\bm{\theta}}(\theta_0)) &=& \left\{ \sum_{i=1}^n w_i (\theta_0 + C \|R_i\|) - \theta_0 \right\}^2 + \sum_{i=1}^n w_i^2 \left\{ \frac{1}{4} - (\theta_0 + C \|R_i\|)^2 \right\}.
\end{eqnarray*}

We now present a method for numerically solving the minimax problem (\ref{minimax_problem_simple}). We define $g(\bm{w};\theta_0) \equiv \text{MSE}(\bm{w},\tilde{\bm{\theta}}(\theta_0))$ and $\overline{g}(\bm{w}) \equiv \max_{\theta_0 \in [-1/2,0]} g(\bm{w};\theta_0)$. Because both $\bm{w} \mapsto \left( \sum_{i=1}^n w_i \theta_i - \theta_0 \right)^2$ and $\bm{w} \mapsto \sum_{i=1}^n w_i^2 \left( \frac{1}{4} - \theta_i^2 \right)$ are convex for any $\bm{\theta} \in \Theta$, $g(\bm{w};\theta_0)$ is also convex with respect to $\bm{w}$ for any $\theta_0 \in [-1/2, 0]$. As the maximum of convex functions is also convex, $\overline{g}(\bm{w})$ is a convex function. Therefore, the minimax problem (\ref{minimax_problem_simple}) becomes the following convex optimization problem with linear constraints:
\begin{equation}
    \min \overline{g}(\bm{w}) \hspace{0.2in} \text{subject to $\sum_{i=1}^n w_i \leq 1$ and $w_1 \geq w_2 \geq \cdots \geq w_n \geq 0$.} \nonumber
\end{equation}
We use nonlinear optimization via the augmented Lagrange method \citep{package_rsolnp,Ye_1987} in the implementation in simulations and applications. 

\begin{remark}
In the implementation, we compute $\bar{g}(\bm{w})$ by conducting a scalar-valued grid search to optimize $\theta_0$. Nevertheless, $g(\bm{w}; \theta_0)$ is a quadratic function in $\theta_0$ and $\overline{g}(\bm{w})$ has a closed-form expression. Let $u(\bm{w}) \equiv \sum_{i=1}^n w_i$ and $k(\bm{w}) \equiv \sum_{i=1}^n w_i \|R_i\|$. Then, $g(\bm{w};\theta_0)$ can be written as
\begin{eqnarray*}
g(\bm{w};\theta_0) &=& \left\{ C k(\bm{w}) - (1-u(\bm{w})) \theta_0 \right\}^2 + \sum_{i=1}^n w_i^2 \left( - \theta_0^2 - 2 C\|R_i\| \theta_0 + \frac{1}{4} - C^2 \|R_i\|^2 \right) \\
&=& \left\{ (1-u(\bm{w}))^2 - \sum_{i=1}^n w_i^2 \right\} \theta_0^2 - 2C \left\{ k(\bm{w})(1-u(\bm{w})) + \sum_{i=1}^n w_i^2 \|R_i\| \right\} \theta_0 \\
& & \hspace{1.0in} + C^2 k(\bm{w})^2 + \sum_{i=1}^n w_i \left( \frac{1}{4} - C^2 \|R_i\|^2 \right),
\end{eqnarray*}
where $k(\bm{w})(1-u(\bm{w})) + \sum_{i=1}^n w_i^2 \|R_i\| = \sum_{i=1}^n w_i\|R_i\|(1-\sum_{j\neq i}w_j) \geq 0$ for any $\bm{w} \in \mathcal{W}$. Hence, if $(1-u(\bm{w}))^2 - \sum_{i=1}^n w_i^2 \geq 0$, then $g(\bm{w};\theta_0)$ is maximized at $\theta_0 = -1/2$. If $(1-u(\bm{w}))^2 - \sum_{i=1}^n w_i^2 < 0$, $g(\bm{w};\theta_0)$ is maximized at $\theta_0 = \max\{-1/2, \beta(\bm{w})\}$, where
\begin{equation*}
\beta(\bm{w}) \ \equiv \ \frac{ C \left\{ k(\bm{w})(1-u(\bm{w})) + \sum_{i=1}^n w_i^2 \|R_i\| \right\}}{(1-u(\bm{w}))^2 - \sum_{i=1}^n w_i^2}.
\end{equation*}
Combining these two cases, $g(\bm{w};\theta_0)$ is maximized at $\theta_0 = -1/2$ if and only if the following inequality holds:
\begin{eqnarray}
    & & C \left\{ k(\bm{w})(1-u(\bm{w})) + \sum_{i=1}^n w_i^2 \|R_i\| \right\} + \frac{1}{2}\left\{ (1-u(\bm{w}))^2 - \sum_{i=1}^n w_i^2 \right\} \ \geq \ 0. \label{interior_condition}
\end{eqnarray}
If (\ref{interior_condition}) does not hold, then $g(\bm{w}; \theta_0)$ is maximized at $\theta_0 = \beta(\bm{w})$. As a result, we obtain
\begin{equation}
    \overline{g}(\bm{w}) \ = \ \begin{cases}
        g\left(\bm{w}; -\frac{1}{2} \right), & \text{if (\ref{interior_condition}) holds} \\
        \psi ( \bm{w}), & \text{if (\ref{interior_condition}) does not hold}
    \end{cases},\nonumber
\end{equation}
where $\psi (\bm{w}) \equiv C^2 k(\bm{w})^2 + \sum_{i=1}^n w_i^2 (1/4 - C^2 \|R_i\|^2) - \frac{ C^2 \left\{ k(\bm{w})(1-u(\bm{w})) +\sum_{i=1}^n w_i^2 \|R_i\| \right\}^2}{(1-u(\bm{w}))^2 -  \sum_{i=1}^n w_i^2}$.
\end{remark}

\begin{remark}\label{rem:ATE}
In this remark, we return to the original setup introduced in Section \ref{sec:settings}, where we observe both the treated sample $\{ Y_{i,+}, R_{i,+} \}_{i=1}^{n_{+}}$ and the untreated sample $\{ Y_{i,-}, R_{i,-} \}_{i=1}^{n_{-}}$.
We consider the estimation of $f(1,R_0)-f(0,R_0)$, which can be interpreted as the conditional average treatment effect (ATE) at the cutoff $R_0$. We can estimate the ATE by separately constructing the minimax linear shrinkage estimators for $f(1,R_0)$ and $f(0,R_0)$ using the treated and untreated samples, respectively.
Specifically, let $\hat{\bm{w}}_{+}$ and $\hat{\bm{w}}_{-}$ be the optimal weights minimizing the maximum MSEs among the linear shrinkage estimators for $f(1,R_0)$ and $f(0,R_0)$.
We can then estimate the conditional ATE $f(1,R_0)-f(0,R_0)$ using the following estimator:
\begin{align}
\sum_{i=1}^{n_{+}}\hat{w}_{i,+} \left( Y_{i,+} - \frac{1}{2} \right) - \sum_{i=1}^{n_{-}}\hat{w}_{i,-} \left( Y_{i,-} - \frac{1}{2} \right).\label{ATE_estimator}
\end{align}
Note that this estimator does not minimize the maximum MSE for the ATE estimation among the estimators that take the difference between two linear shrinkage estimators; the MSE for $f(1,R_0)-f(0,R_0)$ is not equal to the sum of the MSEs for $f(1,R_0)$ and $f(0,R_0)$. In Appendix \ref{sec:ATE}, we consider the joint optimization of the weight vectors ${\bm{w}}_{+}$ and ${\bm{w}}_{-}$ and obtain results similar to those of Theorem \ref{thm:max_problem} and Lemmas \ref{lem:monotone} and \ref{lem:zero_weight}.
Specifically, the maximum MSE for the ATE can be calculated by simultaneously optimizing two parameters, $f(1,R_0)$ and $f(0,R_0)$. Moreover, the optimal weight vectors are nonincreasing in the distance from the cutoff $R_0$.
Although joint optimization of ${\bm{w}}_{+}$ and ${\bm{w}}_{-}$ is thus possible, it may require a two-dimensional grid search to calculate the worst-case MSE and potentially introduce instability in the resulting ATE estimates. Therefore, we use the separately optimized weights $\hat{\bm{w}}_{+}$ and $\hat{\bm{w}}_{-}$ to estimate the ATE in our simulations and empirical application.
\end{remark}

\section{Comparison with Gaussian-motivated estimators}\label{sec:gauss}

Many existing studies have considered minimax estimation problems for unbounded outcomes with known variances, primarily motivated by the Gaussian model. We compare our proposed estimator with a Gaussian-motivated minimax linear estimator when the underlying data-generating process is a binary-outcome model.

As a Gaussian-motivated estimator, we consider the minimax linear estimator for an unbounded space of mean vectors with known variances under a smoothness restriction, following the existing minimax analysis in RD designs \citep{Armstrong.Kolesar2018,Imbens.Wager2019}. Note that if the outcome $Y_i$ is normally distributed, that is, $Y_i \sim N(p_i, \sigma_i^2)$, the MSE of a linear estimator $\hat{p}_0(\bm{w}) = \frac{1}{2} + \sum_{i=1}^n w_i \left( Y_i-\frac{1}{2} \right)$ with $\bm{w}\in \mathbb{R}^{n}$ is given by
\[
E\left[ (\hat{p}_0(\bm{w}) - p_0)^2 \right] \ = \ \left\{ \frac{1}{2} + \sum_{i=1}^n w_i \left( p_i - \frac{1}{2} \right) - p_0 \right\}^2 + \sum_{i=1}^n w_i^2 \sigma_i^2.
\]
Letting $\theta_i  = p_i - 1/2$, the MSE can be written as follows:
\[
\left( \sum_{i=1}^n w_i \theta_i - \theta_0 \right)^2 + \sum_{i=1}^n w_i^2 \sigma_i^2.
\]
As a smoothness restriction, we impose the Lipschitz constraint, leading to the following parameter space:
\[
\Theta_g \ \equiv \ \left\{ \bm{\theta} \in \mathbb{R}^{n+1} : |\theta_i - \theta_j| \leq C \| R_i-R_j \| \ \text{for all $i$ and $j$} \right\}.
\]
The minimax linear estimator is a solution of the following problem:
\begin{equation}
\min_{\bm{w}\in\mathbb{R}^n} \max_{\bm{\theta} \in \Theta_g} \left\{ \left( \sum_{i=1}^n w_i \theta_i - \theta_0 \right)^2 + \sum_{i=1}^n w_i^2 \sigma_i^2  \right\}. \label{minimax_normal}
\end{equation}

We refer to the linear estimator that solves (\ref{minimax_normal}) as the Gaussian estimator.\footnote{Note that this estimator is a minimax linear estimator without normality of $Y_i$ as long as the variance is known and the parameter space is $\Theta_g$.
The normality of $Y_i$ is exploited for finite-sample valid inference based on a linear estimator.}
The above minimax problem (\ref{minimax_normal}) differs from the original binary-outcome problem (\ref{minimax_problem_2}) in three respects.
First, the minimum in (\ref{minimax_normal}) is considered among all linear estimators, including those with negative weights. Second, the parameter space in (\ref{minimax_normal}) is unbounded. Finally, but most importantly, the variance in (\ref{minimax_normal}) does not depend on the parameter $\bm{\theta}$; hence, the maximum MSE is attained at the parameter values that maximize the squared bias.

In Appendix \ref{sec:gauss_weights}, we derive the form of the optimal weights that solve the minimax problem (\ref{minimax_normal}) by applying the results of \cite{donoho1994} to our Gaussian setting.
We show that the optimal weights satisfy $\sum_{i=1}^n w_i = 1$ and $w_i \geq 0$ for all $i$. Hence, the minimax problem (\ref{minimax_normal}) can be solved by minimizing the maximum MSE over $\mathcal{W}$. Specifically, the Gaussian estimator is obtained by solving the following quadratic program:
\[
\min_{\bm{w}} \left\{ C^2 \left( \sum_{i=1}^n w_i \|R_i\| \right)^2 + \sum_{i=1}^n w_i^2 \sigma_i^2 \right\} \ \ \text{s.t.} \ \ \sum_{i=1}^n w_i = 1 \ \text{and} \ w_i \geq 0 \ \text{for all $i$},
\]
where $C^2\left( \sum_{i=1}^n w_i \|R_i\| \right)^2$ is the maximum squared bias of the estimator $\hat{p}_0(\bm{w})$ with $\sum_{i=1}^n w_i = 1$ over $\Theta_g$.

\subsection{Theoretical comparisons}\label{sec:compare_gauss}

We compare the maximum MSE of the proposed estimator with that of the Gaussian estimator in a setting in which the true model is the binary-outcome model described in Section \ref{sec:main}. The Gaussian estimator requires the specification of variance. In the following, we focus on the Gaussian estimator with $\sigma_1^2 = \cdots = \sigma_n^2 = 1/4$ because the variance of a binary variable is less than or equal to $1/4$. Define 
\[
\hat{\bm{w}} \in \mathrm{arg} \min_{\bm{w} \in \mathcal{W}} \max_{\bm{\theta} \in \Theta} \mathrm{MSE}(\bm{w},\bm{\theta}) \ \ \text{and} \ \ \tilde{\bm{w}} \in \mathrm{arg} \min_{\bm{w} \in \mathcal{W}} \max_{\bm{\theta} \in \Theta_g} \mathrm{MSE}_g(\bm{w},\bm{\theta}),
\]
where
\begin{eqnarray*}
\text{MSE}(\bm{w},\bm{\theta}) &=& \left( \sum_{i=1}^n w_i \theta_i - \theta_0 \right)^2 + \sum_{i=1}^n w_i^2 \left( \frac{1}{4} - \theta_i^2 \right), \\
\text{MSE}_g(\bm{w},\bm{\theta}) &=& \left( \sum_{i=1}^n w_i \theta_i - \theta_0 \right)^2 + \frac{1}{4} \sum_{i=1}^n w_i^2.
\end{eqnarray*}
Then, $\hat{p}_0(\hat{\bm{w}})$ is the minimax linear shrinkage estimator when $Y_i$ is binary, and $\hat{p}_0(\tilde{\bm{w}})$ is the minimax linear estimator when $Y_i \sim N(p_i,1/4)$. The following lemma compares the maximum MSEs of $\hat{p}_0(\hat{\bm{w}})$ and $\hat{p}_0(\tilde{\bm{w}})$ when $Y_i$ is binary and the parameter space is bounded.

\begin{lemma} \label{lem:compare_gauss}
If $\hat{u} \equiv \sum_{i=1}^n \hat{w}_i > 0$, then we obtain
\begin{equation}
1 \ \leq \ \frac{\max_{\bm{\theta} \in \Theta} \text{MSE}(\tilde{\bm{w}},\bm{\theta})}{\max_{\bm{\theta} \in \Theta} \text{MSE}(\hat{\bm{w}},\bm{\theta})} \ \leq \ \hat{u}^{-2} \left( 1 + \frac{C^2 \sum_{i=1}^n \hat{w}_i^2 \|R_i\|^2}{\frac{1}{4} \sum_{i=1}^n \hat{w}_i^2 } \right). \label{compare_gauss}
\end{equation}
In addition, the upper bound of (\ref{compare_gauss}) is bounded above by $2 \hat{u}^{-2}$.
\end{lemma}

Lemma \ref{lem:compare_gauss} provides lower and upper bounds of the ratio of the maximum MSEs. As $\hat{\bm{w}}$ minimizes $\max_{\bm{\theta} \in \Theta} \mathrm{MSE}(\bm{w},\bm{\theta})$ over $\mathcal{W}$, the lower bound is trivial. In the proof of Lemma \ref{lem:compare_gauss}, we derive the upper bound based on an upper bound on the numerator and a lower bound on the denominator.

Although the finite-sample bounds in Lemma \ref{lem:compare_gauss} may be loose, sharper bounds can be obtained if we consider an asymptotic setting in which the sample size increases.
In the following, we consider a triangular array $\{(R_{n,1},\ldots,R_{n,n})\}_{n\in \mathbb{N}}$, where $(R_{n,1},\ldots,R_{n,n})$ is a deterministic vector that collects the values of the running variable when the sample size is $n$.
We fix the value of the Lipschitz constant $C$ as $n$ varies.
In this asymptotic regime, we show that under mild conditions, the convergence rate of $\hat{p}_0(\hat{\bm{w}})$ is $O_p(n^{-1/3})$ and the ratio of the maximum MSEs of $\hat{p}_0(\hat{\bm{w}})$ and $\hat{p}_0(\tilde{\bm{w}})$ converges to one as $n \to \infty$.
For the brevity of the notation, we suppress the first index $n$ of $(R_{n,1},\ldots,R_{n,n})$ below.

To establish the asymptotic result, we consider a univariate running variable $R_i$ and assume that it is bounded and that the empirical distribution of $\|R_i\|$ is bounded above and below by linear functions.\footnote{The convergence holds under a weaker condition, which may be plausible for a multivariate running variable. See Remark \ref{rem:hat_u} for a discussion of the general case.}
\begin{assumption}\label{ass:asymptotic}
The running variables $\{R_1, \ldots, R_n\} \in \mathbb{R}$ satisfy the following conditions:
\begin{itemize}
\item[(i)] $0 \leq \|R_1\| \leq \ldots \leq \|R_n\| \leq 1$.
\item[(ii)] There exist constants $c_1 > c_0 > 0$ such that, for any sufficiently large $n \in \mathbb{N}$, $c_0 x - n^{-1/3} \leq F_n(x) \leq c_1 x  + n^{-1/3}$ for all $x \in [0,1]$, where $F_n(\cdot)$ is the empirical distribution of $\|R_i\|$ when the sample size is $n$, that is,
\[
F_n(x) \ \equiv \ \frac{1}{n} \sum_{i=1}^n 1\{\|R_i\| \leq x\}.
\]
\end{itemize}
\end{assumption}

Figure \ref{fig:empirical_dist} illustrates Assumption \ref{ass:asymptotic} (ii). For example, when $R_i = i/n$ for all $i=1, \ldots, n$, this assumption is satisfied for $0 < c_0 < 1 < c_1$. More generally, Assumption \ref{ass:asymptotic} (ii) requires the empirical distribution $F_n(x)$ to be bounded by a pair of linear functions. 

\begin{figure}[H]
    \centering
    \includegraphics[width=12cm]{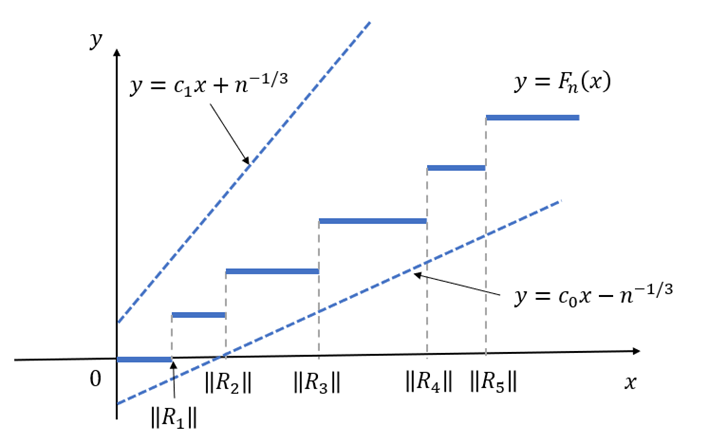}
    \caption{The blue solid line denotes $y = F_n(x)$. The blue dotted lines denote functions $y = c_1 x+n^{-1/3}$ and $y=c_0 x - n^{-1/3}$.}
    \label{fig:empirical_dist}
\end{figure}

\begin{theorem}\label{thm:asymptotic}
Under Assumption \ref{ass:asymptotic}, we obtain $\max_{\bm{\theta} \in \Theta} \text{MSE}(\hat{\bm{w}},\bm{\theta}) = O(n^{-2/3})$ and
\begin{equation}
\frac{\max_{\bm{\theta} \in \Theta} \text{MSE}(\tilde{\bm{w}},\bm{\theta})}{\max_{\bm{\theta} \in \Theta} \text{MSE}(\hat{\bm{w}},\bm{\theta})} \ \to \ 1 . \nonumber
\end{equation}
\end{theorem}

Theorem \ref{thm:asymptotic} shows that the convergence rate of $\hat{p}_0(\hat{\bm{w}})$ is $O_p(n^{-1/3})$. This convergence rate is the same as that of standard nonparametric estimators under the Lipschitz constraint in univariate RD designs.
Theorem \ref{thm:asymptotic} also shows that the maximum MSE of $\hat{p}_0(\tilde{\bm{w}})$ is asymptotically identical to that of $\hat{p}_0(\hat{\bm{w}})$. The Gaussian estimator $\hat{p}_0(\tilde{\bm{w}})$ minimizes the maximum MSE when $Y_i \sim N(p_i, 1/4)$ and the parameter space is unbounded. This result implies that the Gaussian estimator is asymptotically optimal in terms of the maximum MSE for a particular sequence of distributions of the running variable, even when outcomes are binary.

\begin{remark}\label{rem:hat_u}
The convergence of $\max_{\bm{\theta} \in \Theta}\mathrm{MSE}(\hat{\bm{w}},\bm{\theta})$ and $\max_{\bm{\theta} \in \Theta}\mathrm{MSE}(\tilde{\bm{w}},\bm{\theta})$ to zero requires weaker conditions than Assumption \ref{ass:asymptotic}. Specifically, the convergence may hold for multidimensional $R_i$. For example, suppose that for any $\epsilon > 0$, the sample size satisfying $\|R_i\| \leq \epsilon$ goes to infinity as $n \to \infty$. That is, letting $N(\epsilon) \equiv \max\{i \in \{1,\ldots,n\} : \|R_i\| \leq \epsilon \}$, $N(\epsilon) \to \infty$ holds for all $\epsilon > 0$. This condition is weaker than Assumption \ref{ass:asymptotic} (ii) and is plausible in a multidimensional case as well.
To show the convergence of $\max_{\bm{\theta} \in \Theta}\mathrm{MSE}(\hat{\bm{w}},\bm{\theta})$ and $\max_{\bm{\theta} \in \Theta}\mathrm{MSE}(\tilde{\bm{w}},\bm{\theta})$ under this condition, we use the following relationship:
\begin{equation*}
    \max_{\bm{\theta} \in \Theta} \mathrm{MSE}(\hat{\bm{w}},\bm{\theta}) \ \leq \ \max_{\bm{\theta} \in \Theta} \mathrm{MSE}(\tilde{\bm{w}},\bm{\theta}) \ \leq \ \max_{\bm{\theta} \in \Theta_g} \mathrm{MSE}_g(\tilde{\bm{w}},\bm{\theta}),
\end{equation*}
where the second inequality holds since $\mathrm{MSE}(\bm{w},\bm{\theta}) \leq \mathrm{MSE}_g(\bm{w},\bm{\theta})$ and $\Theta \subset \Theta_g$.
For any $\epsilon > 0$, we obtain
\begin{eqnarray*}
    \max_{\bm{\theta} \in \Theta_g} \mathrm{MSE}_g(\tilde{\bm{w}},\bm{\theta}) &=& 
    \min_{\bm{w} \in \mathcal{W}: \sum_{i=1}^n w_i = 1} \left\{ C^2 \left( \sum_{i=1}^n w_i \|R_i\| \right)^2 + \frac{1}{4} \sum_{i=1}^n w_i^2 \right\} \\
    & \leq & C^2 \left( \frac{1}{N(\epsilon)} \sum_{i=1}^{N(\epsilon)} \|R_i\| \right)^2 + \frac{1}{4 N(\epsilon)} \ \leq \  C^2 \epsilon^2 + \frac{1}{4 N(\epsilon)} \ \to \ C^2 \epsilon^2,
\end{eqnarray*}
where the first inequality is obtained by setting $\bm{w}=\left( \underbrace{\frac{1}{N(\epsilon)}, \ldots , \frac{1}{N(\epsilon)}}_{N(\epsilon)}, 0, \ldots , 0 \right)'$ and the convergence follows from the assumption that  $N(\epsilon) \to \infty$. Hence, $\max_{\bm{\theta} \in \Theta_g} \mathrm{MSE}_g(\tilde{\bm{w}},\bm{\theta}) \rightarrow 0$ since $\epsilon$ can be arbitrarily small.
Therefore, $\max_{\bm{\theta} \in \Theta} \mathrm{MSE}(\hat{\bm{w}},\bm{\theta})$ and $\max_{\bm{\theta} \in \Theta} \mathrm{MSE}(\tilde{\bm{w}},\bm{\theta})$ also converge to zero.
\end{remark}

\begin{remark}\label{rem:hat_u_2}
The shrinkage factor $\hat{u}=\sum_{i=1}^n \hat{w}_i$ converges to one under mild conditions. Consequently, the upper bound $2\hat u^{-2}$ of the ratio of the maximum MSEs given in Lemma \ref{lem:compare_gauss} converges to $2$.
To see that $\hat{u}$ converges to one, we use the following relationship:
$\frac{1}{4}(1-\hat{u})^2 \leq \max_{\bm{\theta} \in \Theta} \mathrm{MSE}(\hat{\bm{w}},\bm{\theta})$, as shown in the proof of Theorem \ref{thm:asymptotic}.
From the discussion in Remark \ref{rem:hat_u}, under mild conditions, we have $\max_{\bm{\theta} \in \Theta} \mathrm{MSE}(\hat{\bm{w}},\bm{\theta})\rightarrow 0$, which implies the shrinkage factor $\hat{u}$ converges to one.
\end{remark}

\subsection{Numerical comparisons}\label{sec:numerical_gauss}

Although the efficiency gain from our estimator relative to the Gaussian estimator can be small in large samples, their behaviors are quite different in finite samples. We demonstrate the finite-sample comparisons of our estimator with the Gaussian estimator through numerical analyses. Figures \ref{fig:compare_gauss_50} and \ref{fig:compare_gauss_500} plot the weights $w_1,\ldots,w_n$ for samples of observations whose running-variable values are equally spaced between $0$ and $1$. Figure \ref{fig:compare_gauss_50} plots the weights of our estimator (\textit{rdbinary}) and the Gaussian estimator (\textit{gauss}) for a sample size of $50$ and four Lipschitz constant values. Figure \ref{fig:compare_gauss_500} shows the corresponding plots for a sample size of $500$. The weights of the Gaussian estimator are computed under the assumption that the variance is homoskedastic and $1/4$ for all units, as in Section \ref{sec:compare_gauss}. For a small sample size of $50$, our estimator exhibits moderate size of shrinkage, whereas the Gaussian estimator exhibits no shrinkage. For $C > 0$, the weights of the Gaussian estimator are triangular, while our estimator's weights display mild nonlinearity. Also, the Gaussian weights have thicker tails than ours. These differences in shape arise because the Gaussian estimator is constructed under homoskedasticity and the maximum possible variance of $1/4$, whereas our estimator optimizes the weights under potential heteroskedasticity.

\begin{figure}[H]
    \centering
    \includegraphics[width=0.8\hsize]{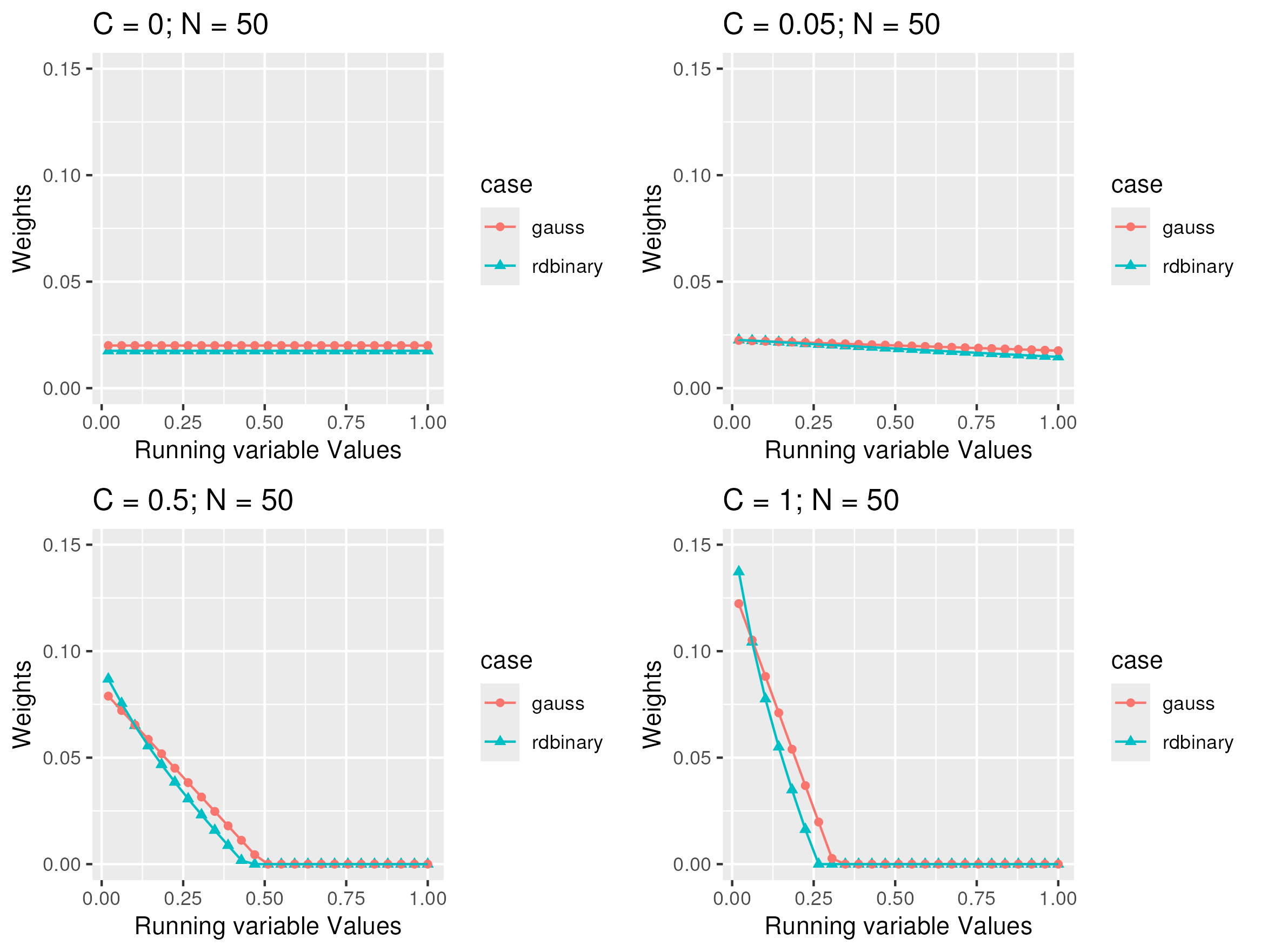}
    \caption{Comparison of weights for equally spaced grids (N = 50)}
    \label{fig:compare_gauss_50}
\end{figure}

\begin{figure}[H]
    \centering
    \includegraphics[width=0.8\hsize]{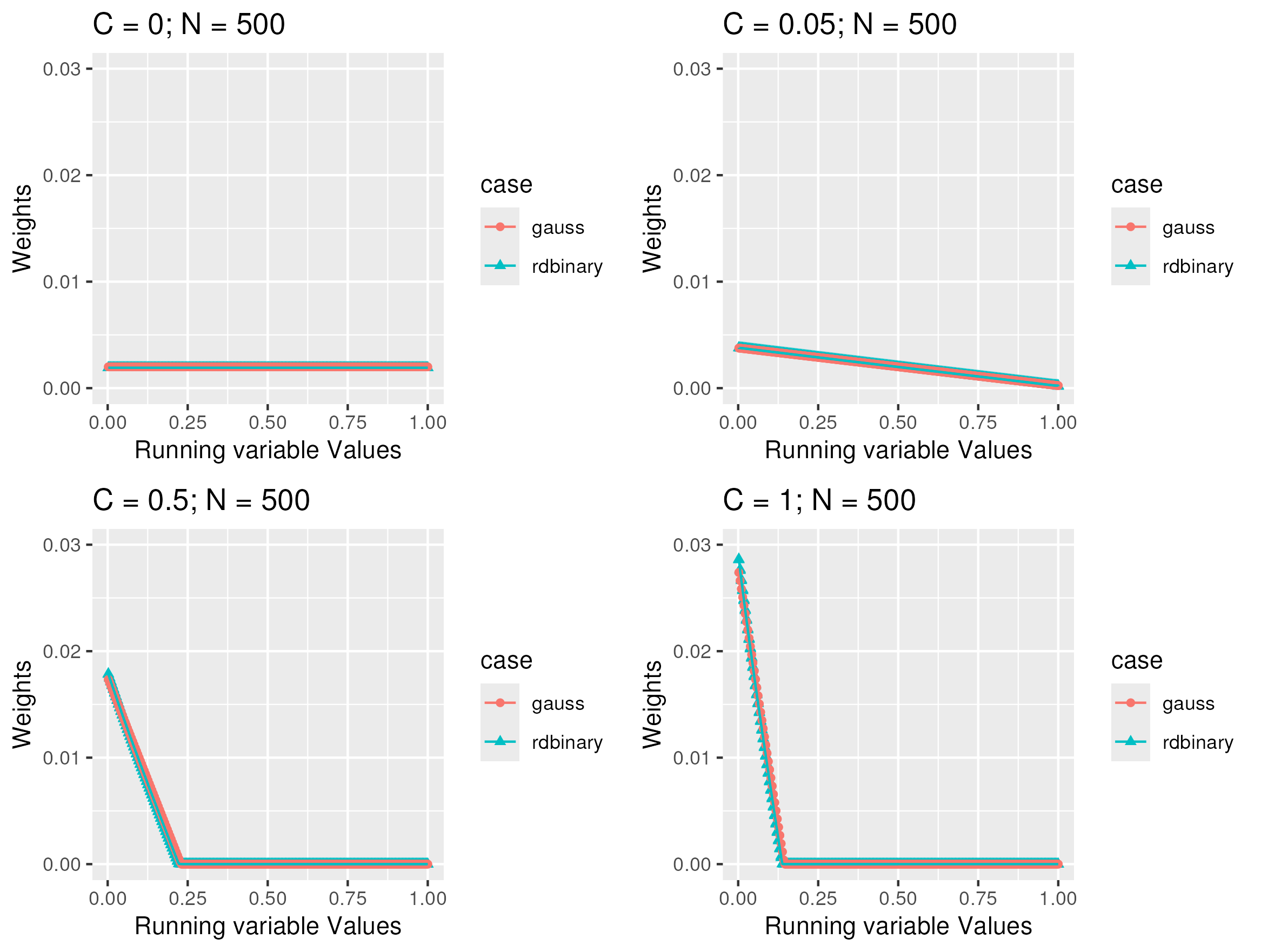}
    \caption{Comparison of weights for equally spaced grids (N = 500)}
    \label{fig:compare_gauss_500}
\end{figure}

By contrast, the two estimators appear almost equivalent for a sufficiently large sample size of $500$. The shape of our estimator's weights is still sharper than that of the Gaussian weights for $C = 1$; however, the differences between the two weights are negligible compared to the case with a small sample size of $50$. 

\begin{figure}[H]
    \centering
    \includegraphics[width=\hsize]{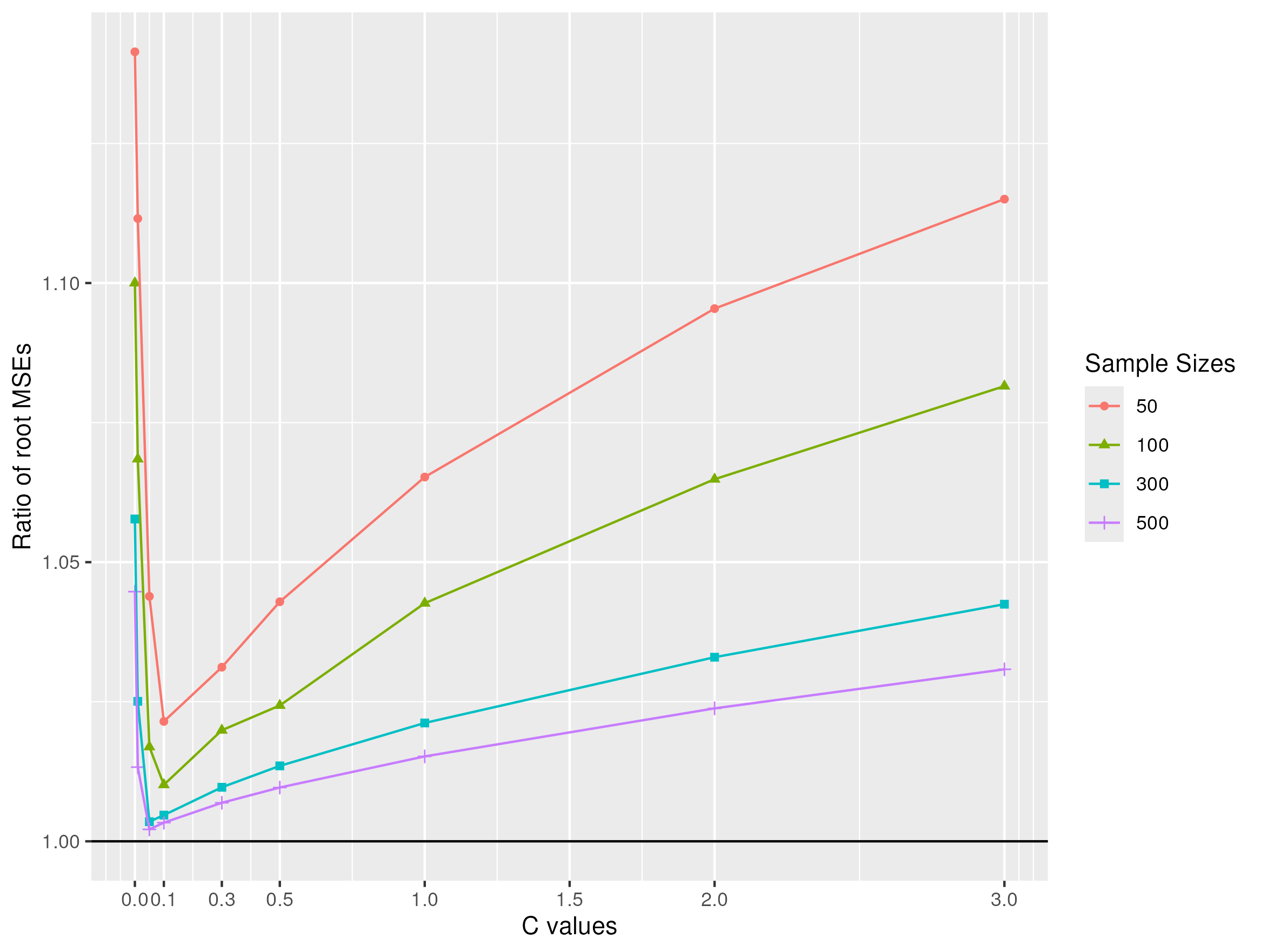}
    \caption{Ratio of maximum root MSE of Gaussian to that of rdbinary}
    \label{fig:compare_mse}
\end{figure}

More distinct differences appear in the maximum root MSEs in small samples. Figure \ref{fig:compare_mse} reports the ratio of the maximum root MSE of the Gaussian estimator to that of our estimator, calculated in the binary-outcome model. For a small sample size of $50$, the Gaussian estimator has $5\%$ to $20\%$ larger root MSEs than our estimator. Hence, our estimator achieves substantial improvement over the Gaussian estimator in small samples.

Nevertheless, these ratios shrink as the sample size increases and the gaps shrink below $5\%$ when $N = 500$. This property is consistent with the theoretical result that the ratio of the worst-case MSEs converges to one as the sample size increases. In summary, our estimator is substantially different from and superior to the Gaussian estimator in finite samples, whereas the two estimators behave similarly in large samples.

\section{Uniformly valid finite-sample inference}\label{sec:inference}

In this section, we return to the original setup introduced in Section \ref{sec:settings}, where we observe both the treated sample $\{ Y_{i,+}, R_{i,+} \}_{i=1}^{n_{+}}$ and the untreated sample $\{ Y_{i,-}, R_{i,-} \}_{i=1}^{n_{-}}$.
We propose an inference procedure with respect to $\tau \equiv f(1,R_0)-f(0,R_0)$ based on a given linear shrinkage estimator. Let $p_{i,+} \equiv f(1,R_{i,+})$, $p_{i,-} \equiv f(0,R_{i,-})$, and $R_{0,+}=R_{0,-}=0$, so that $Y_{i,+}$ and $Y_{i,-}$ follow Bernoulli distribution with parameters $p_{i,+}$ and $p_{i,-}$, respectively. Similar to the previous sections, we assume that $p_{i,+}$ and $p_{i,-}$ satisfy $\bm{p}_{+} \equiv (p_{0,+}, p_{1,+}, \ldots, p_{n_{+},+})' \in \mathcal{P}_{+}$ and $\bm{p}_{-} \equiv (p_{0,-}, p_{1,-}, \ldots, p_{n_{-},-})' \in \mathcal{P}_{-}$, where
\begin{eqnarray*}
    \mathcal{P}_{+} &\equiv & \left\{ \bm{p}_{+} \in [0,1]^{n_{+}+1}: |p_{i,+}-p_{j,+}| \leq C \| R_{i,+} - R_{j,+} \| \ \text{for all $i$ and $j$} \right\}, \\
    \mathcal{P}_{-} &\equiv & \left\{ \bm{p}_{-} \in [0,1]^{n_{-}+1}: |p_{i,-}-p_{j,-}| \leq C \| R_{i,-} - R_{j,-} \| \ \text{for all $i$ and $j$} \right\}.
\end{eqnarray*}
We propose an inference procedure of $\tau = p_{0,+}-p_{0,-}$ based on the estimator $\hat{\tau} \equiv \hat{p}_{0,+}(\bm{w}_{+}) - \hat{p}_{0,-}(\bm{w}_{-})$, where
\begin{eqnarray*}
    \hat{p}_{0,+}(\bm{w}_{+}) &\equiv &  \frac{1}{2} + \sum_{i=1}^{n_{+}}w_{i,+} \left( Y_{i,+} - \frac{1}{2} \right), \\
    \hat{p}_{0,-}(\bm{w}_{-}) &\equiv &  \frac{1}{2} + \sum_{i=1}^{n_{-}}w_{i,-} \left( Y_{i,-} - \frac{1}{2} \right).
\end{eqnarray*}
Our inference procedure is valid for any linear estimator with nonnegative weights (even if $\sum_{i=1}^{n_{+}}w_{i,+}>1$ or $\sum_{i=1}^{n_{-}}w_{i,-}>1$) when the outcome is binary. Hence, we can conduct an inference using the linear shrinkage estimator proposed in the previous sections. The following argument does not extend to general bounded outcomes. In Appendix \ref{sec:inf_bounded}, we consider an inference procedure for general bounded outcomes based on Hoeffding's inequality.

\subsection{One-sided test}\label{sec:one-sided}
We provide finite-sample valid confidence intervals by inverting tests that are uniformly valid over the Lipschitz class. We begin our analysis with a one-sided test. Using a uniformly valid one-sided test, we then construct a uniformly valid two-sided test and confidence interval.

Specifically, we consider a one-sided test for the following null and alternative hypotheses:
\begin{equation}
    H_0: \tau = \tau_0 \ \ \text{vs.} \ \ H_1: \tau > \tau_0. \label{hypothesis_one-sided}\nonumber
\end{equation}
We propose the following testing procedure based on the linear estimator $\hat{\tau}$:
\begin{equation}
    \hat{\tau} > \gamma \ \ \Rightarrow \ \ \text{reject $H_0$}, \label{one-sided_test}\nonumber
\end{equation}
where $\gamma$ is a critical value. The critical value $\gamma$ must satisfy $P_{\bm{p}}(\hat{\tau} - \tau_0 > \gamma) \leq \alpha$ for any parameter $\bm{p} \equiv (\bm{p}_{+}', \bm{p}_{-}')' \in \mathcal{P}_{\ast} \equiv \mathcal{P}_{+} \times \mathcal{P}_{-}$ satisfying $H_0$. Hence, we must choose the critical value $\gamma^{\ast}(\tau_0)$ such that
\begin{equation}
    \max_{\bm{p} \in \mathcal{P}(\tau_0)} P_{\bm{p}}\left( \hat{\tau} > \gamma^{\ast}(\tau_0) \right) \ \leq \ \alpha, \label{one-sided_prob}
\end{equation}
where $\mathcal{P}(\tau_0) \equiv \left\{ \bm{p} \in \mathcal{P}_{\ast} : p_{0,+} - p_{0,-} = \tau_0 \right\}$. The critical value $\gamma^*(\tau_0)$ provides a uniformly valid one-sided test for finite samples. 

To obtain an appropriate critical value, we must calculate $\max_{\bm{p} \in \mathcal{P}(\tau_0)} P(\hat{\tau} > \gamma)$. The following theorem shows that we can calculate $\max_{\bm{p} \in \mathcal{P}(\tau_0)} P(\hat{\tau} > \gamma)$ by optimizing a single parameter.

\begin{theorem}\label{thm:inference}
Define
\begin{eqnarray*}
    \tilde{\bm{p}}_{+}(p) &\equiv & \left( p, \min\{p+C\|R_{1,+}\|, 1 \}, \ldots , \min\{p+C\|R_{n_{+},+}\|, 1 \} \right)', \\
    \tilde{\bm{p}}_{-}(p) &\equiv & \left( p, \max\{p-C\|R_{1,-}\|, 0 \}, \ldots , \max\{p-C\|R_{n_{-},-}\|, 0 \} \right)', \\
    \tilde{\bm{p}}(p,\tau_0) &\equiv & \left( \tilde{\bm{p}}_{+}(p)', \tilde{\bm{p}}_{-}(p-\tau_0)' \right)'.
\end{eqnarray*}
If $w_{i,+} \geq 0$ and $w_{i,-} \geq 0$ for all $i$, we obtain
\begin{equation}
    \max_{\bm{p} \in \mathcal{P}(\tau_0)} P_{\bm{p}}(\hat{\tau} > \gamma) \ = \ \max_{p \in [\max \{0,\tau_0\}, \min \{1,1+\tau_0 \}]} P_{\tilde{\bm{p}}(p,\tau_0)} (\hat{\tau} > \gamma). \label{inference_max}
\end{equation}
\end{theorem}

Theorem \ref{thm:inference} is obtained using first-order stochastic dominance. Suppose that $(Y_1, \ldots, Y_n)' \in \{0,1\}^n$ and $(\tilde{Y}_1, \ldots, \tilde{Y}_n)' \in \{0,1\}^n$ follow $n$-dimensional independent Bernoulli distributions with parameters $\bm{p} \in \mathbb{R}^n$ and $\tilde{\bm{p}} \in \mathbb{R}^n$, respectively, and each element of $\bm{p}$ is smaller than or equal to that of $\tilde{\bm{p}}$. Then, if $w_i$ is nonnegative for all $i$, $\sum_{i=1}^n w_i \tilde{Y}_i$ has first-order stochastic dominance over $\sum_{i=1}^n w_i Y_i$. Hence, if we fix $p_{0,+}$ and $p_{0,-}$, then $P_{\bm{p}}(\hat{\tau} > \gamma)$ is maximized at $\bm{p} = \left( \tilde{\bm{p}}_{+}(p_{0,+})', \tilde{\bm{p}}_{-}(p_{0,-})' \right)'$, implying (\ref{inference_max}) holds.

From Theorem \ref{thm:inference}, we can obtain the critical value $\gamma^{\ast}(\tau_0)$ satisfying (\ref{one-sided_prob}) using the following algorithm:\vspace{0.05in}
\begin{enumerate}
    \item Fix $\gamma \in [-1,1]$ and $p \in [\max \{0,\tau_0\}, \min \{1,1+\tau_0 \}]$.
    \item Calculate the probability
    \begin{equation}
        P\left( \sum_{i=1}^{n_{+}} w_{i,+} (\tilde{Y}_{i,+}-1/2) - \sum_{i=1}^{n_{-}} w_{i,-} (\tilde{Y}_{i,-} - 1/2) \ > \ \gamma \right) \label{prob_reject}
    \end{equation}
    by drawing a large number of samples $\{\tilde{Y}_{1,+}, \ldots \tilde{Y}_{n_{+},+}, \tilde{Y}_{1,-}, \ldots, \tilde{Y}_{n_{-},-}\}$ from the $(n_{+}+n_{-})$-dimensional independent Bernoulli distribution with parameter $\tilde{\bm{p}} = \left( \tilde{\bm{p}}_{+}(p)', \tilde{\bm{p}}_{-}(p-\tau_0)' \right)'$.
    \item Maximize the probability (\ref{prob_reject}) with respect to $p \in [\max \{0,\tau_0\}, \min \{1,1+\tau_0 \}]$ numerically and define $\pi(\gamma)$ as the maximum of (\ref{prob_reject}).
    \item Derive $\gamma^{\ast}(\tau_0) = \text{arg} \min \{ \gamma : \pi(\gamma) \leq \alpha \}$.\vspace{0.05in}
\end{enumerate}

\begin{remark}\label{rem:one-sided_gamma}
As the critical value $\gamma^{\ast}(\tau_0)$ depends on the hypothesized value $\tau_0$, the critical value must be calculated for each hypothesized value.
We can show that the critical value $\gamma^{\ast}(\tau_0)$ increases with the hypothesized value $\tau_0$. Suppose that $-1 \leq \tau_0 \leq \tilde{\tau}_0 \leq 1$ and $p_{0,+} - p_{0,-} = \tau_0$. Then, there exist $\tilde{p}_{0,+}$ and $\tilde{p}_{0,-}$ such that $\tilde{p}_{0,-} \leq p_{0,-}$, $\tilde{p}_{0,+} \geq p_{0,+}$, and $\tilde{p}_{0,+}-\tilde{p}_{0,+} = \tilde{\tau}_0$. From an argument similar to that in the proof of Theorem \ref{thm:inference}, we obtain
\[
P_{\left( \tilde{\bm{p}}_{+}(p_{0,+}), \tilde{\bm{p}}_{-}(p_{0,-}) \right)}\left( \hat{\tau} > \gamma \right) \ \leq \ P_{\left( \tilde{\bm{p}}_{+}(\tilde{p}_{0,+}), \tilde{\bm{p}}_{-}(\tilde{p}_{0,-}) \right)} \left( \hat{\tau} > \gamma \right) \ \ \text{for any $\gamma$.}
\]
This result implies that $\gamma^{\ast}(\tau_0)$ is increasing in $\tau_0$. Hence, if the null hypothesis $H_0:\tau = \tilde{\tau}_0$ is rejected, then the null hypothesis $H_0:\tau = \tau_0$ must also be rejected for any $\tau_0<\tilde\tau_0$.
\end{remark}

\subsection{Two-sided test and confidence interval}\label{sec:two-sided}
Next, we construct a uniformly valid two-sided test and confidence interval using the one-sided test proposed in Section \ref{sec:one-sided}. We consider the following null and alternative hypotheses:
\begin{equation}
    H_0: \tau = \tau_0 \ \ \text{vs.} \ \ H_1: \tau \neq \tau_0. \label{hypothesis_two-sided}\nonumber
\end{equation}
Similar to the one-sided test, we propose the following testing procedure based on the linear estimator $\hat{\tau}$:
\begin{equation}
    \hat{\tau} \not\in [\gamma_l,\gamma_r]  \ \ \Rightarrow \ \ \text{reject $H_0$}, \label{two-sided_test}\nonumber
\end{equation}
where the critical values $\gamma_l$ and $\gamma_r$ must satisfy $P_{\bm{p}}(\hat{\tau} \not\in [\gamma_l,\gamma_r] ) \leq \alpha$ under $H_0$. Hence, we must choose the critical values $\gamma_l^{\ast}(\tau_0)$ and $\gamma_r^{\ast}(\tau_0)$ such that
\begin{equation}
    \max_{\bm{p} \in \mathcal{P}(\tau_0)} P_{\bm{p}}\left( \hat{\tau}  \not\in [\gamma_l^{\ast}(\tau_0),\gamma_r^{\ast}(\tau_0)] \right) \ \leq \ \alpha. \label{two-sided_prob}
\end{equation}

However, it is challenging to derive a simple expression for the maximum of the probability $P_{\bm{p}}\left( \hat{\tau}\not\in [\gamma_l,\gamma_r] \right)$, unlike for the one-sided testing. Therefore, we instead calculate an upper bound on the maximum of $P_{\bm{p}}\left( \hat{\tau} \not\in [\gamma_l,\gamma_r] \right)$:
\begin{eqnarray*}
    \max_{\bm{p} \in \mathcal{P}(\tau_0)} P_{\bm{p}}( \hat{\tau} \not\in [\gamma_l,\gamma_r] ) &=& \max_{\bm{p} \in \mathcal{P}(\tau_0)} \left\{ P_{\bm{p}}( \hat{\tau}  > \gamma_r) + P_{\bm{p}}( \hat{\tau}  < \gamma_l) \right\} \\
    & \leq & \max_{\bm{p} \in \mathcal{P}(\tau_0)} P_{\bm{p}}( \hat{\tau} > \gamma_r) + \max_{\bm{p} \in \mathcal{P}(\tau_0)} P_{\bm{p}}( \hat{\tau}  < \gamma_l)\\
    &=& \pi_r(\gamma_r) + \pi_l(\gamma_l),
\end{eqnarray*}
where $\pi_r(\gamma_r) \equiv \max_{\bm{p} \in \mathcal{P}(\tau_0)} P_{\bm{p}}( \hat{\tau}  > \gamma_r)$ and $\pi_l(\gamma_l) \equiv \max_{\bm{p} \in \mathcal{P}(\tau_0)} P_{\bm{p}}( \hat{\tau}  < \gamma_l)$. We can calculate $\pi_r(\gamma_r)$ as in Section \ref{sec:one-sided} and $\pi_l(\gamma_l)$ similarly. We then propose the following critical values $\gamma_r^{\ast}(\tau_0)$ and $\gamma_l^{\ast}(\tau_0)$:
\begin{equation*}
    \gamma_r^{\ast}(\tau_0) = \text{arg} \min \{\gamma_r:\pi_r(\gamma_r) \leq \alpha / 2\} \ \ \text{and} \ \ \gamma_l^{\ast}(\tau_0) = \text{arg} \max \{ \gamma_l:\pi_l(\gamma_l) \leq \alpha / 2 \}.
\end{equation*}
By construction, these critical values $\gamma_r^{\ast}(\tau_0)$ and $\gamma_l^{\ast}(\tau_0)$ satisfy (\ref{two-sided_prob}).

We obtain the confidence region of $\tau$ by inverting the testing procedure. We define $\widehat{CR}_{1-\alpha}$ as the set of hypothesized values that are not rejected by the proposed two-sided test, that is,
\[
\widehat{CR}_{1-\alpha} \ \equiv \ \left\{ \tau_0 \in [0,1] :   \gamma_l^{\ast}(\tau_0) \leq \hat{\tau} \leq \gamma_r^{\ast}(\tau_0) \right\}.
\]
By construction, $\widehat{CR}_{1-\alpha}$ satisfies
\[
\min_{\bm{p} \in \mathcal{P}_{\ast}} P_{\bm{p}} \left( \tau \in \widehat{CR}_{1-\alpha} \right) \ \geq \ 1-\alpha.
\]
In other words, this confidence region is uniformly valid over the Lipschitz class.

This confidence region is an interval. As discussed in Remark \ref{rem:one-sided_gamma}, $\gamma_r^{\ast}(\tau_0)$ is increasing in $\tau_0$. Similarly, $\gamma_l^{\ast}(\tau_0)$ is increasing in $\tau_0$. Suppose that $t_1 < t_2$ and $t_1, t_2 \in \widehat{CR}_{1-\alpha}$. Then, for any $t \in [t_1,t_2]$, we obtain
\[
\gamma_l^{\ast}(t) \leq \gamma_l^{\ast}(t_2) \leq \hat{\tau} \ \ \text{and} \ \ \hat{\tau} \leq \gamma_r^{\ast}(t_1) \leq \gamma_r^{\ast}(t).
\]
Hence, any $t$ within the interval $[t_1,t_2]$ must be contained in the confidence region $\widehat{CR}_{1-\alpha}$, which means that $\widehat{CR}_{1-\alpha}$ is an interval. Consequently, searching for the boundary points of $\widehat{CR}_{1-\alpha}$ is sufficient for constructing the confidence interval.

\begin{remark}\label{rem:CI_computation}
For example, we can calculate the left boundary point of $\widehat{CR}_{1-\alpha}$ using the following algorithm:\vspace{0.05in}
\begin{enumerate}
    \item Let $t_0 = 0$ and calculate $\gamma^{\ast}_r(t_0)$.
    \item For $k \geq 0$, if $\hat{\tau} > \gamma^{\ast}_r(t_k)$, set $t_{k+1} = t_k + 2^{-k-1}$. If not, set $t_{k+1} = t_k - 2^{-k-1}$.
    \item By repeating the above process, $t_k$ converges to the left boundary point of $\widehat{CR}_{1-\alpha}$.\vspace{0.05in}
\end{enumerate}
Using this algorithm, we can avoid calculating the critical value $\gamma^{\ast}_r(\tau_0)$ for every $\tau_0 \in [-1,1]$. We can calculate the right boundary point of $\widehat{CR}_{1-\alpha}$ similarly.
\end{remark}

\section{Simulation Results and an Empirical Application}\label{sec:numerical}
\subsection{Monte Carlo simulation}\label{sec:mc}
We demonstrate the performance of our estimator relative to existing estimators in Monte Carlo simulations. We compare our estimator (\textit{rdbinary}) with three different estimators: (1) the Gaussian estimator (\textit{gauss}) with homoskedastic variance $\sigma_i^2 = 1/4$ as in Section \ref{sec:compare_gauss}; (2) the \cite{XU20171}'s estimator (\textit{rd.mnl}), which is specific to multinomial outcomes, including binary outcomes as a special case; and (3) the \cite{Calonico.Cattaneo.Titiunik2014}'s estimator (\textit{rdrobust}).\footnote{We obtained the \textit{rd.mnl} package files from the author's personal website (\url{https://sites.google.com/view/kelixu/home}). For \textit{rd.mnl} and \textit{rdrobust}, we use their default specifications with bias-corrected robust estimation and inference.}

\begin{figure}[H]
    \centering
    \begin{minipage}{0.49\textwidth}
        \centering
        \includegraphics[width=0.95\textwidth]{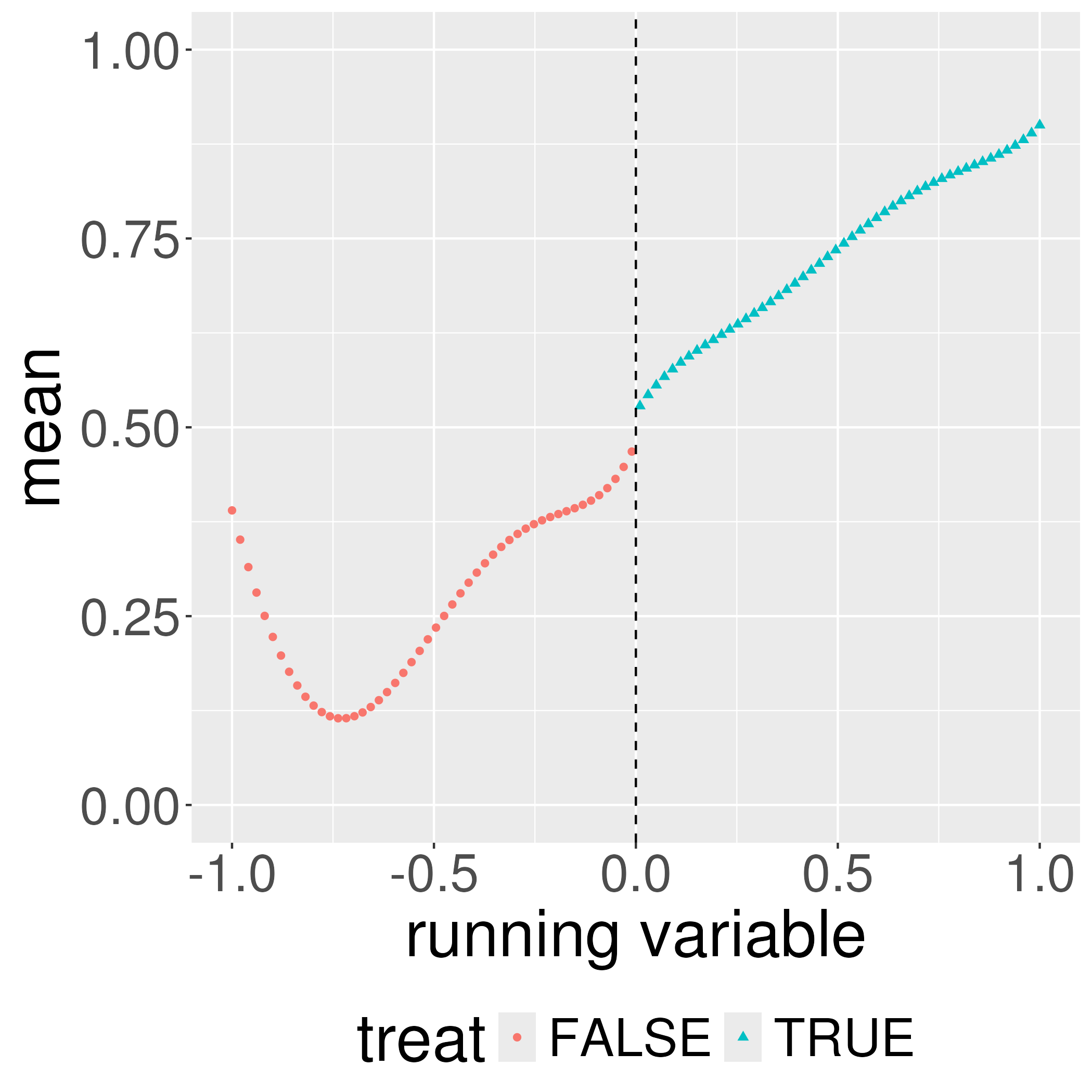}
        \caption{The \cite{Lee2008} model}
        \label{fig:shape_lee}
    \end{minipage}
    \hfill
    \begin{minipage}{0.49\textwidth}
        \centering
        \includegraphics[width=0.95\textwidth]{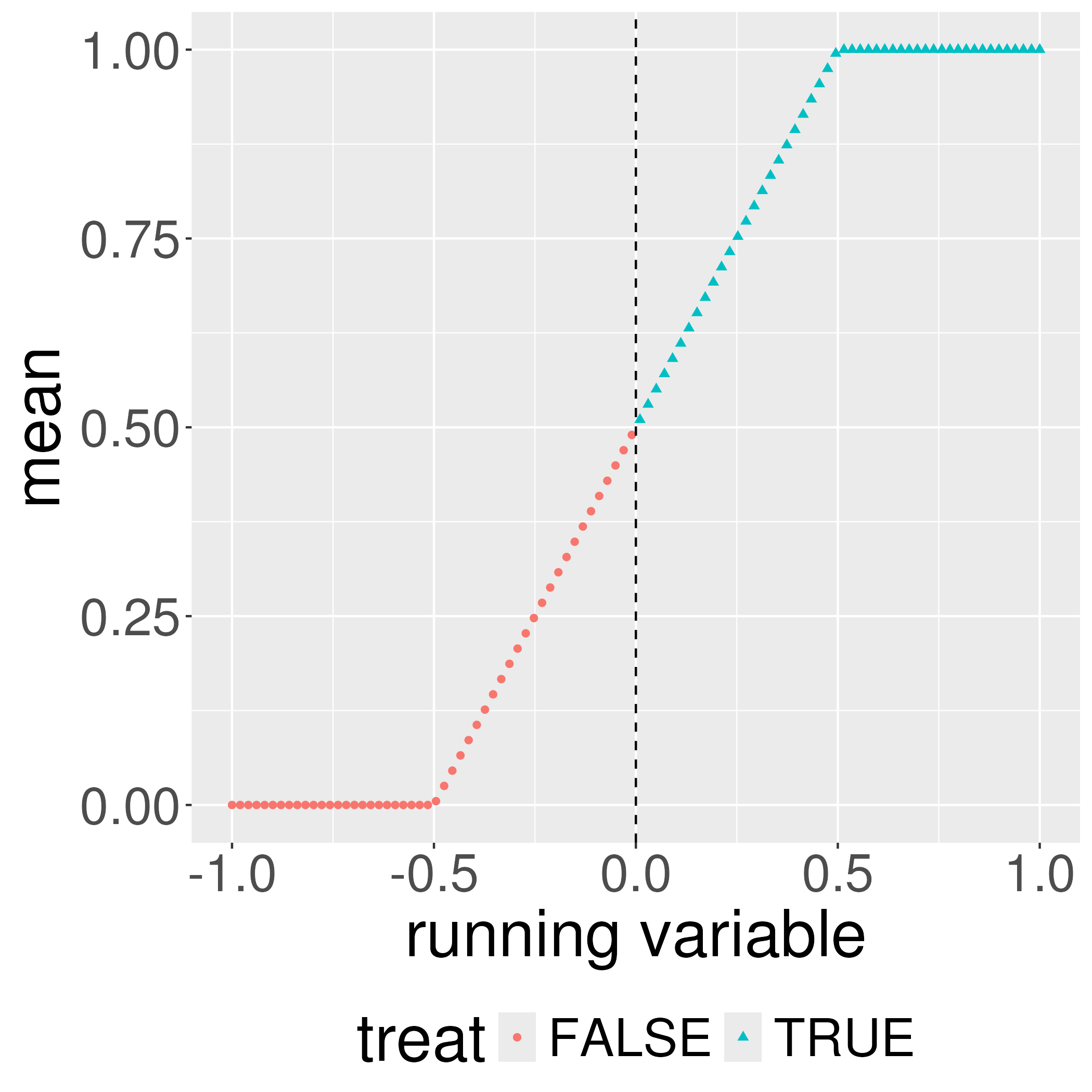}
        \caption{The worst-case model}
        \label{fig:shape_worst}
    \end{minipage}
    \vfill
    \begin{minipage}{0.49\textwidth}
        \centering
        \includegraphics[width=0.95\textwidth]{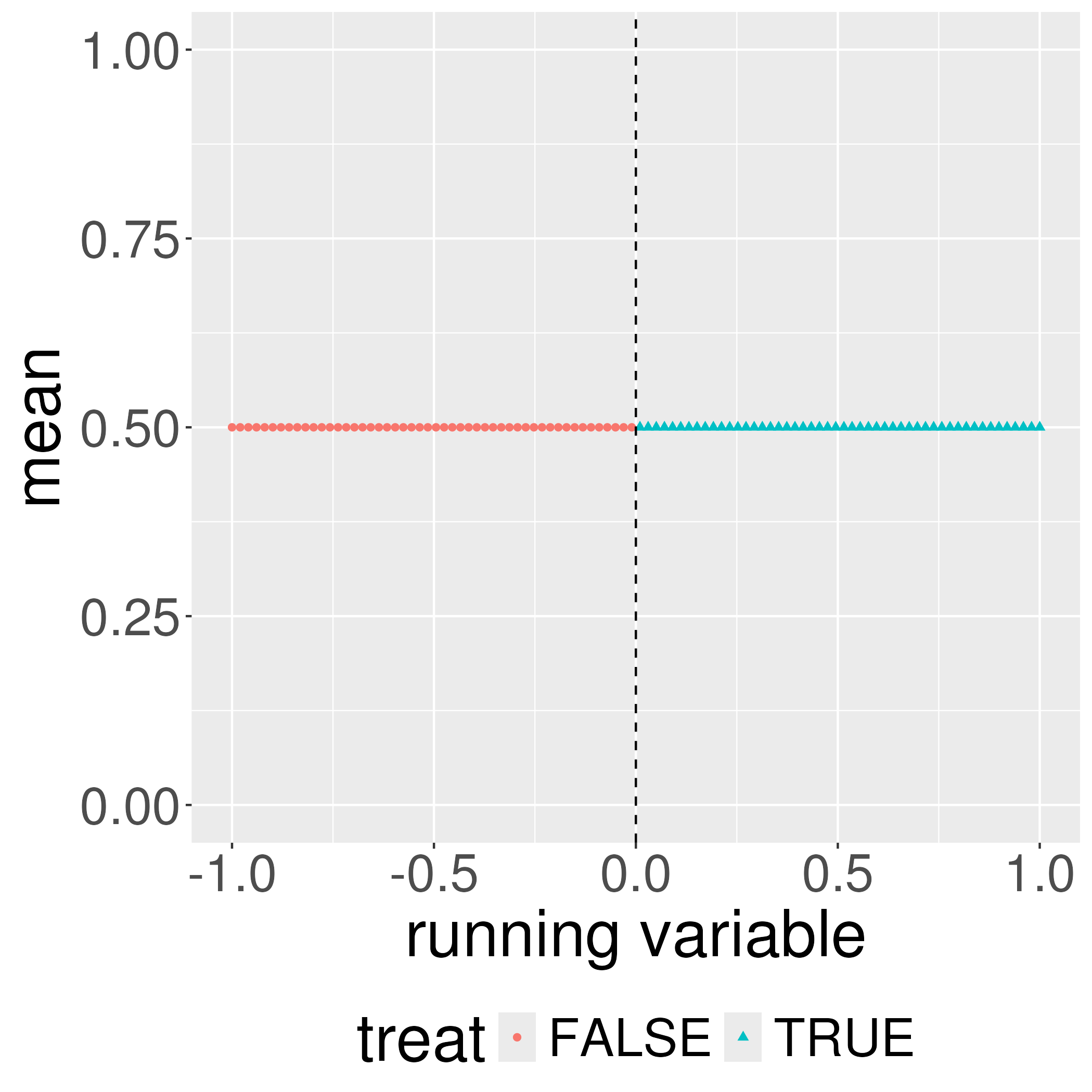}
        \caption{The flat model}
        \label{fig:shape_flat}
    \end{minipage}
\end{figure}

We compare their performance for three sample sizes $(N \in \{50, 100, 500\})$ of observations whose running-variable values are equally spaced between $-1$ and $1$. We consider the following three different models of the conditional mean of a binary dependent variable: (1) the \cite{Lee2008} model, which is a polynomial approximation of the conditional mean for \cite{Lee2008}'s data and is frequently used in simulation studies for RD designs; (2) the ``worst-case'' model, which is the parameter value $\bm{p}$ maximizing the MSE of any linear shrinkage estimator among parameter values such that $p_{0,+}=p_{0,-}=1/2$;\footnote{Note that the worst-case MSE of a linear shrinkage estimator is not necessarily attained at the parameter values of this model, since $p_{0,+}$ and $p_{0,-}$ are fixed at $1/2$.} and (3) a flat model in which the conditional probability is constant at $0.5$.
The three designs are illustrated in Figures \ref{fig:shape_lee}--\ref{fig:shape_flat}.
For each model, the dependent variable takes $1$ with the probability specified as \textit{mean} and otherwise $0$.

We consider the estimation and inference of $\tau=p_{0,+}-p_{0,-}$.
We use the true value of the Lipschitz constant $C$ for each design to implement our proposed method and the Gaussian method.
Our proposed estimator for $\tau$ is given by $\hat{\tau} = \hat{p}_{0,+}(\hat{\bm{w}}_{+}) - \hat{p}_{0,-}(\hat{\bm{w}}_{-})$, where $\hat{\bm{w}}_{+}$ and $\hat{\bm{w}}_{-}$ are chosen to minimize the worst-case MSE for the estimation of $p_{0,+}$ and $p_{0,-}$, respectively, as in Section \ref{sec:main}.
We then use $\hat{\tau}$ to construct a two-sided confidence interval for $\tau$ following the procedure in Section \ref{sec:inference}.\footnote{We computed the pair of critical values $\gamma_{r}^*(\tau_0)$ and $\gamma_{l}^*(\tau_0)$ by computing $\pi_{r}(\gamma_r)$ and $\pi_l(\gamma_l)$ with $3000$ draws of an $n$-dimensional Bernoulli random vector for each. The confidence intervals were constructed by inverting tests evaluated at $300$ grid points of $\tau_0$.}
An alternative, the Gaussian estimator, is $\tilde{\tau} = \hat{p}_{0,+}(\tilde{\bm{w}}_{+}) - \hat{p}_{0,-}(\tilde{\bm{w}}_{-})$, where $\tilde{\bm{w}}_{+}$ and $\tilde{\bm{w}}_{-}$ minimize the worst-case MSE for the estimation of $p_{0,+}$ and $p_{0,-}$, respectively, under the misspecified model where $Y_i \sim N(p_i,1/4)$, as in Section \ref{sec:gauss}.
Following \cite{kolesar2018discrete} and \cite{Armstrong2021ATE}, we construct a two-sided fixed-length confidence interval centered at $\tilde{\tau}$ with finite-sample validity under the Gaussian model.
Specifically, the $100\cdot (1-\alpha)\%$ confidence interval is given by $\left(\tilde{\tau}\pm \rm{cv}_\alpha\left(\rm{maxbias}(\tilde{\tau})/\rm{sd}(\tilde{\tau})\right)\cdot \rm{sd}(\tilde{\tau})\right)$, where $\rm{maxbias}(\tilde{\tau})$ denotes the maximum bias of $\tilde{\tau}$ under the Lipschitz class and $\rm{cv}_\alpha(b)$ denotes the $1-\alpha$ quantile of $|N(b,1)|$, the folded normal distribution with location and scale parameters $(b,1)$.

First, we demonstrate the point-estimation properties of our proposed estimator. Tables \ref{tab:point_Lee} and \ref{tab:point_100} compare the root MSE and bias for the estimation of the ATE at the cutoff, computed from $3000$ replication draws. Table \ref{tab:point_Lee} compares three different sample sizes under the Lee model. For all sample sizes, our estimator has substantially smaller MSEs than the other estimators. Furthermore, the differences decrease as the sample size increases, and the MSEs are relatively similar for $N = 500$. We note that the same pattern is confirmed for different designs with different Lipschitz constants $C$. Hence, our estimator is superior to the existing estimators in small samples, whereas their behaviors resemble in larger samples.


\begin{table}[ht]
\centering
\begin{tabular}{lrrrrrr}
  \hline
 & \multicolumn{2}{c}{N = 50} & \multicolumn{2}{c}{N = 100} & \multicolumn{2}{c}{N = 500} \\ \hline
 & root &  & root &  & root &  \\ 

 & MSE & Bias & MSE & Bias & MSE & Bias \\ 

  \hline
rdbinary & 0.264 & 0.063 & 0.223 & 0.067 & 0.141 & 0.065 \\ 
  gauss & 0.302 & 0.124 & 0.248 & 0.107 & 0.149 & 0.078 \\ 
  rd.mnl & 0.356 & 0.020 & 0.284 & 0.027 & 0.142 & 0.042 \\ 
  rdrobust & 0.578 & 0.037 & 0.423 & 0.033 & 0.190 & 0.036 \\ 
   \hline
\end{tabular}
\caption{Simulation: Point Estimates (Lee)} 
\label{tab:point_Lee}
\end{table}

\begin{table}[ht]
\centering
\begin{tabular}{lrrrrrr}
  \hline
 & \multicolumn{2}{c}{worst-case} & \multicolumn{2}{c}{Lee} & \multicolumn{2}{c}{flat-50} \\ \hline
 & root &  & root &  & root &  \\ 

 & MSE & Bias & MSE & Bias & MSE & Bias \\ 

  \hline
rdbinary & 0.239 & 0.136 & 0.223 & 0.067 & 0.088 & 0.000 \\ 
  gauss & 0.288 & 0.205 & 0.248 & 0.107 & 0.100 & 0.000 \\ 
  rd.mnl & 0.349 & -0.006 & 0.284 & 0.027 & 0.253 & -0.004 \\ 
  rdrobust & 0.417 & 0.001 & 0.423 & 0.033 & 0.423 & -0.004 \\ 
   \hline
\end{tabular}
\caption{Simulation: Point Estimates N = 100} 
\label{tab:point_100}
\end{table}

Second, we demonstrate the inference properties of our proposed method. Tables \ref{tab:CI_Lee} and \ref{tab:CI_100} compare the average length and coverage probability of the four confidence intervals, computed from $5000$ replication draws. Table \ref{tab:CI_Lee} shows that our confidence interval has shorter lengths with guaranteed coverage than \textit{rd.mnl} and \textit{rdrobust} for different sample sizes. Unlike the point estimation results, the differences in lengths remain similar as the sample size increases. Note that the Gaussian confidence interval has slightly shorter lengths while achieving the $95\%$ coverage for the Lee design. Nevertheless, the Gaussian confidence interval does \textit{not guarantee} uniform coverage across designs; for instance, it falls below $95\%$ for the flat design when $N=100$. This behavior is consistent with the fact that the Gaussian confidence interval is constructed under a \textit{missspecified} model, where the outcomes, and hence the linear estimators, are assumed to follow a normal distribution.
In contrast, our proposed confidence interval provides guaranteed coverage, a feature that makes it preferable.
We also note that the \textit{rdrobust} confidence interval is based on large-sample asymptotics and not specifically designed for binary outcomes, which results in unsatisfactory coverage properties for all designs, especially with small samples.

\begin{table}[ht]
\centering
\begin{tabular}{lrrrrrr}
  \hline
 & \multicolumn{2}{c}{N = 50} & \multicolumn{2}{c}{N = 100} & \multicolumn{2}{c}{N = 500} \\ \hline
 & CI length & coverage & CI length & coverage & CI length & coverage \\ 

  \hline
rdbinary & 1.463 & 0.989 & 1.232 & 0.988 & 0.763 & 0.991 \\ 
  gauss & 1.417 & 0.992 & 1.172 & 0.987 & 0.691 & 0.984 \\ 
  rd.mnl & 1.712 & 0.946 & 1.625 & 0.953 & 1.161 & 0.967 \\ 
  rdrobust & 1.615 & 0.888 & 1.481 & 0.906 & 0.814 & 0.929 \\ 
   \hline
\end{tabular}
\caption{Simulation: Interval Estimates (Lee)} 
\label{tab:CI_Lee}
\end{table}

\begin{table}[ht]
\centering
\begin{tabular}{lrrrrrr}
  \hline
 & \multicolumn{2}{c}{worst-case} & \multicolumn{2}{c}{Lee} & \multicolumn{2}{c}{flat-50} \\ \hline
 & CI length & coverage & CI length & coverage & CI length & coverage \\ 

  \hline
rdbinary & 1.156 & 0.978 & 1.232 & 0.988 & 0.414 & 0.963 \\ 
  gauss & 1.090 & 0.961 & 1.172 & 0.987 & 0.392 & 0.943 \\ 
  rd.mnl & 1.455 & 0.932 & 1.625 & 0.953 & 1.667 & 0.968 \\ 
  rdrobust & 1.469 & 0.908 & 1.481 & 0.906 & 1.498 & 0.906 \\ 
   \hline
\end{tabular}
\caption{Simulation: Interval Estimates N = 100} 
\label{tab:CI_100}
\end{table}

\subsection{Application}\label{sec:application}
We apply our estimator to a small-sample RD study by \cite{Brollo.Nannicini.Perotti.Tabellini2013} \citep{brolloReplicationDataPolitical2019}. \cite{Brollo.Nannicini.Perotti.Tabellini2013} exploit a regional fiscal rule in Brazil to study the impact of an additional government fiscal transfer on the frequency of corruption in local politics. In Brazil, 40 percent of municipal revenue comes from the \textit{Fundo de Participação dos Municipios} (FPM), which is allocated based on the population size of municipalities. Specifically, each municipality is allocated to one of the nine brackets based on its population. The bracketing fiscal rule induces population thresholds that discontinuously alter the amount of FPM transfers. Following \cite{Brollo.Nannicini.Perotti.Tabellini2013}, we focus on the first seven thresholds because of sample size limitations.

This study is particularly suited for our method for two reasons. First, their primary dependent variables are binary indicators. Specifically, they study the impact of the fiscal rule on two measures of \textit{corruption} indicators:
\begin{quote}
    broad corruption, which includes irregularities that could also be interpreted as bad administration rather than as overt corruption; and narrow corruption, which only includes severe irregularities that are also more likely to be visible to voters. \citep[page 1774]{Brollo.Nannicini.Perotti.Tabellini2013}
\end{quote}
Second, the sample size is relatively small. In particular, within each cutoff neighborhood, the sample size is limited to less than $400$ and is mostly around $100$ to $200$. In these small samples, our estimator is expected to be superior to other estimators that are based on asymptotic approximations.

The following tables present our \textit{rdbinary} estimates and \textit{rdrobust} estimates.\footnote{The original study runs global polynomial estimations for each cutoff neighborhood, as well as for the whole sample by pooling across cutoff neighborhoods. Their primary estimation is the fuzzy design, but we focus on the reduced-form sharp design estimates.} Tables \ref{tab:application_broad_pool} and \ref{tab:application_narrow_pool} report the pooled estimates over multiple cutoffs for the broad and narrow corruption indicators. $Crot$ denotes the rule-of-thumb value for the Lipschitz constant $C$, which is the largest (in absolute value) slope estimate from the binscatter estimation using the \textit{binsreg} package \citep{CattaneoCrumpFarrellFeng_2024}. In all tables, we report the point estimates and confidence intervals for three different values of the constant $C$: $Crot$; one-half of $Crot$; and $1.5$ times $Crot$.

\begin{table}[H]
\centering
\begin{tabular}{llll}
  \hline
estimator & C & point & CI \\ 
  \hline
rdrobust &  & 0.160 & [-0.033, 0.325] \\ 
  rdbinary & 0.5*Crot & 0.130 & [-0.021, 0.283] \\ 
  rdbinary & Crot & 0.147 & [-0.038, 0.342] \\ 
  rdbinary & 1.5*Crot & 0.145 & [-0.078, 0.367] \\ 
   \hline
\end{tabular}
\caption{Broad corruption pooled (N = 1202)} 
\label{tab:application_broad_pool}
\end{table}

\begin{table}[H]
\centering
\begin{tabular}{llll}
  \hline
estimator & C & point & CI \\ 
  \hline
rdrobust &  & 0.164 & [-0.054, 0.387] \\ 
  rdbinary & 0.5*Crot & 0.131 & [-0.011, 0.276] \\ 
  rdbinary & Crot & 0.154 & [-0.024, 0.338] \\ 
  rdbinary & 1.5*Crot & 0.155 & [-0.057, 0.366] \\ 
   \hline
\end{tabular}
\caption{Narrow corruption pooled (N = 1202)} 
\label{tab:application_narrow_pool}
\end{table}

For both indicators, our \textit{rdbinary} estimates appear to be similar to the \textit{rdrobust} estimates, which are valid for large samples. The sample size is $1,202$ for the entire pooling sample and, hence, is sufficiently large for the \textit{rdrobust} estimator.\footnote{We do not report \textit{rd.mnl} estimates because \textit{rd.mnl} sometimes fails to select a bandwidth in this dataset, particularly for small samples.} For both methods and outcomes, the $95\%$ confidence intervals include $0$. This finding differs from that of the original study,  which reports significant positive effects on the frequency of corruption. This difference highlights the importance of using local nonparametric estimations in RD designs. 

By pooling samples across multiple cutoffs, we obtain a sufficiently large sample. Nevertheless, heterogeneity across different cutoffs may be of interest, as the original study explores cutoff-specific estimates. However, only a few hundred observations are available around each individual cutoff, and the asymptotic approximation may not perform well in such small samples.

Tables \ref{tab:application_broad_each_1}, \ref{tab:application_broad_each_3}, and \ref{tab:application_broad_each_6} present our \textit{rdbinary} and \textit{rdrobust} estimates of the impact on the broad corruption indicator for seven different subsamples around each individual cutoff. See the Online Appendix for qualitatively similar results for the narrow corruption indicator. For all specifications, the confidence intervals for each subsample are much wider than those for the pooled sample. Nevertheless, our \textit{rdbinary} method tends to yield much shorter confidence intervals than \textit{rdrobust}. For example, Cutoff 3 has a sample size of $225$, which is too small for \textit{rdrobust} to offer any meaningful implications from its confidence interval. By contrast, our \textit{rdbinary} confidence intervals offer reasonable lower bounds for the impact on the broad corruption measure, which are not very negative compared with the lower bound of the confidence interval from \textit{rdrobust}.

\begin{table}[H]
\centering
\begin{tabular}{lllll}
  & \multicolumn{2}{c}{Cutoff 1} & \multicolumn{2}{c}{Cutoff 2}\\ \hline
estimator & point & CI & point & CI \\ 
  \hline
rdrobust & 0.038 & [-0.372, 0.447] & 0.057 & [-0.307, 0.422] \\ 
  rdbinary (0.5Crot) & 0.075 & [-0.128, 0.280] & 0.168 & [-0.186, 0.520] \\ 
  rdbinary (Crot) & 0.071 & [-0.193, 0.337] & 0.146 & [-0.298, 0.576] \\ 
  rdbinary (1.5Crot) & 0.072 & [-0.234, 0.375] & 0.140 & [-0.352, 0.632] \\ 
   \hline
\end{tabular}
\caption{Broad: at cutoffs 1 (N = 385) and  2 (N = 218)} 
\label{tab:application_broad_each_1}
\end{table}

\begin{table}[H]
\centering
\begin{tabular}{lllll}
  & \multicolumn{2}{c}{Cutoff 3} & \multicolumn{2}{c}{Cutoff 4}\\ \hline
estimator & point & CI & point & CI \\ 
  \hline
rdrobust & -0.099 & [-0.533, 0.335] &  0.058 & [-0.572, 0.687] \\ 
  rdbinary (0.5Crot) &  0.192 & [-0.088, 0.467] & -0.045 & [-0.458, 0.364] \\ 
  rdbinary (Crot) &  0.229 & [-0.117, 0.572] & -0.015 & [-0.518, 0.469] \\ 
  rdbinary (1.5Crot) &  0.232 & [-0.173, 0.635] &  0.011 & [-0.542, 0.570] \\ 
   \hline
\end{tabular}
\caption{Broad: at cutoffs 3 (N = 225) and  4 (N = 139)} 
\label{tab:application_broad_each_3}
\end{table}

\begin{table}[H]
\centering
\begin{tabular}{lllllll}
  & \multicolumn{2}{c}{Cutoff 5} & \multicolumn{2}{c}{Cutoff 6} & \multicolumn{2}{c}{Cutoff 7}\\ \hline
estimator & point & CI & point & CI & point & CI \\ 
  \hline
rdrobust & 0.719 & [-0.863, 2.302] & -0.078 & [-1.157, 1.000] & 2.096 & [-1.431, 5.623] \\ 
  rdbinary (0.5Crot) & 0.185 & [-0.232, 0.607] &  0.151 & [-0.307, 0.603] & 0.039 & [-0.490, 0.567] \\ 
  rdbinary (Crot) & 0.279 & [-0.263, 0.816] &  0.109 & [-0.458, 0.679] & 0.199 & [-0.512, 0.863] \\ 
  rdbinary (1.5Crot) & 0.330 & [-0.313, 0.936] &  0.081 & [-0.586, 0.721] & 0.246 & [-0.562, 0.963] \\ 
   \hline
\end{tabular}
\caption{Broad: at cutoffs 5 (N = 116) and  6 (N = 73) and  7 (N = 46)} 
\label{tab:application_broad_each_6}
\end{table}

\section{Conclusion}

Empirical studies often use RD designs with small samples, where
estimation and inference are particularly challenging.
Methods that rely on large-sample properties may not perform well in such settings.
Several finite-sample minimax estimators have been proposed. However, these estimators typically require knowledge of the variance, which is generally unavailable. Furthermore, the finite-sample validity of inference procedures based on these estimators requires the normality of the outcome variable.

In this study, we provide a minimax optimal estimator for RD designs with a binary outcome variable, together with its inference procedure.
The only tuning parameter of our proposed estimator is the Lipschitz constant; our estimator does not require specification of the conditional variance function, unlike existing minimax estimators in RD designs. Our estimator is also applicable to any bounded outcome variable. Thus, we offer a practical method that serves as a last resort for RD studies with relatively small effective sample sizes.

We demonstrate that the proposed estimator is superior to existing estimators in finite samples through numerical and simulation exercises. In a numerical exercise, we show that our estimator is $5$ to $20$\% more efficient in the worst-case root MSEs than a feasible version of the existing minimax optimal estimators for extremely small samples. Through simulation studies, we show that our estimator has much smaller MSEs than existing methods for sufficiently small sample sizes. Furthermore, we demonstrate that our inference procedure generates shorter confidence intervals with guaranteed coverage rates compared with existing large-sample methods. In the empirical application to a small-sample RD study, we document that our method generates similar results to the standard large-sample procedure for sufficiently large samples but provides much more informative results for small samples.

Our contribution provides a critical baseline for developing estimation procedures for binary or limited outcome variables in RD designs. Recent studies, such as \cite{noack_bias_aware_2024}, have considered bias-aware inference for fuzzy RD designs. As mentioned in the Introduction, the binary treatment status is a primary dependent variable in the first stage of the fuzzy design. Applying our results is not necessarily straightforward in that context because the first-stage estimand appears in the denominator of the target parameter. We reserve further extensions and generalizations of our results for future research.

\bibliographystyle{ecta}
\bibliography{reference} 

\appendix
\setcounter{table}{0}
\setcounter{figure}{0}
\renewcommand{\thetable}{A.\arabic{table}}
\renewcommand{\thefigure}{A.\arabic{figure}}
\section{Proofs}\label{sec:proofs}
\begin{proof}[Proof of Lemma \ref{lem:max_prob_1}]
Let $\tilde{\bm{\theta}} = ( \theta_0, \theta_1 + 2 |\theta_0 - \theta_1 |_{+}, \ldots , \theta_n + 2 |\theta_0 - \theta_n |_{+} )'$, where $|a|_{+} \equiv \max\{a,0\}$. If $\theta_i \geq \theta_0$, then $\tilde{\theta}_i - \tilde{\theta}_0 = \theta_i - \theta_0 \geq 0$. If $\theta_i < \theta_0$, then $\tilde{\theta}_i - \tilde{\theta}_0 = \theta_i + 2(\theta_0 - \theta_i) - \theta_0 = \theta_0 - \theta_i \geq 0$. Hence, $\tilde{\bm{\theta}}$ satisfies $\tilde{\theta}_i \geq \tilde{\theta}_0$.

Next, we show $\tilde{\bm{\theta}} \in \Theta$. If $\theta_i \geq \theta_0$, then we have $\tilde{\theta}_i = \theta_i \in [-1/2,1/2]$. If $\theta_i < \theta_0$, then we have $\tilde{\theta}_i = \theta_i + 2(\theta_0 - \theta_i) = \theta_0 +(\theta_0 - \theta_i) \in [-1/2,1/2]$ because $\theta_0 \in [-1/2, 0]$ and $\theta_0 - \theta_i \in [0,1/2]$. Hence, $\tilde{\bm{\theta}} \in [-1/2,1/2]^{n+1}$. It suffices to show that $|\tilde{\theta}_i - \tilde{\theta}_j| \leq C \|R_i-R_j\|$ for all $i$ and $j$. We consider the following three cases: (i) $\theta_i \geq \theta_0$ and $\theta_j \geq \theta_0$, (ii) $\theta_i \geq \theta_0$ and $\theta_j < \theta_0$, (iii) $\theta_i < \theta_0$ and $\theta_j< \theta_0$. In case (i), we have $|\tilde{\theta}_i - \tilde{\theta}_j| = |\theta_i - \theta_j| \leq C \|R_i-R_j\|$. In case (ii), we have
\begin{eqnarray*}
    |\tilde{\theta}_i - \tilde{\theta}_j| &=& \left| \theta_i - ( 2\theta_0 - \theta_j ) \right| \ = \ \left| (\theta_i - \theta_0) + (\theta_j - \theta_0 ) \right| \\
    & \leq & (\theta_i - \theta_0) + (\theta_0 - \theta_j ) \ = \ (\theta_i - \theta_j) \ \leq \ C \|R_i-R_j\|.
\end{eqnarray*}
Similarly, in case (iii), we have
\begin{eqnarray*}
    |\tilde{\theta}_i - \tilde{\theta}_j| &=& \left| ( 2\theta_0 - \theta_i ) - ( 2\theta_0 - \theta_j ) \right| \ = \ |\theta_i - \theta_j| \ \leq \ C \|R_i-R_j\|.
\end{eqnarray*}
Therefore, we obtain $\tilde{\bm{\theta}} \in \Theta$.

Finally, we show that $\text{MSE}(\bm{w},\bm{\theta}) \leq \text{MSE}(\bm{w}, \tilde{\bm{\theta}})$. Because we have $\theta_i \leq \tilde{\theta}_i$ and $\tilde{\theta}_0 = \theta_0$, we obtain $( \sum_{i=1}^n w_i \tilde{\theta}_i - \tilde{\theta}_0 )^2 \geq ( \sum_{i=1}^n w_i \theta_i - \theta_0 )^2$ when $\sum_{i=1}^n w_i \theta_i - \theta_0 \geq 0$. In addition, as shown above, we have $\tilde{\theta}_i - \tilde{\theta}_0 = |\theta_i - \theta_0|$ for all $i$. Because $\sum_{i=1}^n w_i \leq 1$ and $\theta_0 \leq 0$, we obtain
\begin{eqnarray*}
    \sum_{i=1}^n w_i \tilde{\theta}_i - \tilde{\theta}_0 &=& \sum_{i=1}^n w_i (\tilde{\theta}_i - \tilde{\theta}_0) - \left( 1 - \sum_{i=1}^n w_i \right) \theta_0 \ \geq \ \sum_{i=1}^n w_i (\tilde{\theta}_i - \tilde{\theta}_0) \\
    &=& \sum_{i=1}^n w_i |\theta_i-\theta_0| \ \geq \ \sum_{i=1}^n w_i (\theta_0-\theta_i) \\
    &=& \theta_0 - \sum_{i=1}^n w_i \theta_i - \left( 1 - \sum_{i=1}^n w_i \right) \theta_0 \ \geq \ \theta_0 - \sum_{i=1}^n w_i \theta_i.
\end{eqnarray*}
This implies that $( \sum_{i=1}^n w_i \tilde{\theta}_i - \tilde{\theta}_0 )^2 \geq ( \sum_{i=1}^n w_i \theta_i - \theta_0 )^2$ also holds when $\sum_{i=1}^n w_i \theta_i - \theta_0 \leq 0$. Furthermore, if $\theta_i < \theta_0$, then we have
\begin{eqnarray*}
    \tilde{\theta}_i^2 &=& (2\theta_0 - \theta_i)^2 \ = \ \theta_i^2 - 4\theta_0 \theta_i + 4 \theta_0^2 \ = \ \theta_i^2 + 4 \theta_0 (\theta_0 - \theta_i) \ \leq \ \theta_i^2.
\end{eqnarray*}
Because $\theta_i \geq \theta_0$ implies $\tilde{\theta}_i = \theta_i$, we obtain $1/4 - \tilde{\theta}_i^2 \geq 1/4 - \theta_i^2$. Therefore, we obtain $\text{MSE}(\bm{w},\bm{\theta}) \leq \text{MSE}(\bm{w}, \tilde{\bm{\theta}})$.
\end{proof}

\begin{proof}[Proof of Theorem \ref{thm:max_problem}]
As discussed in Section \ref{sec:minimax}, if $\bm{\theta} = (\theta_0, \ldots, \theta_n)'\in\Theta$ satisfies (\ref{theta_condition}), we obtain
\begin{equation*}
    \text{MSE}(\bm{w},\bm{\theta}) \ \leq \ \text{MSE}(\bm{w},\tilde{\bm{\theta}}(\theta_0)) \ \ \text{for all $\bm{w} \in \mathcal{W}$}. 
\end{equation*}
Because $\tilde{\bm{\theta}}(\theta_0) \in \Theta$ holds for all $\theta_0 \in [-1/2, 0]$, we obtain  (\ref{max_problem}).
\end{proof}

\begin{proof}[Proof of Lemma \ref{lem:monotone}]
Suppose that $\bm{w} \equiv (w_1, \ldots , w_n)' \in \mathcal{W}$ satisfies $w_j < w_{j+1}$ for some $j$. Letting $\tilde{\bm{w}} \equiv (w_1, \ldots, w_{j-1}, w_{j+1}, w_j, w_{j+2}, \ldots, w_n)'$, we have $\tilde{\bm{w}} \in \mathcal{W}$. For any $\bm{\theta} \in \Theta$, we observe that
\begin{eqnarray*}
    & & \text{MSE}(\bm{w},\bm{\theta}) - \text{MSE}(\tilde{\bm{w}}, \bm{\theta}) \\
    &=& \left( \sum_{i=1}^n w_i \theta_i - \theta_0 \right)^2 - \left( \sum_{i=1}^n w_i \theta_i - w_j \theta_j - w_{j+1} \theta_{j+1} + w_{j+1} \theta_j + w_{j} \theta_{j+1}  - \theta_0 \right)^2 \\
    & & + w_j^2 \left( 1/4 - \theta_j^2 \right) + w_{j+1}^2 \left( 1/4 - \theta_{j+1}^2 \right) - w_{j+1}^2 \left( 1/4 - \theta_j^2 \right) - w_j^2 \left( 1/4 - \theta_{j+1}^2 \right) \\
    &=&  \left( \sum_{i=1}^n w_i \theta_i - \theta_0 \right)^2 - \left\{ \left( \sum_{i=1}^n w_i \theta_i - \theta_0 \right) - (w_j - w_{j+1})(\theta_j - \theta_{j+1}) \right\}^2 \\
    & & - (w_j^2 - w_{j+1}^2)(\theta_j^2 - \theta_{j+1}^2) \\
    &=& 2 \left( \sum_{i=1}^n w_i \theta_i - \theta_0 \right) (w_j - w_{j+1})(\theta_j - \theta_{j+1}) - (w_j - w_{j+1})^2(\theta_j - \theta_{j+1})^2 \\
    & &  - (w_j - w_{j+1})(\theta_j - \theta_{j+1})(w_j + w_{j+1})(\theta_j + \theta_{j+1}) \\
    &=& (w_j - w_{j+1})(\theta_j - \theta_{j+1}) \left\{ 2 \left( \sum_{i=1}^n w_i \theta_i - \theta_0 \right) - 2 (w_j \theta_j + w_{j+1} \theta_{j+1}) \right\}.
\end{eqnarray*}
If $\bm{\theta}$ satisfies (\ref{theta_condition}), we obtain
\begin{eqnarray*}
    & & \left( \sum_{i=1}^n w_i \theta_i - \theta_0 \right) - (w_j \theta_j + w_{j+1} \theta_{j+1}) \\
    &=& \sum_{i \neq j, \, j+1} w_i \theta_i - \theta_0 \ \geq \ \left( \sum_{i \neq j, \, j+1} w_i - 1 \right) \theta_0 \ \geq \ 0.
\end{eqnarray*}
Because $\tilde{\bm{\theta}}(\theta_0)$ satisfies (\ref{theta_condition}) for all $\theta_0 \in [-1/2,0]$, we obtain
\begin{equation*}
    \text{MSE}(\bm{w},\tilde{\bm{\theta}}(\theta_0)) \ \geq \ \text{MSE}(\tilde{\bm{w}}, \tilde{\bm{\theta}}(\theta_0)) \ \ \text{for all $\theta_0 \in [-1/2,0]$.}
\end{equation*}
It follows from Theorem \ref{thm:max_problem} that we obtain
\[
\max_{\bm{\theta} \in \Theta} \text{MSE}(\bm{w},\bm{\theta}) \ \geq \ \max_{\bm{\theta} \in \Theta} \text{MSE}(\tilde{\bm{w}},\bm{\theta}).
\]
Hence, if $w_j < w_{j+1}$, then we can reduce the maximum MSE by exchanging $w_j$ for $w_{j+1}$. Therefore, by repeating this procedure until the weight vector becomes monotone, we can obtain $\tilde{\bm{w}} \in \mathcal{W}_0$ such that $\max_{\bm{\theta} \in \Theta} \text{MSE}(\tilde{\bm{w}},\bm{\theta}) \leq \max_{\bm{\theta} \in \Theta} \text{MSE}(\bm{w},\bm{\theta})$.
\end{proof}

\begin{proof}[Proof of Lemma \ref{lem:zero_weight}]
We observe that
\begin{eqnarray*}
    \frac{\partial}{\partial w_j} \text{MSE}(\bm{w},\bm{\theta}) &=& 2 \theta_j \left( \sum_{i=1}^n w_i \theta_i - \theta_0 \right) + 2 w_j \left( 1/4 - \theta_j^2 \right) \\
    &=& 2 \theta_j \left( \sum_{i \neq j} w_i \theta_i - \theta_0 \right) + w_j / 2,
\end{eqnarray*}
where $\sum_{i \neq j} w_i \theta_i - \theta_0 \geq 0$ when (\ref{theta_condition}) holds. If $\bm{\theta}$ satisfies (\ref{theta_condition}) and $\theta_j \geq 0$, $\frac{\partial}{\partial w_j} \text{MSE}(\bm{w},\bm{\theta})$ is nonnegative for all $\bm{w} \in \mathcal{W}$. If $C \|R_j\| \geq 1/2$, then the $j$-th element of $\tilde{\bm{\theta}}(\theta_0)$ is nonnegative for any $\theta_0 \in [-1/2,0]$. Hence, we have
\begin{equation*}
    \frac{\partial}{\partial w_j} \text{MSE}(\bm{w},\tilde{\bm{\theta}}(\theta_0)) \ \geq \ 0 \ \ \text{for any $\theta_0 \in [-1/2,0]$ and $\bm{w} \in \mathcal{W}$.}
\end{equation*}
Therefore, if $C \|R_j\| \geq 1/2$, then we obtain $\text{MSE}(\bm{w},\tilde{\bm{\theta}}(\theta_0)) \geq \text{MSE}(\tilde{\bm{w}},\tilde{\bm{\theta}}(\theta_0))$, where $\tilde{\bm{w}} \equiv (w_1, \ldots, w_{j-1}, 0, w_{j+1}, \ldots, w_n)'$. As a result, combined with Lemma \ref{lem:monotone}, we obtain $\min_{\bm{w} \in \mathcal{W}} \max_{\bm{\theta} \in \Theta} \text{MSE}(\bm{w},\bm{\theta}) = \min_{\bm{w} \in \mathcal{W}_1} \max_{\bm{\theta} \in \Theta} \text{MSE}(\bm{w},\bm{\theta})$.
\end{proof}

\begin{proof}[Proof of Lemma \ref{lem:compare_gauss}]
Because $\hat{\bm{w}}$ minimizes $\max_{\bm{\theta} \in \Theta} \mathrm{MSE}(\bm{w},\bm{\theta})$, the lower bound is trivial. Hence, we consider the upper bound. Because $\mathrm{MSE}(\bm{w},\bm{\theta}) \leq \mathrm{MSE}_g(\bm{w},\bm{\theta})$ and $\Theta \subset \Theta_g$, we have
\begin{eqnarray*}
\frac{\max_{\bm{\theta} \in \Theta} \mathrm{MSE}(\tilde{\bm{w}},\bm{\theta})}{\max_{\bm{\theta} \in \Theta} \mathrm{MSE}(\hat{\bm{w}},\bm{\theta})} & \leq & \frac{\max_{\bm{\theta} \in \Theta_g} \mathrm{MSE}_g(\tilde{\bm{w}},\bm{\theta})}{\max_{\bm{\theta} \in \Theta} \mathrm{MSE}(\hat{\bm{w}},\bm{\theta})} \ = \ \frac{\min_{\bm{w} \in \mathcal{W}} \max_{\bm{\theta} \in \Theta_g} \mathrm{MSE}_g(\bm{w},\bm{\theta})}{\max_{\bm{\theta} \in \Theta} \mathrm{MSE}(\hat{\bm{w}},\bm{\theta})}.
\end{eqnarray*}

First, we derive a lower bound of $\max_{\bm{\theta} \in \Theta} \mathrm{MSE}(\hat{\bm{w}},\bm{\theta})$. 
Since $\hat{w}_i = 0$ if $C\|R_i\| \geq 1/2$ by Lemma \ref{lem:zero_weight}, we obtain
\begin{eqnarray*}
\max_{\bm{\theta} \in \Theta} \mathrm{MSE}(\hat{\bm{w}},\bm{\theta}) &\ge&  \mathrm{MSE}(\hat{\bm{w}},\tilde{\bm{\theta}}(0)) \\
&=& C^2 \left( \sum_{i=1}^n \hat{w}_i \|R_i\| \right)^2 + \sum_{i=1}^n \hat{w}_i^2 \left( \frac{1}{4} - C^2 \|R_i\|^2 \right).
\end{eqnarray*}
Next, we derive an upper bound of $\min_{\bm{w} \in \mathcal{W}} \max_{\bm{\theta} \in \Theta_g} \mathrm{MSE}_g(\bm{w},\bm{\theta})$. If $\bm{w} \in \mathcal{W}$ satisfies $\sum_{i=1}^n w_i = 1$, then $\max_{\bm{\theta} \in \Theta_g} \mathrm{MSE}_g(\bm{w},\bm{\theta})$ can be written as follows:
\[
\max_{\bm{\theta} \in \Theta_g} \mathrm{MSE}_g(\bm{w},\bm{\theta}) \ = \ C^2 \left( \sum_{i=1}^n w_i \|R_i\| \right)^2 + \frac{1}{4} \sum_{i=1}^n w_i^2.
\]
Because $\tilde{\bm{w}}$ satisfies $\sum_{i=1}^n \tilde{w}_i = 1$, we obtain
\begin{eqnarray*}
\min_{\bm{w} \in \mathcal{W}} \max_{\bm{\theta} \in \Theta_g} \mathrm{MSE}_g(\bm{w},\bm{\theta}) &=& \min_{\bm{w} \in \mathcal{W}: \sum_{i=1}^n w_i = 1} \max_{\bm{\theta} \in \Theta_g} \mathrm{MSE}_g(\bm{w},\bm{\theta}) \\
&=& \min_{\bm{w} \in \mathcal{W}: \sum_{i=1}^n w_i = 1} \left\{ C^2 \left( \sum_{i=1}^n w_i \|R_i\| \right)^2 + \frac{1}{4} \sum_{i=1}^n w_i^2 \right\} \\
& \leq & C^2 \left( \sum_{i=1}^n (\hat{w}_i / \hat{u}) \|R_i\| \right)^2 + \frac{1}{4} \sum_{i=1}^n (\hat{w}_i / \hat{u})^2 \\
&=& \hat{u}^{-2} \left\{ C^2 \left( \sum_{i=1}^n \hat{w}_i \|R_i\| \right)^2 + \frac{1}{4} \sum_{i=1}^n \hat{w}_i^2 \right\}.
\end{eqnarray*}
Therefore, we obtain
\begin{eqnarray*}
\frac{\max_{\bm{\theta} \in \Theta} \mathrm{MSE}(\tilde{\bm{w}},\bm{\theta})}{\max_{\bm{\theta} \in \Theta} \mathrm{MSE}(\hat{\bm{w}},\bm{\theta})} &\le& \left( 1 + \frac{C^2 \sum_{i=1}^n \hat{w}_i^2 \|R_i\|^2}{C^2 \left( \sum_{i=1}^n \hat{w}_i \|R_i\| \right)^2 + \sum_{i=1}^n \hat{w}_i^2 \left( \frac{1}{4} - C^2 \|R_i\|^2 \right)} \right) \hat{u}^{-2} \\
&=& \left( 1 + \frac{C^2 \sum_{i=1}^n \hat{w}_i^2 \|R_i\|^2}{C^2 \sum_{i \neq j} \hat{w}_i \hat{w}_j \|R_i\| \|R_j\|  + \frac{1}{4} \sum_{i=1}^n \hat{w}_i^2 } \right) \hat{u}^{-2} \\
& \leq & \left( 1 + \frac{C^2 \sum_{i=1}^n \hat{w}_i^2 \|R_i\|^2}{\frac{1}{4} \sum_{i=1}^n \hat{w}_i^2 } \right) \hat{u}^{-2}.
\end{eqnarray*}
Because $\hat{w}_i = 0$ holds if $C\|R_i\| \geq 1/2$, we have $C^2 \sum_{i=1}^n \hat{w}_i^2 \|R_i\|^2 \leq \frac{1}{4} \sum_{i=1}^n \hat{w}_i^2$. As a result, the upper bound of (\ref{compare_gauss}) is bounded above by $2 \hat{u}^{-2}$.
\end{proof}

\begin{proof}[Proof of Theorem \ref{thm:asymptotic}]
We present and prove a more general version of Theorem \ref{thm:asymptotic}, which covers the H\"{o}lder class (\ref{holder_class}) described in Remark \ref{rem:holder}.
For a fixed pair of the smoothness constants $\gamma\in (0,1]$ and $C$, we redefine $\Theta$ and $\Theta_g$ as
$$
\Theta \equiv \{\bm{\theta} \in [-1/2,1/2]^{n+1} : |\theta_i - \theta_j| \leq C \| R_i-R_j \|^\gamma \ \text{for all $i$ and $j$} \}
$$
and
$$
\Theta_g \ \equiv \ \left\{ \bm{\theta} \in \mathbb{R}^{n+1} : |\theta_i - \theta_j| \leq C \| R_i-R_j \|^\gamma \ \text{for all $i$ and $j$} \right\},
$$
respectively.
Let
\[
\hat{\bm{w}} \in \mathrm{arg} \min_{\bm{w} \in \mathcal{W}} \max_{\bm{\theta} \in \Theta} \mathrm{MSE}(\bm{w},\bm{\theta}) \ \ \text{and} \ \ \tilde{\bm{w}} \in \mathrm{arg} \min_{\bm{w} \in \mathcal{W}} \max_{\bm{\theta} \in \Theta_g} \mathrm{MSE}_g(\bm{w},\bm{\theta}).
\]
We generalize Assumption \ref{ass:asymptotic} and Theorem \ref{thm:asymptotic} as follows.

\begin{assumption}\label{ass:asymptotic_holder}
The running variables $\{R_1, \ldots, R_n\} \in \mathbb{R}$ satisfy the following conditions:
\begin{itemize}
\item[(i)] $0 \leq \|R_1\| \leq \ldots \leq \|R_n\| \leq 1$.
\item[(ii)] There exist constants $c_1 > c_0 > 0$ such that, for any sufficiently large $n \in \mathbb{N}$, $c_0 x - n^{-\frac{1}{2\gamma+1}} \leq F_n(x) \leq c_1 x  + n^{-\frac{1}{2\gamma+1}}$ for all $x \in [0,1]$, where $F_n(\cdot)$ is the empirical distribution of $\|R_i\|$ when the sample size is $n$, that is,
$F_n(x) \ \equiv \ \frac{1}{n} \sum_{i=1}^n 1\{\|R_i\| \leq x\}$.
\end{itemize}
\end{assumption}

\begin{theorem}\label{thm:asymptotic_holder}
Under Assumption \ref{ass:asymptotic_holder}, we obtain $\max_{\bm{\theta} \in \Theta} \text{MSE}(\hat{\bm{w}},\bm{\theta}) = O(n^{-\frac{2\gamma}{2\gamma+1}})$ and
\begin{equation}
\frac{\max_{\bm{\theta} \in \Theta} \text{MSE}(\tilde{\bm{w}},\bm{\theta})}{\max_{\bm{\theta} \in \Theta} \text{MSE}(\hat{\bm{w}},\bm{\theta})} \ \to \ 1 . \nonumber
\end{equation}
\end{theorem}

In the following, we prove Theorem \ref{thm:asymptotic_holder}.
As a preliminary step, we note that the results of Theorem \ref{thm:max_problem} and Lemmas \ref{lem:monotone} and \ref{lem:zero_weight} also hold for the H\"{o}lder class. These can be established using similar arguments, after redefining $\tilde{\bm{\theta}}(t)$ and $\mathcal{W}_1$ as
\[
\tilde{\bm{\theta}}(t) \equiv \left( \tilde{\theta}_0(t), \tilde{\theta}_1(t), \ldots, \tilde{\theta}_n(t) \right)' \ \text{and} \ \tilde{\theta}_i(t) \equiv \min \{t + C \|R_i\|^\gamma, 1/2 \} \ \text{for $i = 0,1, \ldots, n$,}
\]
and $\mathcal{W}_1 \equiv \left\{ \bm{w} \in \mathcal{W}_0 : w_i = 0 \ \text{if $C\|R_i\|^\gamma \geq 1/2$} \right\}$, respectively.
Similarly, the bounds derived in Lemma \ref{lem:compare_gauss} hold with the upper bound replaced as follows:
\begin{align}
1 \ \leq \ \frac{\max_{\bm{\theta} \in \Theta} \text{MSE}(\tilde{\bm{w}},\bm{\theta})}{\max_{\bm{\theta} \in \Theta} \text{MSE}(\hat{\bm{w}},\bm{\theta})} \ \leq \ \hat{u}^{-2} \left( 1 + \frac{C^2 \sum_{i=1}^n \hat{w}_i^2 \|R_i\|^{2\gamma}}{\frac{1}{4} \sum_{i=1}^n \hat{w}_i^2 } \right). \label{compare_gauss_holder}
\end{align}
Also, the optimal weights that minimize $\max_{\bm{\theta} \in \Theta_g} \mathrm{MSE}_g(\bm{w},\bm{\theta})$ over $\bm{w}\in\mathbb{R}^n$ for the H\"{o}lder class satisfy $\sum_{i=1}^n w_i = 1$ and $w_i \geq 0$ for all $i$, similar to the optimal weights for the Lipschitz class.
This follows from the arguments in Appendix \ref{sec:gauss_weights}.
Finally, the maximum MSE of the estimator $\hat{p}_0(\bm{w})$ with $\sum_{i=1}^n w_i = 1$ over $\Theta_g$ is given by
$$
\max_{\bm{\theta} \in \Theta_g} \mathrm{MSE}_g(\bm{w},\bm{\theta}) =C^2 \left( \sum_{i=1}^n w_i \|R_i\|^\gamma \right)^2 + \frac{1}{4}\sum_{i=1}^n w_i^2.
$$

Now, let $\alpha=\frac{1}{2\gamma+1}$.
We consider a sufficiently large $n \in \mathbb{N}$ such that $c_0 x - n^{-\alpha} \leq F_n(x) \leq c_1 x  + n^{-\alpha}$ for all $x \in [0,1]$.
For any $\epsilon \ge 0$, let $N(\epsilon) \equiv \max\{ i \in \{1,\ldots,n\} : \|R_{i}\| \leq \epsilon \}$. Because $\|R_1\| \leq \cdots \leq \|R_n\|$, we have $N(\epsilon) = n F_n(\epsilon)$ for $\epsilon \in [0,1]$. Hence, under Assumption \ref{ass:asymptotic_holder}, we obtain
\[
c_0 n \epsilon - n^{1-\alpha} \ \leq \ N(\epsilon) \ \leq \ c_1 n \epsilon + n^{1-\alpha}, \ \ \ \forall \epsilon \in [0,1].
\]

First, we discuss the convergence rate of $\hat{p}_0(\hat{\bm{w}})$. Since $\mathrm{MSE}(\bm{w},\bm{\theta}) \leq \mathrm{MSE}_g(\bm{w},\bm{\theta})$ and $\Theta \subset \Theta_g$, we have
\begin{eqnarray}
\max_{\bm{\theta} \in \Theta} \mathrm{MSE}(\hat{\bm{w}},\bm{\theta})  &=& \min_{\bm{w} \in \mathcal{W}} \max_{\bm{\theta} \in \Theta} \mathrm{MSE}(\bm{w},\bm{\theta}) \nonumber \\
& \leq & \min_{\bm{w} \in \mathcal{W}} \max_{\bm{\theta} \in \Theta_g} \mathrm{MSE}_g(\bm{w},\bm{\theta}) \ = \ \max_{\bm{\theta} \in \Theta_g} \mathrm{MSE}_g(\tilde{\bm{w}},\bm{\theta}).\label{upper_minimax}\nonumber
\end{eqnarray}
This implies that if $\max_{\bm{\theta} \in \Theta_g} \mathrm{MSE}_g(\tilde{\bm{w}},\bm{\theta})$ converges to zero as $n\rightarrow\infty$, $\max_{\bm{\theta} \in \Theta} \mathrm{MSE}(\hat{\bm{w}},\bm{\theta})$ converges to zero no slower than $\max_{\bm{\theta} \in \Theta_g} \mathrm{MSE}_g(\tilde{\bm{w}},\bm{\theta})$.
If $2 c_0^{-1} n^{-\alpha} \leq \epsilon \leq 1$, then $N(\epsilon) \geq c_0 n \epsilon-n^{1-\alpha} \geq n^{1-\alpha}>0$, and we obtain
\begin{eqnarray*}
\max_{\bm{\theta} \in \Theta_g} \mathrm{MSE}_g(\tilde{\bm{w}},\bm{\theta}) &=& \min_{\bm{w} \in \mathcal{W}: \sum_{i=1}^n w_i = 1} \max_{\bm{\theta} \in \Theta_g} \mathrm{MSE}_g(\bm{w},\bm{\theta}) \\
&=& \min_{\bm{w} \in \mathcal{W}: \sum_{i=1}^n w_i = 1} \left\{ C^2 \left( \sum_{i=1}^n w_i \|R_i\|^\gamma \right)^2 + \frac{1}{4} \sum_{i=1}^n w_i^2 \right\} \\
& \leq & C^2 \left( \frac{1}{N(\epsilon)} \sum_{i=1}^{N(\epsilon)} \|R_i\|^\gamma \right)^2 + \frac{1}{4N(\epsilon)} \\
& \leq & C^2 \epsilon^{2\gamma} + \frac{1}{4N(\epsilon)} \ \leq \ C^2 \epsilon^{2\gamma} + \frac{1}{4(c_0 n \epsilon-n^{1-\alpha})},
\end{eqnarray*}
where the first equality holds since $\tilde{\bm{w}}$ satisfies $\sum_{i=1}^n \tilde{w}_i = 1$,
the first inequality follows by setting 
$$
\bm{w} = \left( \underbrace{\frac{1}{N(\epsilon)}, \ldots, \frac{1}{N(\epsilon)}}_{N(\epsilon)}, 0, \ldots , 0 \right)',
$$ 
and the second inequality holds since $\|R_i\| \le \epsilon$ for $i=1,\dots,N(\epsilon)$.
If we set $\epsilon = O(n^{-\frac{1}{2\gamma+1}})$ satisfying $\epsilon \geq 2 c_0^{-1} n^{-\alpha}$, which exists for $\alpha\ge \frac{1}{2\gamma+1}$, then the right-hand side becomes $O(n^{-\frac{2\gamma}{2\gamma+1}})$.
For example, if we set $\epsilon=2 c_0^{-1} n^{-\frac{1}{2\gamma+1}}$, which satisfies $\epsilon \geq 2 c_0^{-1} n^{-\alpha}$ for $\alpha\ge \frac{1}{2\gamma+1}$, then the right-hand side becomes
$$
4^\gamma C^2  c_0^{-2\gamma} n^{-\frac{2\gamma}{2\gamma+1}} + \frac{1}{4(2 n^{\frac{2\gamma}{2\gamma+1}}-n^{1-\alpha})}=\left(4^\gamma C^2  c_0^{-2\gamma}+\frac{1}{4(2-n^{\frac{1}{2\gamma+1}-\alpha})}\right)n^{-\frac{2\gamma}{2\gamma+1}}=O(n^{-\frac{2\gamma}{2\gamma+1}}).
$$
Hence, $\max_{\bm{\theta} \in \Theta} \mathrm{MSE}(\hat{\bm{w}},\bm{\theta})\le \max_{\bm{\theta} \in \Theta_g} \mathrm{MSE}_g(\tilde{\bm{w}},\bm{\theta})=O(n^{-\frac{2\gamma}{2\gamma+1}})$.

Next, we show that $\frac{\max_{\bm{\theta} \in \Theta} \text{MSE}(\tilde{\bm{w}},\bm{\theta})}{\max_{\bm{\theta} \in \Theta} \text{MSE}(\hat{\bm{w}},\bm{\theta})} \ \to \ 1$.
For any $\bm{w} \in \mathcal{W}$, we observe that
\begin{eqnarray*}
\max_{\bm{\theta} \in \Theta} \mathrm{MSE}(\hat{\bm{w}},\bm{\theta}) & \geq & \max_{\bm{\theta} \in \Theta} \left( \sum_{i=1}^n \hat{w}_i \theta_i - \theta_0 \right)^2 \ \geq \ \max_{\bm{\theta} \in \Theta : \theta_0 = -1/2} \left( \sum_{i=1}^n \hat{w}_i \theta_i + \frac{1}{2} \right)^2 \\
& = & \max_{\bm{\theta} \in \Theta: \theta_0 = -1/2} \left\{ \sum_{i=1}^n \hat{w}_i (\theta_i + 1/2) + \frac{1}{2} \left( 1 - \sum_{i=1}^n \hat{w}_i \right)  \right\}^2 \\
& \geq & \frac{1}{4} \left( 1 - \sum_{i=1}^n \hat{w}_i \right)^2 \ = \ \frac{1}{4} \left( 1 - \hat{u} \right)^2,
\end{eqnarray*}
where the last inequality follows from $\theta_i + 1/2 \geq 0$. This implies that $\hat{u}$ converges to one because $\max_{\bm{\theta} \in \Theta} \mathrm{MSE}(\hat{\bm{w}},\bm{\theta})$ converges to zero.
From the finite-sample bounds (\ref{compare_gauss_holder}), it suffices to show  that 
\begin{equation}
\frac{C^2 \sum_{i=1}^n \hat{w}_{i}^2 \|R_{i}\|^{2\gamma}}{\frac{1}{4} \sum_{i=1}^n \hat{w}_{i}^2 } \ \to \ 0. \label{sc_asymptotics}
\end{equation}

If $2 c_0^{-1} n^{-\alpha} \leq \epsilon \leq 1$, we can bound the left-hand side of (\ref{sc_asymptotics}) as follows:
\begin{eqnarray}
\frac{C^2 \sum_{i=1}^n \hat{w}_{i}^2 \|R_{i}\|^{2\gamma}}{\frac{1}{4} \sum_{i=1}^n \hat{w}_{i}^2 } & = & \frac{4C^2 \sum_{i=1}^{N(\epsilon)} \hat{w}_i^2 \|R_i\|^{2\gamma} + 4C^2 \sum_{i=N(\epsilon)+1}^n \hat{w}_i^2 \|R_i\|^{2\gamma} }{\sum_{i=1}^n \hat{w}_{i}^2} \nonumber\\
& \leq & 4C^2 \left\{ \frac{\epsilon^{2\gamma} \left( \sum_{i=1}^{N(\epsilon)} \hat{w}_i^2 \right) + \hat{w}_{N(\epsilon)}^2 \left( \sum_{i=N(\epsilon)+1}^n \|R_i\|^{2\gamma} \right)}{\sum_{i=1}^n \hat{w}_{i}^2} \right\} \nonumber\\
& \leq & 4 C^2 \left\{ \epsilon^{2\gamma} + \frac{n \hat{w}_{N(\epsilon)}^2}{\sum_{i=1}^n \hat{w}_{i}^2} \right\}, \label{sc_asymptotics_ineq}
\end{eqnarray}
where the first inequality follows since $\|R_i\| \le \epsilon$ for $i=1,\dots,N(\epsilon)$ and $\hat{w}_i \le \hat{w}_{N(\epsilon)}$ for $i=N(\epsilon)+1,\ldots,n$ (by the extension of Lemma \ref{lem:monotone} for the H\"{o}lder class), and the second inequality follows since $\|R_i\|\leq 1$ for all $i=N(\epsilon)+1,\ldots,n$ by Assumption \ref{ass:asymptotic_holder}.

To further bound the right-hand side of (\ref{sc_asymptotics_ineq}), we obtain a lower bound on $\sum_{i=1}^n \hat{w}_{i}^2$ and an upper bound on $\hat{w}_{N(\epsilon)}^2$.
A lower bound on $\sum_{i=1}^n \hat{w}_{i}^2$ is given by
\begin{align}
\sum_{i=1}^n \hat{w}_{i}^2 = \hat{u}^{2} \sum_{i=1}^n (\hat{w}_{i} / \hat{u})^2 \geq \hat{u}^2 n^{-1}, \label{sc_asymptotics_lb}
\end{align}
where the inequality follows from the fact that $\sum_{i=1}^nw_i^2\ge n^{-1}$ for any $(w_1,\ldots,w_n)\in\mathbb{R}^n$ such that $\sum_{i=1}^nw_i=1$.
To obtain an upper bound on $\hat{w}_{N(\epsilon)}^2$, we observe that, if $2 c_0^{-1} n^{-\alpha} \leq \epsilon \leq 1$,
\begin{eqnarray*}
O(n^{-\frac{2\gamma}{2\gamma+1}}) & = & \max_{\bm{\theta} \in \Theta} \text{MSE}(\hat{\bm{w}},\bm{\theta}) \ \geq \ \text{MSE}(\hat{\bm{w}},\tilde{\bm{\theta}}(0)) \ge \left( \sum_{i=1}^n \hat{w}_i \min \{C \|R_i\|^\gamma, 1/2 \} \right)^2\\
& = & C^2 \left( \sum_{i=1}^n \hat{w}_i \| R_i \|^\gamma \right)^2 \ \geq \ C^2 \hat{w}_{N(\epsilon)}^2 \left( \sum_{i=1}^{N(\epsilon)} \| R_i \|^\gamma \right)^2,
\end{eqnarray*}
where the second equality follows since $\hat w_i=0$ if $C \| R_i \|^\gamma\ge 1/2$ (by the extension of Lemma \ref{lem:zero_weight} for the H\"{o}lder class) and the last inequality follows since $\hat{w}_i \ge \hat{w}_{N(\epsilon)}$ for $i=1,\ldots,N(\epsilon)$ (by the extension of Lemma \ref{lem:monotone} for the H\"{o}lder class).
Below, we show that
\begin{align}
\sum_{i=1}^{N(\epsilon)} \| R_i \|^\gamma \ge \frac{(N(\epsilon)-n^{1-\alpha})^2}{2c_1n} \ \ \text{if } \epsilon > 2 c_0^{-1} n^{-\alpha}. \label{sc_asymptotics_integral}
\end{align}
Once it is shown, it follows from the assumption $N(\epsilon)\ge c_0 n \epsilon - n^{1-\alpha}$ that
\begin{eqnarray*}
O(n^{-\frac{2\gamma}{2\gamma+1}}) & \geq & C^2 \hat{w}_{N(\epsilon)}^2 \frac{(N(\epsilon)-n^{1-\alpha})^4}{4c_1^2n^2} \\
& \geq & C^2 \hat{w}_{N(\epsilon)}^2 \frac{(c_0n\epsilon-2n^{1-\alpha})^4}{4c_1^2n^2}=\frac{C^2}{4c_1^2} \hat{w}_{N(\epsilon)}^2(c_0n^{1/2}\epsilon-2n^{1/2-\alpha})^4,
\end{eqnarray*}
which implies that there exists $c_2>0$ (which is independent of $\epsilon$ and $n$) such that
\begin{align}
\hat{w}_{N(\epsilon)}^2\le\frac{c_2 n^{-\frac{2\gamma}{2\gamma+1}}}{(c_0n^{1/2}\epsilon-2n^{1/2-\alpha})^{4}} \ \ \text{if } \epsilon > 2 c_0^{-1} n^{-\alpha}. \label{sc_asymptotics_ub}
\end{align}
Now we show (\ref{sc_asymptotics_integral}).
Since $\|R_i\|\leq 1$ for all $i=1,\ldots,N(\epsilon)$ by Assumption \ref{ass:asymptotic_holder}, $\sum_{i=1}^{N(\epsilon)} \| R_i \|^\gamma\ge \sum_{i=1}^{N(\epsilon)} \| R_i \|$.
Therefore, it suffices to show that, if $\epsilon > 2 c_0^{-1} n^{-\alpha}$, $\sum_{i=1}^{N(\epsilon)} \| R_i \| \ge \frac{(N(\epsilon)-n^{1-\alpha})^2}{2c_1n}$.
Let $x^*(\epsilon)\equiv\frac{N(\epsilon)-n^{1-\alpha}}{c_1n}$, which is the unique solution to $c_1nx+n^{1-\alpha}=N(\epsilon)$ with respect to $x$.
If $\epsilon > 2 c_0^{-1} n^{-\alpha}$, we must have $0< x^*(\epsilon) \le \|R_{N(\epsilon)}\|$, since $c_1n\cdot 0+n^{1-\alpha}< c_0 n\epsilon -n^{1-\alpha}\le N(\epsilon)$ and $c_1n\|R_{N(\epsilon)}\|+n^{1-\alpha}\ge N(\|R_{N(\epsilon)}\|)= N(\epsilon)$ under Assumption \ref{ass:asymptotic_holder}.
Here, $N(\|R_{N(\epsilon)}\|)= N(\epsilon)$ holds by the definition of $N(\epsilon)$.
Also, let
\begin{eqnarray*}
g_\epsilon(x)\equiv\begin{cases}
    c_1 nx+n^{1-\alpha} & \text{if } 0< x\le x^*(\epsilon),\\
    N(\epsilon) & \text{if } x^*(\epsilon)< x\le 1.
    \end{cases}
\end{eqnarray*}
Note that $g_\epsilon(x)\ge N(x)$ for all $x\in (0,\|R_{N(\epsilon)}\|]$, since $g_\epsilon(x)=c_1 nx+n^{1-\alpha}\ge N(x)$ if $0< x\le x^*(\epsilon)$ by Assumption \ref{ass:asymptotic_holder} and $g_\epsilon(x)=N(\epsilon)=N(\|R_{N(\epsilon)}\|)\ge N(x)$ if $x^*(\epsilon) < x \le \|R_{N(\epsilon)}\|$.
Therefore, we have
$$
\int_0^{\|R_{N(\epsilon)}\|}(N(\epsilon)-N(x))dx\ge \int_0^{\|R_{N(\epsilon)}\|}(N(\epsilon)-g_\epsilon(x))dx.
$$
We calculate each of both sides of the above inequality:
\begin{eqnarray*}
    &&\int_0^{\|R_{N(\epsilon)}\|}(N(\epsilon)-N(x))dx\\
    &&=\int_0^{\|R_{1}\|}(N(\epsilon)-0)dx+\int_{\|R_{1}\|}^{\|R_{2}\|}(N(\epsilon)-1)dx+\cdots+\int_{\|R_{N(\epsilon)-1}\|}^{\|R_{N(\epsilon)}\|}(N(\epsilon)-(N(\epsilon)-1))dx\\
    &&=N(\epsilon)\|R_1\|+(N(\epsilon)-1)(\|R_2\|-\|R_1\|)+\cdots+(\|R_{N(\epsilon)}\|-\|R_{N(\epsilon)-1}\|)\\
    &&=\sum_{i=1}^{N(\epsilon)} \| R_i \|
\end{eqnarray*}
and
\begin{eqnarray*}
    \int_0^{\|R_{N(\epsilon)}\|}(N(\epsilon)-g_\epsilon(x))dx
    & = &\int_0^{x^*(\epsilon)}(N(\epsilon)-c_1 nx-n^{1-\alpha})dx = \frac{(N(\epsilon)-n^{1-\alpha})^2}{2c_1n}.
\end{eqnarray*}
Thus, we obtain $\sum_{i=1}^{N(\epsilon)} \| R_i \| \ge \frac{(N(\epsilon)-n^{1-\alpha})^2}{2c_1n}$.
See Figure \ref{fig:bound_N_ep} for the intuition for this argument.

\begin{figure}[H]
    \centering
    \includegraphics[width=12cm]{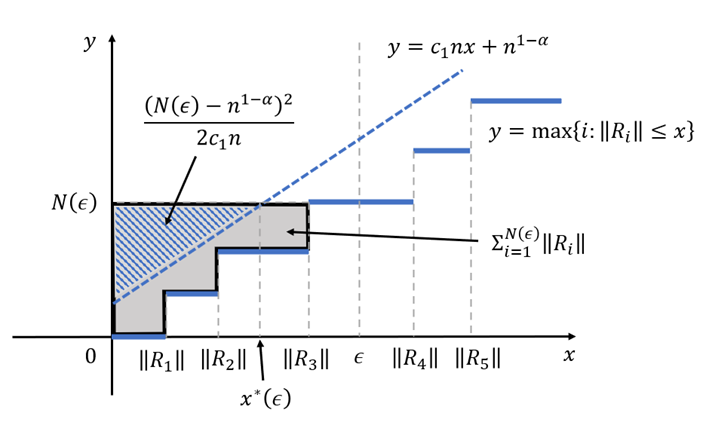}
    \caption{The blue solid line denotes a function $y = \max\{i: \|R_i\| \leq x\}$ and the blue dotted line denotes a function $y = c_1 n x + n^{1-\alpha}$. The area of the gray region is $\sum_{i=1}^{N(\epsilon)} \| R_i \|$ and the area of the shaded triangle is $\frac{(N(\epsilon)-n^{1-\alpha})^2}{2c_1 n}$.}
    \label{fig:bound_N_ep}
\end{figure}

Finally, combining the lower bound (\ref{sc_asymptotics_lb}) on $\sum_{i=1}^n \hat{w}_{i}^2$ and the upper bound (\ref{sc_asymptotics_ub}) on $\hat{w}_{N(\epsilon)}^2$ obtained above yields the following bound on the right-hand side of (\ref{sc_asymptotics_ineq}): if $\epsilon > 2 c_0^{-1} n^{-\alpha}$,
\begin{eqnarray*}
4 C^2 \left\{ \epsilon^{2\gamma} + \frac{n \hat{w}_{N(\epsilon)}^2}{\sum_{i=1}^n \hat{w}_{i}^2} \right\} & \leq & 4 C^2 \left\{ \epsilon^{2\gamma} + \frac{c_2 n^{\frac{2\gamma+2}{2\gamma+1}}}{(c_0n^{1/2}\epsilon-2n^{1/2-\alpha})^{4}}\hat u^{-2} \right\}\\
& = & 4 C^2 \left\{ \epsilon^{2\gamma} + \frac{c_2}{(c_0n^{\frac{\gamma}{4\gamma+2}}\epsilon-2n^{\frac{\gamma}{4\gamma+2}-\alpha})^{4}}\hat u^{-2} \right\}.
\end{eqnarray*}
If we set $\epsilon = n^{-\beta}$ for some $0<\beta<\min\{\frac{\gamma}{4\gamma+2},\alpha\}$ (which is equivalent to $0<\beta<\frac{\gamma}{4\gamma+2}$ as $\alpha=\frac{1}{2\gamma+1}\ge \frac{\gamma}{4\gamma+2}$), then $\epsilon > 2 c_0^{-1} n^{-\alpha}$ for any sufficiently large $n$, and we have
\begin{eqnarray*}
4 C^2 \left\{ \epsilon^{2\gamma} + \frac{c_2}{(c_0n^{\frac{\gamma}{4\gamma+2}}\epsilon-2n^{\frac{\gamma}{4\gamma+2}-\alpha})^{4}}\hat u^{-2} \right\}&=&4 C^2 \left\{ n^{-2\beta\gamma} + \frac{c_2}{(c_0n^{\frac{\gamma}{4\gamma+2}-\beta}-2n^{\frac{\gamma}{4\gamma+2}-\alpha})^{4}}\hat u^{-2} \right\}\\
&=&o(1).
\end{eqnarray*}
Therefore, we obtain the desired result because (\ref{sc_asymptotics}) holds.

\end{proof}

\begin{proof}[Proof of Theorem \ref{thm:inference}]
Fix $\bm{p} = (\bm{p}_{+}',\bm{p}_{-}')' = (p_{0,+}, \ldots, p_{n_{+},+}, p_{0,-}, \ldots, p_{n_{-},-})' \in \mathcal{P}(\tau_0)$. Define
\begin{eqnarray*}
\bm{Y} &\equiv& (Y_{1,+}, \ldots, Y_{n_{+},+}, Y_{1,-}, \ldots, Y_{n_{-},-})', \\
\tilde{\bm{Y}} &\equiv& (\tilde{Y}_{1,+}, \ldots, \tilde{Y}_{n_{+},+}, \tilde{Y}_{1,-}, \ldots, \tilde{Y}_{n_{-},-})', \\
\hat{\tau}(\bm{Y}) &\equiv& \sum_{i=1}^{n_{+}}w_{i,+} \left( Y_{i,+} - \frac{1}{2} \right) - \sum_{i=1}^{n_{-}}w_{i,-} \left( Y_{i,-} - \frac{1}{2} \right),
\end{eqnarray*}
where $\bm{Y}$ and $\tilde{\bm{Y}}$ follow Bernoulli distribution with parameters $\bm{p}$ and $\tilde{\bm{p}}(p_{0,+},\tau_0)$, respectively. Then, $\hat{\tau}(\bm{Y})$ is increasing function in $Y_{i,+}$ and decreasing in $Y_{i,-}$. Because $\tilde{Y}_{i,+}$ has first-order stochastic dominance over $Y_{i,+}$ and $-\tilde{Y}_{i,-}$ has first-order stochastic dominance over $-Y_{i,-}$, it follows from Lemma 1 of \cite{ishihara2023bandwidth} that we have
\begin{eqnarray*}
    P(\hat{\tau}(\bm{Y}) > \gamma) \ \leq \ P(\hat{\tau}(\tilde{\bm{Y}}) > \gamma).
\end{eqnarray*}
In addition, we have $\tilde{\bm{p}}(p,\tau_0) \in \mathcal{P}(\tau_0)$ for any $p \in [\max \{0,\tau_0\}, \min \{1,1+\tau_0 \}]$. Hence, we obtain (\ref{inference_max}).
\end{proof}

\section{Minimax estimation for the average treatment effect}\label{sec:ATE}
In this section, we consider the same setting in Remark \ref{rem:ATE} at the end of Section \ref{sec:main} and provide how to compute the maximum MSE for the ATE. We consider the following estimator of the ATE:
\begin{eqnarray*}
\hat{\tau}(\bm{w}_{+}, \bm{w}_{-}) &\equiv& \hat{p}_{0,+}(\bm{w}_{+}) - \hat{p}_{0,-}(\bm{w}_{-}) \ = \ \sum_{i=1}^{n_{+}}w_{i,+} \left( Y_{i,+} - \frac{1}{2} \right) - \sum_{i=1}^{n_{-}}w_{i,-} \left( Y_{i,-} - \frac{1}{2} \right) ,
\end{eqnarray*}
where $\bm{w}_{+} \equiv (w_{1,+}, \ldots, w_{n_{+},+})'$ and $\bm{w}_{-} \equiv (w_{1,-}, \ldots, w_{n_{-},-})'$. Similar to Section \ref{sec:main}, 
we assume that $\bm{w}_{+} \in \mathcal{W}_{+}$ and $\bm{w}_{-} \in \mathcal{W}_{-}$ hold, where
\begin{eqnarray*}
\mathcal{W}_{+} & \equiv & \left\{ \bm{w}_{+} \in \mathbb{R}^{n_{+}} : \sum_{i=1}^{n_{+}} w_{i,+} \leq 1 \ \text{and} \ w_{i,+} \geq 0 \ \text{for all $i$} \right\}, \\
\mathcal{W}_{-} & \equiv & \left\{ \bm{w}_{-} \in \mathbb{R}^{n_{-}} : \sum_{i=1}^{n_{-}} w_{i,-} \leq 1 \ \text{and} \ w_{i,-} \geq 0 \ \text{for all $i$} \right\}.
\end{eqnarray*}
Suppose that $Y_{i,+}$ and $Y_{i,-}$ follow Bernoulli distribution with parameters $p_{i,+}$ and $p_{i,-}$, respectively. Letting $\theta_{i,+} \equiv p_{i,+} - 1/2$, $\theta_{i,-} \equiv p_{i,-} - 1/2$, $\bm{\theta}_{+} \equiv (\theta_{0,+}, \ldots, \theta_{n_{+},+})'$, and $\bm{\theta}_{-} \equiv (\theta_{0,-}, \ldots, \theta_{n_{-},-})'$, we consider the following parameter spaces:
\begin{eqnarray*}
\Theta_{+} &\equiv & \left\{ \bm{\theta}_{+} \in [-1/2,1/2]^{n_{+}+1} : |\theta_{i,+} - \theta_{j,+} | \leq C \|R_{i,+} - R_{j,+}\| \ \text{for all $i$ and $j$} \right\}, \\
\Theta_{-} &\equiv & \left\{ \bm{\theta}_{-} \in [-1/2,1/2]^{n_{-}+1} : |\theta_{i,-} - \theta_{j,-} | \leq C \|R_{i,-} - R_{j,-}\| \ \text{for all $i$ and $j$} \right\},
\end{eqnarray*}
where $R_{0,+} = R_{0,-} = 0$. Similar to Section \ref{sec:main}, we assume $\|R_{0,+}\| \leq \|R_{1,+}\| \leq \cdots \leq \|R_{n_{+},+}\|$ and $\|R_{0,-}\| \leq \|R_{1,-}\| \leq \cdots \leq \|R_{n_{-},-}\|$.

Because the ATE is $\tau \equiv \theta_{0,+} - \theta_{0,-}$, the MSE of $\hat{\tau}(\bm{w}_{+}, \bm{w}_{-})$ can be written as follows:
\begin{eqnarray*}
& &\text{MSE}_{ate}(\bm{w}_{+},\bm{w}_{-}, \bm{\theta}_{+},\bm{\theta}_{-}) \ \equiv \ E\left[ \left\{ \hat{\tau}(\bm{w}_{+}, \bm{w}_{-}) - \tau \right\}^2 \right] \\
&=& E\left[ \left\{ \sum_{i=1}^{n_{+}}w_{i,+} \left( Y_{i,+} - \frac{1}{2} \right) - \theta_{0,+} \right\}^2 \right] + E\left[ \left\{ \sum_{i=1}^{n_{-}}w_{i,-} \left( Y_{i,-} - \frac{1}{2} \right) - \theta_{0,-} \right\}^2 \right] \\
& & \hspace{0.8in} - 2 E\left[ \left\{ \sum_{i=1}^{n_{+}}w_{i,+} \left( Y_{i,+} - \frac{1}{2} \right) - \theta_{0,+} \right\} \left\{ \sum_{i=1}^{n_{-}}w_{i,-} \left( Y_{i,-} - \frac{1}{2} \right) - \theta_{0,-} \right\}  \right] \\
&=& b_{+}(\bm{w}_{+},\bm{\theta}_{+})^2 + V_{+}(\bm{w}_{+},\bm{\theta}_{+}) \\
& & \hspace{0.5in} + b_{-}(\bm{w}_{-}, \bm{\theta}_{-})^2 + V_{-}(\bm{w}_{-}, \bm{\theta}_{-}) - 2 b_{+}(\bm{w}_{+},\bm{\theta}_{+}) b_{-}(\bm{w}_{-}, \bm{\theta}_{-}),
\end{eqnarray*}
where 
\begin{eqnarray*}
& & b_{+}(\bm{w}_{+},\bm{\theta}_{+}) \equiv \sum_{i=1}^{n_{+}}w_{i,+}\theta_{i,+} - \theta_{0,+}, \ \ \ \ V_{+}(\bm{w}_{+},\bm{\theta}_{+}) \equiv \sum_{i=1}^{n_{+}}w_{i,+}^2 \left( \frac{1}{4} - \theta_{i,+}^2 \right),  \\
& & b_{-}(\bm{w}_{-}, \bm{\theta}_{-}) \equiv \sum_{i=1}^{n_{-}}w_{i,-}\theta_{i,-} - \theta_{0,-}, \ \ \ \ V_{-}(\bm{w}_{-}, \bm{\theta}_{-}) \equiv \sum_{i=1}^{n_{-}}w_{i,-}^2 \left( \frac{1}{4} - \theta_{i,-}^2 \right).
\end{eqnarray*}
We define
\begin{eqnarray*}
\tilde{\bm{\theta}}_{+}(\theta_{0,+}) & \equiv & \left(\theta_{0,+}, \min\{\theta_{0,+} + C \|R_{1,+}\|, 1/2 \}, \ldots, \min\{\theta_{0,+} + C \|R_{n_{+},+}\|, 1/2 \} \right)', \\
\tilde{\bm{\theta}}_{-}(\theta_{0,-}) & \equiv & \left(\theta_{0,-}, \max\{\theta_{0,-} - C \|R_{1,-}\|, -1/2 \}, \ldots, \max\{\theta_{0,-} - C \|R_{n_{-},-}\|, -1/2 \} \right)'.
\end{eqnarray*}
Similar to Theorem \ref{thm:max_problem}, we obtain the following theorem.

\begin{theorem}\label{thm:max_ate}
For $\bm{w}_{+} \in \mathcal{W}_{+}$ and $\bm{w}_{-} \in \mathcal{W}_{-}$, we obtain
\begin{eqnarray}
& & \max_{\bm{\theta}_{+} \in \Theta_{+}, \bm{\theta}_{-} \in \Theta_{-}} \mathrm{MSE}_{ate}(\bm{w}_{+},\bm{w}_{-}, \bm{\theta}_{+},\bm{\theta}_{-}) \nonumber \\
&=& \max_{\theta_{0,+} \in [-1/2,0], \ \theta_{0,-} \in [0,1/2]} \mathrm{MSE}_{ate} \left( \bm{w}_{+},\bm{w}_{-}, \tilde{\bm{\theta}}_{+}(\theta_{0,+}),\tilde{\bm{\theta}}_{-}(\theta_{0,-}) \right). \label{max_ate}
\end{eqnarray}
\end{theorem}

\begin{proof}
Fix $\theta_{0,+} \in \Theta_{+}$ and $\theta_{0,-} \in \Theta_{-}$. Without loss of generality, we assume $\theta_{0,+} \leq 0$. Define
\begin{eqnarray*}
\bar{\bm{\theta}}_{+}(\theta_{0,+}) & \equiv & \left(\theta_{0,+}, \max\{\theta_{0,+} - C \|R_{1,+}\|, -1/2 \}, \ldots, \max\{\theta_{0,+} - C \|R_{n_{+},+}\|, -1/2 \} \right)', \\
\bar{\bm{\theta}}_{-}(\theta_{0,-}) & \equiv & \left(\theta_{0,-}, \min\{\theta_{0,-} + C \|R_{1,-}\|, 1/2 \}, \ldots, \min\{\theta_{0,-} + C \|R_{n_{-},-}\|, 1/2 \} \right)'.
\end{eqnarray*}

First, we consider the case where $\theta_{0,+} \in [-1/2,0]$ and $\theta_{0,-} \in [0,1/2]$. In this case, Theorem \ref{thm:max_problem} implies that $b_{+}(\bm{w}_{+},\bm{\theta}_{+})^2 + V_{+}(\bm{w}_{+},\bm{\theta}_{+})$ is maximized at $\bm{\theta}_{+} = \tilde{\bm{\theta}}_{+}(\theta_{0,+})$ and $b_{-}(\bm{w}_{-}, \bm{\theta}_{-})^2 + V_{-}(\bm{w}_{-}, \bm{\theta}_{-})$ is maximized at $\bm{\theta}_{-} = \tilde{\bm{\theta}}_{-}(\theta_{0,-})$. 
Note that
$$
b_{+}(\bm{w}_{+},\bm{\theta}_{+})=\sum_{i=1}^{n_{+}}w_{i,+}(\theta_{i,+} - \theta_{0,+})-\left(1-\sum_{i=1}^{n_{+}}w_{i,+}\right)\theta_{0,+}
$$
and
$$
b_{-}(\bm{w}_{-},\bm{\theta}_{-})=\sum_{i=1}^{n_{-}}w_{i,-}(\theta_{i,-} - \theta_{0,-})-\left(1-\sum_{i=1}^{n_{-}}w_{i,-}\right)\theta_{0,-}.
$$
Then $|b_{+}(\bm{w}_{+},\bm{\theta}_{+})|$ is maximized at $\bm{\theta}_{+} = \tilde{\bm{\theta}}_{+}(\theta_{0,+})$ or $\bar{\bm{\theta}}_{+}(\theta_{0,+})$ and $|b_{-}(\bm{w}_{-}, \bm{\theta}_{-})|$ is maximized at $\bm{\theta}_{-} = \tilde{\bm{\theta}}_{-}(\theta_{0,-})$ or $\bar{\bm{\theta}}_{-}(\theta_{0,-})$. Because $\bm{w}_{+} \in \mathcal{W}_{+}$ and $\theta_{0,+} \leq 0$, we have $|b_{+}(\bm{w}_{+}, \bar{\bm{\theta}}_{+}(\theta_{0,+}))| \leq |b_{+}(\bm{w}_{+}, \tilde{\bm{\theta}}_{+}(\theta_{0,+}))|=b_{+}(\bm{w}_{+}, \tilde{\bm{\theta}}_{+}(\theta_{0,+}))$. Similarly, we have $|b_{-}(\bm{w}_{-},\bar{\bm{\theta}}_{-}(\theta_{0,-}))| \leq |b_{-}(\bm{w}_{-}, \tilde{\bm{\theta}}_{-}(\theta_{0,-}))|=-b_{-}(\bm{w}_{-}, \tilde{\bm{\theta}}_{-}(\theta_{0,-}))$. These results imply that $-b_{+}(\bm{w}_{+},\bm{\theta}_{+}) b_{-}(\bm{w}_{-}, \bm{\theta}_{-})$ is maximized at $\bm{\theta}_{+} = \tilde{\bm{\theta}}_{+}(\theta_{0,+})$ and $\bm{\theta}_{-} = \tilde{\bm{\theta}}_{-}(\theta_{0,-})$. Therefore, if $\theta_{0,+} \in [-1/2,0]$ and $\theta_{0,-} \in [0,1/2]$ hold, we obtain
\begin{equation*}
\mathrm{MSE}_{ate}(\bm{w}_{+},\bm{w}_{-},\bm{\theta}_{+},\bm{\theta}_{-}) \ \leq \ \mathrm{MSE}_{ate} \left( \bm{w}_{+},\bm{w}_{-}, \tilde{\bm{\theta}}_{+}(\theta_{0,+}),\tilde{\bm{\theta}}_{-}(\theta_{0,-}) \right). 
\end{equation*}

Next, we consider the case where $\theta_{0,+} \in [-1/2,0]$ and $\theta_{0,-} \in [-1/2,0)$. From Theorem \ref{thm:max_problem}, we have
\begin{eqnarray*}
& & \mathrm{MSE}_{ate}(\bm{w}_{+},\bm{w}_{-},\bm{\theta}_{+},\bm{\theta}_{-}) \\
& \leq & b_{+}(\bm{w}_{+}, \tilde{\bm{\theta}}_{+}(\theta_{0,+}))^2 + V_{+} (\bm{w}_{+}, \tilde{\bm{\theta}}_{+}(\theta_{0,+})) \\
& & \hspace{0.2in} + b_{-}(\bm{w}_{-}, \bar{\bm{\theta}}_{-}(\theta_{0,-}))^2 + V_{-} (\bm{w}_{-}, \bar{\bm{\theta}}_{-}(\theta_{0,-})) - 2 b_{+}(\bm{w}_{+}, \bm{\theta}_{+}) b_{-}(\bm{w}_{-}, \bm{\theta}_{-}).
\end{eqnarray*}
Moreover, since $\theta_{0,+} \le 0 $ and $\theta_{0,-} \le 0$, it follows from an argument similar to the one above that $|b_{+}(\bm{w}_{+}, \bm{\theta}_{+})|\le |b_{+}(\bm{w}_{+}, \tilde{\bm{\theta}}_{+}(\theta_{0,+}))|=b_{+}(\bm{w}_{+}, \tilde{\bm{\theta}}_{+}(\theta_{0,+}))$ and $|b_{-}(\bm{w}_{-}, \bm{\theta}_{-})|\le |b_{-}(\bm{w}_{-}, \bar{\bm{\theta}}_{-}(\theta_{0,-}))|=b_{-}(\bm{w}_{-}, \bar{\bm{\theta}}_{-}(\theta_{0,-}))$.
Hence,
\begin{eqnarray*}
& & \mathrm{MSE}_{ate}(\bm{w}_{+},\bm{w}_{-},\bm{\theta}_{+},\bm{\theta}_{-}) \\
& \leq & b_{+}(\bm{w}_{+}, \tilde{\bm{\theta}}_{+}(\theta_{0,+}))^2 + V_{+} (\bm{w}_{+}, \tilde{\bm{\theta}}_{+}(\theta_{0,+})) \\
& & \hspace{0.2in} + b_{-}(\bm{w}_{-}, \bar{\bm{\theta}}_{-}(\theta_{0,-}))^2 + V_{-} (\bm{w}_{-}, \bar{\bm{\theta}}_{-}(\theta_{0,-})) + 2 b_{+}(\bm{w}_{+}, \tilde{\bm{\theta}}_{+}(\theta_{0,+})) b_{-}(\bm{w}_{-}, \bar{\bm{\theta}}_{-}(\theta_{0,-})).
\end{eqnarray*}
Because $b_{-}(\bm{w}_{-},\bm{\theta}_{-}) = - b_{-}(\bm{w}_{-},-\bm{\theta}_{-})$, $V_{-}(\bm{w}_{-},\bm{\theta}_{-}) = V_{-}(\bm{w}_{-},-\bm{\theta}_{-})$, and $-\bar{\bm{\theta}}_{-}(\theta_{0,-})=\tilde{\bm{\theta}}_{-}(-\theta_{0,-})$, we obtain
\begin{eqnarray*}
& & \mathrm{MSE}_{ate}(\bm{w}_{+},\bm{w}_{-},\bm{\theta}_{+},\bm{\theta}_{-}) \\
& \leq & b_{+}(\bm{w}_{+}, \tilde{\bm{\theta}}_{+}(\theta_{0,+}))^2 + V_{+} (\bm{w}_{+}, \tilde{\bm{\theta}}_{+}(\theta_{0,+})) + b_{-}(\bm{w}_{-}, -\bar{\bm{\theta}}_{-}(\theta_{0,-}))^2 + V_{-} (\bm{w}_{-}, -\bar{\bm{\theta}}_{-}(\theta_{0,-})) \\
& & \hspace{0.2in} - 2 b_{+}(\bm{w}_{+}, \tilde{\bm{\theta}}_{+}(\theta_{0,+})) b_{-}(\bm{w}_{-}, -\bar{\bm{\theta}}_{-}(\theta_{0,-}))\\
&=& \mathrm{MSE}_{ate} \left( \bm{w}_{+},\bm{w}_{-}, \tilde{\bm{\theta}}_{+}(\theta_{0,+}),-\bar{\bm{\theta}}_{-}(\theta_{0,-}) \right)\\
&=& \mathrm{MSE}_{ate} \left( \bm{w}_{+},\bm{w}_{-}, \tilde{\bm{\theta}}_{+}(\theta_{0,+}),\tilde{\bm{\theta}}_{-}(-\theta_{0,-}) \right).
\end{eqnarray*}
Hence, it is enough to consider the maximization of $\mathrm{MSE}_{ate} \left( \bm{w}_{+},\bm{w}_{-}, \tilde{\bm{\theta}}_{+}(\theta_{0,+}),\tilde{\bm{\theta}}_{-}(\theta_{0,-}) \right)$ over $\theta_{0,+} \in [-1/2,0]$ and $\theta_{0,-} \in [0,1/2]$. As a result, we obtain (\ref{max_ate}).
\end{proof}

Next, we derive the weight vector $(\bm{w}_{+}',\bm{w}_{-}')'$ that minimizes the maximum MSE. Similar to Lemmas \ref{lem:monotone} and \ref{lem:zero_weight}, we obtain the following lemma.

\begin{lemma}
We obtain
\begin{eqnarray}
& &  \min_{\bm{w}_{+} \in \mathcal{W}_{+}, \bm{w}_{-} \in \mathcal{W}_{-}} \ \max_{\bm{\theta}_{+} \in \Theta_{+}, \bm{\theta}_{-} \in \Theta_{-}} \mathrm{MSE}_{ate}(\bm{w}_{+},\bm{w}_{-},\bm{\theta}_{+},\bm{\theta}_{-}) \nonumber \\
& & \hspace{0.8in} = \  \min_{\bm{w}_{+} \in \mathcal{W}_{+}^1, \bm{w}_{-} \in \mathcal{W}_{-}^1} \ \max_{\bm{\theta}_{+} \in \Theta_{+}, \bm{\theta}_{-} \in \Theta_{-}} \mathrm{MSE}_{ate}(\bm{w}_{+},\bm{w}_{-},\bm{\theta}_{+},\bm{\theta}_{-}), \label{monotone_zero_ate}
\end{eqnarray}
where
\begin{eqnarray*}
\mathcal{W}_{+}^1 & \equiv & \left\{ \bm{w}_{+} \in \mathcal{W}_{+} : w_{1,+} \geq \cdots \geq w_{n_{+},+} \ \text{and} \ w_{i,+} = 0 \ \text{if $C \| R_{i,+} \| \geq 1/2$} \right\}, \\
\mathcal{W}_{-}^1 & \equiv & \left\{ \bm{w}_{-} \in \mathcal{W}_{-} : w_{1,-} \geq \cdots \geq w_{n_{-},-} \ \text{and} \ w_{i,-} = 0 \ \text{if $C \| R_{i,-} \| \geq 1/2$} \right\}.
\end{eqnarray*}
\end{lemma}

\begin{proof}
By Theorem \ref{thm:max_ate}, it suffices to consider the MSE at $(\tilde{\bm{\theta}}_{+}(\theta_{0,+}),\tilde{\bm{\theta}}_{-}(\theta_{0,-}))$ for $\theta_{0,+} \in [-1/2,0]$ and $\theta_{0,-} \in [0,1/2]$.
Suppose that $\bm{w}_{+} \equiv (w_{1,+} , \ldots, w_{n_{+},+})' \in \mathcal{W}_{+}$ satisfies $w_{j,+} < w_{j+1, +}$ for some $j$. Letting $\tilde{\bm{w}}_{+} \equiv (w_{1,+}, \ldots, w_{j-1,+}, w_{j+1,+}, w_{j,+}, w_{j+2,+}, \ldots, w_{n_{+},+})'$, we have $\tilde{\bm{w}}_{+} \in \mathcal{W}_{+}$.
From the proof of Lemma \ref{lem:monotone}, for any $\theta_{0,+} \in [-1/2,0]$,
\[
b_{+}(\bm{w}_{+}, \tilde{\bm{\theta}}_{+}(\theta_{0,+}))^2 + V_{+} (\bm{w}_{+}, \tilde{\bm{\theta}}_{+}(\theta_{0,+})) \ \geq \ b_{+}(\tilde{\bm{w}}_{+}, \tilde{\bm{\theta}}_{+}(\theta_{0,+}))^2 + V_{+} (\tilde{\bm{w}}_{+}, \tilde{\bm{\theta}}_{+}(\theta_{0,+})).
\]
In addition, noting that
$$
b_{+}(\bm{w}_{+}, \bm{\theta}_{+}) - b_{+}(\tilde{\bm{w}}_{+}, \bm{\theta}_{+})=(w_{j,+}-w_{j+1,+})(\theta_{j,+}-\theta_{j+1,+})
$$
and
$$
b_{+}(\tilde{\bm{w}}_{+},\bm{\theta}_{+})=\sum_{i=1}^{n_{+}}\tilde w_{i,+}(\theta_{i,+} - \theta_{0,+})-\left(1-\sum_{i=1}^{n_{+}}\tilde w_{i,+}\right)\theta_{0,+},
$$
we have
$b_{+}(\bm{w}_{+}, \tilde{\bm{\theta}}_{+}(\theta_{0,+}))\ge b_{+}(\tilde{\bm{w}}_{+}, \tilde{\bm{\theta}}_{+}(\theta_{0,+}))\ge 0$ for any $\theta_{0,+} \in [-1/2,0]$.
Also, $b_{-}(\bm{w}_{-}, \tilde{\bm{\theta}}_{-}(\theta_{0,-}))\le 0$ for any $\theta_{0,-} \in [0,1/2]$.
Hence, we have
\[
\mathrm{MSE}_{ate}\left( \bm{w}_{+},\bm{w}_{-},\tilde{\bm{\theta}}_{+}(\theta_{0,+}),\tilde{\bm{\theta}}_{-}(\theta_{0,-}) \right) \ \geq \ \mathrm{MSE}_{ate}\left( \tilde{\bm{w}}_{+},\bm{w}_{-},\tilde{\bm{\theta}}_{+}(\theta_{0,+}),\tilde{\bm{\theta}}_{-}(\theta_{0,-}) \right).
\]
Therefore, if $w_{j,+} < w_{j+1,+}$, then we can reduce the maximum MSE by exchanging $w_{j,+}$ for $w_{j+1,+}$. Similar argument holds for $\bm{w}_{-}$. Therefore, for any $\bm{w}_{+} \in \mathcal{W}_{+}$ and $\bm{w}_{-} \in \mathcal{W}_{-}$, there exists $(\tilde{\bm{w}}_{+}, \tilde{\bm{w}}_{-}) \in \mathcal{W}_{+} \times \mathcal{W}_{-}$ such that $\tilde{w}_{1,+} \geq \cdots \geq \tilde{w}_{n_{+},+}$, $\tilde{w}_{1,-} \geq \cdots \geq \tilde{w}_{n_{-},-}$, and
\[
\mathrm{MSE}_{ate}\left( \bm{w}_{+},\bm{w}_{-},\tilde{\bm{\theta}}_{+}(\theta_{0,+}),\tilde{\bm{\theta}}_{-}(\theta_{0,-}) \right) \ \geq \ \mathrm{MSE}_{ate}\left( \tilde{\bm{w}}_{+},\tilde{\bm{w}}_{-},\tilde{\bm{\theta}}_{+}(\theta_{0,+}),\tilde{\bm{\theta}}_{-}(\theta_{0,-}) \right).
\]

Next, observe that
\begin{eqnarray*}
& & \frac{\partial}{\partial w_{j,+}} \mathrm{MSE}_{ate}(\bm{w}_{+},\bm{w}_{-},\bm{\theta}_{+},\bm{\theta}_{-}) \\ 
&=& 2 \theta_{j,+} \left( \sum_{i=1}^{n_{+}} w_{i,+} \theta_{i,+} - \theta_{0,+} \right) + 2 w_{j,+} \left( 1/4 - \theta_{j,+}^2 \right) -2 \theta_{j,+} b_{-}(\bm{w}_{-}, \bm{\theta}_{-}) \\
&=& 2 \theta_{j,+} \left\{ \left( \sum_{i \neq j} w_{i,+} \theta_{i,+} - \theta_{0,+} \right) - b_{-}(\bm{w}_{-}, \bm{\theta}_{-}) \right\} + w_{j,+}/2\\
&=& 2 \theta_{j,+} \left\{ \left( \sum_{i \neq j} w_{i,+} (\theta_{i,+} - \theta_{0,+})-\left(1- \sum_{i \neq j} w_{i,+}\right)\theta_{0,+} \right) - b_{-}(\bm{w}_{-}, \bm{\theta}_{-}) \right\} + w_{j,+}/2.
\end{eqnarray*}
If $C\|R_{j,+}\| \geq 1/2$, then $j$-th element of $\tilde{\bm{\theta}}_{+}(\theta_{0,+})$ is nonnegative for any $\theta_{0,+} \in [-1/2,0]$. Because $b_{-}(\bm{w}_{-}, \tilde{\bm{\theta}}_{-}(\theta_{0,-})) \leq 0$ for any $\theta_{0,-} \in [0,1/2]$, we obtain
\begin{eqnarray*}
& & \frac{\partial}{\partial w_{j,+}} \mathrm{MSE}_{ate} \left( \bm{w}_{+},\bm{w}_{-}, \tilde{\bm{\theta}}_{+}(\theta_{0,+}),\tilde{\bm{\theta}}_{-}(\theta_{0,-}) \right) \ \geq \ 0 \\ 
& & \hspace{1in} \text{for any $\bm{w}_{+} \in \mathcal{W}_{+}$, $\bm{w}_{-} \in \mathcal{W}_{-}$, $\theta_{0,+} \in [-1/2,0]$, and $\theta_{0,-} \in [0,1/2]$.}
\end{eqnarray*}
Hence, if $C\|R_{j,+}\| \geq 1/2$, we can reduce the maximum MSE by replacing $w_{j,+}$ with $0$. Similarly, if $C\|R_{j,-}\| \geq 1/2$, we can reduce the maximum MSE by replacing $w_{j,-}$ with $0$. As a result, we obtain (\ref{monotone_zero_ate}).
\end{proof}

We now present how one can numerically solve the minimax problem
\[
\min_{\bm{w}_{+} \in \mathcal{W}_{+}, \bm{w}_{-} \in \mathcal{W}_{-}} \ \max_{\bm{\theta}_{+} \in \Theta_{+}, \bm{\theta}_{-} \in \Theta_{-}} \mathrm{MSE}_{ate}(\bm{w}_{+},\bm{w}_{-},\bm{\theta}_{+},\bm{\theta}_{-}).
\]
The MSE of $\hat{\tau}(\bm{w}_{+},\bm{w}_{-})$ can be written as
\begin{eqnarray*}
\mathrm{MSE}_{ate}(\bm{w}_{+},\bm{w}_{-},\bm{\theta}_{+},\bm{\theta}_{-}) &=& \left\{ b_{+}(\bm{w}_{+},\bm{\theta}_{+}) - b_{-}(\bm{w}_{-}, \bm{\theta}_{-}) \right\}^2 \\
& & \hspace{0.8in} + V_{+}(\bm{w}_{+},\bm{\theta}_{+}) + V_{-}(\bm{w}_{-}, \bm{\theta}_{-}),
\end{eqnarray*}
where both $\left\{ b_{+}(\bm{w}_{+},\bm{\theta}_{+}) - b_{-}(\bm{w}_{-}, \bm{\theta}_{-}) \right\}^2$ and $V_{+}(\bm{w}_{+},\bm{\theta}_{+}) + V_{-}(\bm{w}_{-}, \bm{\theta}_{-})$ are convex with respect to $\bm{w} \equiv (\bm{w}_{+},\bm{w}_{-})$. This implies that $\mathrm{MSE}_{ate}(\bm{w}_{+},\bm{w}_{-},\bm{\theta}_{+},\bm{\theta}_{-})$ is convex with respect to $\bm{w}$ for all $\bm{\theta}_{+} \in \Theta_{+}$ and $\bm{\theta}_{-} \in \Theta_{-}$. We define
\begin{eqnarray*}
g(\bm{w}; \theta_{0,+}, \theta_{0,-}) & \equiv & \mathrm{MSE}_{ate} \left( \bm{w}_{+},\bm{w}_{-}, \tilde{\bm{\theta}}_{+}(\theta_{0,+}),\tilde{\bm{\theta}}_{-}(\theta_{0,-}) \right), \\
\overline{g}(\bm{w}) & \equiv & \max_{\theta_{0,+} \in [-1/2,0], \ \theta_{0,-} \in [0,1/2]} g(\bm{w}; \theta_{0,+}, \theta_{0,-}).
\end{eqnarray*}
Because $g(\bm{w}; \theta_{0,+}, \theta_{0,-})$ is convex with respect to $\bm{w}$ for all $\theta_{0,+}$ and $\theta_{0,-}$, $\overline{g}(\bm{w})$ is also convex. We can solve the minimax problem by minimizing $\overline{g}(\bm{w})$ subject to $\bm{w} \in \mathcal{W}_{+}^1 \times \mathcal{W}_{-}^1$.

\section{Confidence intervals with general bounded outcomes}\label{sec:inf_bounded}

\subsection{One-sided confidence interval}

Suppose that $Y_{i,+}\in [0,1]$ for $i=1,\ldots,n_+$ and $Y_{i,-}\in [0,1]$ for $i=1,\ldots,n_-$. We keep assuming that the observed outcomes are independent.
Let $\mathcal{Q}(\bm{p})$ denote the set of distributions of $(Y_{1,+},\ldots,Y_{n_+,+},Y_{1,-},\ldots,Y_{n_-,-})\in [0,1]^{n_++n_-}$ such that $E[Y_{i,+}]=p_{i,+}$ for $i=1,\ldots,n_+$ and $E[Y_{i,-}]=p_{i,-}$ for $i=1,\ldots,n_-$.

We consider a one-sided $100\cdot (1-\alpha)\%$ CI 
$[\hat\tau-\gamma,\infty)$ satisfying
\begin{equation*}
    \inf_{\bm{p} \in \mathcal{P}, Q\in \mathcal{Q}(\bm{p})} P_{\bm{p},Q} \left( \tau \in [\hat\tau-\gamma,\infty) \right) \ \geq \ 1-\alpha, \ \text{ or } \ \sup_{\bm{p} \in \mathcal{P}, Q\in \mathcal{Q}(\bm{p})} P_{\bm{p},Q} \left( \hat\tau - \tau > \gamma \right) \ \leq \ \alpha,
\end{equation*}
where $\mathcal{P}\equiv \mathcal{P}_+\times \mathcal{P}_-$.

We construct a CI using an upper bound on $\sup_{\bm{p} \in \mathcal{P}, Q\in \mathcal{Q}(\bm{p})} P_{\bm{p},Q} \left( \hat\tau - \tau > \gamma \right)$.
Define
$$
{\rm Bias}_{\bm{p}}(\hat\tau)\equiv E_{\bm{p}}[\hat\tau]-\tau=\sum_{i=1}^{n_{+}}w_{i,+} \left( p_{i,+} - \frac{1}{2} \right) - \sum_{i=1}^{n_{-}}w_{i,-} \left( p_{i,-} - \frac{1}{2} \right)-\tau
$$
and
$\overline{{\rm Bias}}_{\mathcal{P}}(\hat\tau)\equiv\max_{\bm{p} \in \mathcal{P}}{\rm Bias}_{\bm{p}}(\hat\tau)$.
First, fix $\bm{p} \in \mathcal{P}$ and $Q\in \mathcal{Q}(\bm{p})$.
For any $\gamma>{\rm Bias}_{\bm{p}}(\hat\tau)$,
\begin{eqnarray*}
    & & P_{\bm{p},Q} \left( \hat\tau - \tau > \gamma \right) \\
    &=& P_{\bm{p},Q} \left( \hat\tau - E_{\bm{p}}[\hat\tau] > \gamma - {\rm Bias}_{\bm{p}}(\hat\tau) \right) \\
    &=& P_{\bm{p},Q} \left( \sum_{i=1}^{n_{+}}w_{i,+} Y_{i,+} + \sum_{i=1}^{n_{-}}(-w_{i,-} Y_{i,-}) - E_{\bm{p}}\left[\sum_{i=1}^{n_{+}}w_{i,+} Y_{i,+} + \sum_{i=1}^{n_{-}}(-w_{i,-} Y_{i,-})\right]> \gamma - {\rm Bias}_{\bm{p}}(\hat\tau) \right) \\
    & \leq & \exp\left(-\frac{2(\gamma - {\rm Bias}_{\bm{p}}(\hat\tau))^2}{\sum_{i=1}^{n_+}w_{i,+}^2+\sum_{i=1}^{n_-}w_{i,-}^2}\right),
\end{eqnarray*}
where the last inequality is obtained by Hoeffding's inequality since $w_{i,+}Y_{i,+}\in [0,w_{i,+}]$ and $-w_{i,-}Y_{i,-}\in [-w_{i,-},0]$.
It follows that for any $\gamma>\overline{{\rm Bias}}_{\mathcal{P}}(\hat\tau)$,
$$
    \sup_{\bm{p} \in \mathcal{P}, Q\in \mathcal{Q}(\bm{p})} P_{\bm{p},Q} \left( \hat\tau - \tau > \gamma \right) 
    \ \leq \ \exp\left(-\frac{2(\gamma - \overline{{\rm Bias}}_{\mathcal{P}}(\hat\tau))^2}{\sum_{i=1}^{n_+}w_{i,+}^2+\sum_{i=1}^{n_-}w_{i,-}^2}\right).
$$
Solving
$$
\exp\left(-\frac{2(\gamma - \overline{{\rm Bias}}_{\mathcal{P}}(\hat\tau))^2}{\sum_{i=1}^{n_+}w_{i,+}^2+\sum_{i=1}^{n_-}w_{i,-}^2}\right)=\alpha
$$
yields
$$
\gamma^*=\overline{{\rm Bias}}_{\mathcal{P}}(\hat\tau)+\left(\frac{\log(1/\alpha)\left(\sum_{i=1}^{n_+}w_{i,+}^2+\sum_{i=1}^{n_-}w_{i,-}^2\right)}{2}\right)^{1/2}.
$$
Under the Lipschitz constraint, $\overline{{\rm Bias}}_{\mathcal{P}}(\hat\tau)={\rm Bias}_{(\tilde{\bm{p}}_{+}(0)',\tilde{\bm{p}}_{-}(1)')'}(\hat\tau)$, where
\begin{eqnarray*}
    \tilde{\bm{p}}_{+}(0) &\equiv & \left( 0, \min\{C\|R_{1,+}\|, 1 \}, \ldots , \min\{C\|R_{n_{+},+}\|, 1 \} \right)', \\
    \tilde{\bm{p}}_{-}(1) &\equiv & \left( 1, \max\{1-C\|R_{1,-}\|, 0 \}, \ldots , \max\{1-C\|R_{n_{-},-}\|, 0 \} \right)'.
\end{eqnarray*}
Therefore, $\gamma^*$ has a closed-form expression.

\subsection{Two-sided confidence interval}

We consider a two-sided $100\cdot (1-\alpha)\%$ CI 
$[\hat\tau-\gamma,\hat\tau+\gamma]$ satisfying
\begin{equation*}
    \inf_{\bm{p} \in \mathcal{P}, Q\in \mathcal{Q}(\bm{p})} P_{\bm{p},Q} \left( \tau \in [\hat\tau-\gamma,\hat\tau+\gamma] \right) \ \geq \ 1-\alpha, \ \text{ or } \ \sup_{\bm{p} \in \mathcal{P}, Q\in \mathcal{Q}(\bm{p})} P_{\bm{p},Q} \left( |\hat\tau - \tau| > \gamma \right) \ \leq \ \alpha.
\end{equation*}
By symmetry of $\mathcal{P}$ with respect to $\bm{p}=(1/2,\ldots,1/2)'$, the minimum bias $\min_{\bm{p} \in \mathcal{P}}{\rm Bias}_{\bm{p}}(\hat\tau)$ is given by $-\overline{{\rm Bias}}_{\mathcal{P}}(\hat\tau)$.
The result from the previous subsection implies that
\begin{equation}
    \sup_{\bm{p} \in \mathcal{P}, Q\in \mathcal{Q}(\bm{p})} P_{\bm{p},Q} \left( |\hat\tau - \tau| > \gamma \right)
    \ \leq \ 2\exp\left(-\frac{2(\gamma - \overline{{\rm Bias}}_{\mathcal{P}}(\hat\tau))^2}{\sum_{i=1}^{n_+}w_{i,+}^2+\sum_{i=1}^{n_-}w_{i,-}^2}\right). \label{eq:hoeffding_two_sided}
\end{equation}
Setting $\gamma=\overline{{\rm Bias}}_{\mathcal{P}}(\hat\tau)+\left(\frac{\log(1/(\alpha/2))\left(\sum_{i=1}^{n_+}w_{i,+}^2+\sum_{i=1}^{n_-}w_{i,-}^2\right)}{2}\right)^{1/2}$ leads to a valid CI.
This CI is computationally attractive, but it can be too conservative since the bias cannot be equal to $\overline{{\rm Bias}}_{\mathcal{P}}(\hat\tau)$ and $-\overline{{\rm Bias}}_{\mathcal{P}}(\hat\tau)$ at once.

To construct a less conservative CI, observe that for any $\gamma>|{\rm Bias}_{\bm{p}}(\hat\tau)|$,
\begin{eqnarray*}
    & & P_{\bm{p},Q} \left( |\hat\tau - \tau| > \gamma \right) \\
    &=& P_{\bm{p},Q} \left( |\hat\tau - E_{\bm{p}}[\hat\tau] + {\rm Bias}_{\bm{p}}(\hat\tau)| > \gamma  \right) \\
    &=& P_{\bm{p},Q} \left( \hat\tau - E_{\bm{p}}[\hat\tau] + {\rm Bias}_{\bm{p}}(\hat\tau)> \gamma  \right) + P_{\bm{p},Q} \left( \hat\tau - E_{\bm{p}}[\hat\tau] + {\rm Bias}_{\bm{p}}(\hat\tau) < - \gamma  \right) \\
    &=& P_{\bm{p},Q} \left( \hat\tau - E_{\bm{p}}[\hat\tau] > \gamma -{\rm Bias}_{\bm{p}}(\hat\tau)\right) + P_{\bm{p},Q} \left( -\hat\tau + E_{\bm{p}}[\hat\tau] >  \gamma + {\rm Bias}_{\bm{p}}(\hat\tau) \right) \\
    & \leq & \exp\left(-\frac{2(\gamma - {\rm Bias}_{\bm{p}}(\hat\tau))^2}{\sum_{i=1}^{n_+}w_{i,+}^2+\sum_{i=1}^{n_-}w_{i,-}^2}\right) + \exp\left(-\frac{2(\gamma + {\rm Bias}_{\bm{p}}(\hat\tau))^2}{\sum_{i=1}^{n_+}w_{i,+}^2+\sum_{i=1}^{n_-}w_{i,-}^2}\right),
\end{eqnarray*}
where the last inequality is obtained by Hoeffding's inequality.
It follows that for any $\gamma>\overline{{\rm Bias}}_{\mathcal{P}}(\hat\tau)$,
\begin{eqnarray*}
    & & \sup_{\bm{p} \in \mathcal{P}, Q\in \mathcal{Q}(\bm{p})} P_{\bm{p},Q} \left( |\hat\tau - \tau| > \gamma \right) \\
    & \leq & \overline{\pi}(\gamma)\equiv \max_{b\in[0,\overline{{\rm Bias}}_{\mathcal{P}}(\hat\tau)]}\left[\exp\left(-\frac{2(\gamma - b)^2}{\sum_{i=1}^{n_+}w_{i,+}^2+\sum_{i=1}^{n_-}w_{i,-}^2}\right) + \exp\left(-\frac{2(\gamma + b)^2}{\sum_{i=1}^{n_+}w_{i,+}^2+\sum_{i=1}^{n_-}w_{i,-}^2}\right)\right].
\end{eqnarray*}
We propose using
$$
\gamma^*=\inf\{\gamma>\overline{{\rm Bias}}_{\mathcal{P}}(\hat\tau): \overline{\pi}(\gamma)\le \alpha\}.
$$

Note that $\overline{\pi}(\gamma)$ is a tighter upper bound than the bound given in (\ref{eq:hoeffding_two_sided}), since
\begin{align*}
    \overline{\pi}(\gamma) &\le \max_{b\in[0,\overline{{\rm Bias}}_{\mathcal{P}}(\hat\tau)]}\exp\left(-\frac{2(\gamma - b)^2}{\sum_{i=1}^{n_+}w_{i,+}^2+\sum_{i=1}^{n_-}w_{i,-}^2}\right)\\
    & \quad \quad +\max_{b\in[0,\overline{{\rm Bias}}_{\mathcal{P}}(\hat\tau)]}\exp\left(-\frac{2(\gamma + b)^2}{\sum_{i=1}^{n_+}w_{i,+}^2+\sum_{i=1}^{n_-}w_{i,-}^2}\right)\\
    &=\exp\left(-\frac{2(\gamma - \overline{{\rm Bias}}_{\mathcal{P}}(\hat\tau))^2}{\sum_{i=1}^{n_+}w_{i,+}^2+\sum_{i=1}^{n_-}w_{i,-}^2}\right)+\exp\left(-\frac{2\gamma^2}{\sum_{i=1}^{n_+}w_{i,+}^2+\sum_{i=1}^{n_-}w_{i,-}^2}\right)\\
    &\le 2\exp\left(-\frac{2(\gamma - \overline{{\rm Bias}}_{\mathcal{P}}(\hat\tau))^2}{\sum_{i=1}^{n_+}w_{i,+}^2+\sum_{i=1}^{n_-}w_{i,-}^2}\right).
\end{align*}
This implies that $\gamma^*\le\overline{{\rm Bias}}_{\mathcal{P}}(\hat\tau)+\left(\frac{\log(1/(\alpha/2))\left(\sum_{i=1}^{n_+}w_{i,+}^2+\sum_{i=1}^{n_-}w_{i,-}^2\right)}{2}\right)^{1/2}$, which leads to a less conservative CI.

\section{Optimal Weights in Gaussian Models}\label{sec:gauss_weights}
We use \cite{donoho1994}'s results to derive optimal weights that solve the minimax problem (\ref{minimax_normal}) and show that the weights satisfy $\sum_{i=1}^nw_i=1$ and $w_i\ge 0$ for all $i$.
See \cite{Armstrong2021ATE} for an application of \cite{donoho1994}'s results in a related setting.
Our Gaussian setting falls into the framework of \cite{donoho1994}.
Specifically, in the notation of \cite{donoho1994}, we observe $\mathbf{y}$ of the form $\mathbf{y}=K\mathbf{x}+\mathbf{z}$ with $\mathbf{x}\in \mathbf{X}$.
Here, $\mathbf{y}=(Y_1/\sigma_1,\ldots,Y_n/\sigma_n)'$, $\mathbf{z}\sim N(0,I_n)$, where $I_n$ is an $n\times n$ identity matrix, $\mathbf{x}=\bm\theta$, $\mathbf{X}=\Theta_g$, and $K\mathbf{x}=(\theta_1/\sigma_1,\ldots,\theta_n/\sigma_n)'$.
The parameter of interest is the linear functional $L\mathbf{x}=\theta_0$.
We derive an affine estimator that minimizes the maximum MSE among all affine estimators (i.e, estimators of form $\hat\theta_0=c+\bm w'\mathbf{y}=c+\sum_{i=1}^n(w_i/\sigma_i)Y_i$ with $c\in\mathbb{R}$ and $\bm w\in\mathbb{R}^n$).

To specialize the results in \cite{donoho1994} to our setting, define the modulus of continuity of $L$:
$$
\omega(\varepsilon)\equiv\sup_{\bm\theta,\bm{\tilde\theta}\in\Theta_g}\left\{L\bm\theta-L\tilde{\bm\theta}:\|K\bm\theta-K\tilde{\bm\theta}\|_2\le \varepsilon\right\}=\sup_{\bm\theta,\bm{\tilde\theta}\in\Theta_g}\left\{\theta_0-\tilde\theta_0:\sum_{i=1}^n\frac{(\theta_i-\tilde\theta_i)^2}{\sigma_i^2}\le \varepsilon^2\right\},
$$
where $\|\cdot\|_2$ is the Euclidean norm on $\mathbb{R}^n$.
Since $\Theta_g$ is convex and centrosymmetric, for any $(\bm\theta,\tilde{\bm\theta})\in \Theta_g\times\Theta_g$, there exists $(\bar{\bm\theta},-\bar{\bm\theta})\in \Theta_g\times\Theta_g$ such that $\bar{\bm\theta}-(-\bar{\bm\theta})=\bm\theta-\tilde{\bm\theta}$ (specifically, set $\bar{\bm\theta}=\frac{1}{2}(\bm\theta-\tilde{\bm\theta})$).
Therefore, the supremum $\omega(\varepsilon)$ is attained at a symmetric pair $(\bm\theta,\tilde{\bm\theta})=(\bm\theta_\varepsilon,-\bm\theta_\varepsilon)$, where $\bm\theta_\varepsilon=(\theta_{\varepsilon,0},\theta_{\varepsilon,1},\ldots,\theta_{\varepsilon,n})'$ solves
\begin{align}
\max_{\bm\theta\in\Theta_g}\left\{2\theta_0:\sum_{i=1}^n\frac{\theta_i^2}{\sigma_i^2}\le \frac{\varepsilon^2}{4}\right\},\label{modulus}
\end{align}
provided that this problem has a solution.
Indeed, it has a solution since the constrained set $\{\bm\theta\in\Theta_g:\sum_{i=1}^n\theta_i^2/\sigma_i^2\le \varepsilon^2/4\}$ is bounded and closed.
Later, in Lemma \ref{lem:modulus}, we will show that $\bm\theta_\varepsilon$ satisfies the inequality constraint with equality (i.e., $\sum_{i=1}^n\theta_{\varepsilon,i}^2/\sigma_i^2= \varepsilon^2/4$) and $\theta_{\varepsilon,i}\ge 0$ for all $i=1,\ldots,n$.
We will also show that $\omega(\cdot)$ is differentiable at $\varepsilon>0$ with $\omega'(\varepsilon)=\frac{\varepsilon}{2\sum_{i=1}^n\theta_{\varepsilon,i}/\sigma_i^2}$.

The results in \cite{donoho1994} (in particular, the arguments in the proof of Theorem 1) then yield the following result.
Let $\varepsilon_0>0$ be a solution to $\frac{(\varepsilon/2)^2}{(\varepsilon/2)^2+1}=\frac{\varepsilon\omega'(\varepsilon)}{\omega(\epsilon)}$ and let $\bm\theta_{\varepsilon_0}$ solve (\ref{modulus}) at $\varepsilon=\varepsilon_0$.
Then, the following estimator minimizes the maximum MSE among all affine estimators:
$$
\hat\theta_0=\omega'(\varepsilon_0)\left(\frac{K\bm\theta_{\varepsilon_0}-K(-\bm\theta_{\varepsilon_0})}{\|K\bm\theta_{\varepsilon_0}-K(-\bm\theta_{\varepsilon_0})\|_2}\right)'\mathbf{y}=\omega'(\varepsilon_0)\frac{\sum_{i=1}^n\theta_{\varepsilon_0,i}Y_i/\sigma_i^2}{\sqrt{\sum_{i=1}^n\theta_{\varepsilon_0,i}^2/\sigma_i^2}}.
$$
Since $\sum_{i=1}^n\theta_{\varepsilon_0,i}^2/\sigma_i^2= \varepsilon_0^2/4$ and $\omega'(\varepsilon_0)=\frac{\varepsilon_0}{2\sum_{i=1}^n\theta_{\varepsilon_0,i}/\sigma_i^2}$ by Lemma \ref{lem:modulus} below, we obtain a simplified form of $\hat\theta_0$:
$$
\hat\theta_0=\frac{\sum_{i=1}^n (\theta_{\varepsilon_0,i}/\sigma_i^2)Y_i}{\sum_{i=1}^n\theta_{\varepsilon_0,i}/\sigma_i^2}=\sum_{i=1}^n\tilde w_iY_i,
$$
where $\tilde w_i=\frac{\theta_{\varepsilon_0,i}/\sigma_i^2}{\sum_{j=1}^n\theta_{\varepsilon_0,j}/\sigma_j^2}$.
Therefore, the minimax affine MSE estimator has no intercept, and the optimal weights satisfy $\sum_{i=1}^n\tilde w_i=1$.
Furthermore, since $\theta_{\varepsilon_0,i}\ge 0$ for all $i=1,\ldots,n$ by Lemma \ref{lem:modulus}, we obtain $\tilde w_i\ge 0$ for all $i=1,\ldots,n$.

\begin{lemma}\label{lem:modulus}
    Let $\varepsilon>0$ and $\bm\theta_\varepsilon$ solve (\ref{modulus}).
    Then, the following holds: (i) $\sum_{i=1}^n\theta_{\varepsilon,i}^2/\sigma_i^2= \varepsilon^2/4$; (ii) $\theta_{\varepsilon,i}\ge 0$ for all $i=1,\ldots,n$; and (iii) $\omega(\cdot)$ is differentiable at $\varepsilon>0$ with $\omega'(\varepsilon)=\frac{\varepsilon}{2\sum_{i=1}^n\theta_{\varepsilon,i}/\sigma_i^2}$.
\end{lemma}

\begin{proof}[Proof of Lemma \ref{lem:modulus}]

We prove (i) by contradiction. Suppose $\sum_{i=1}^n\theta_{\varepsilon,i}^2/\sigma_i^2< \varepsilon^2/4$.
Let $\tilde{\bm\theta}(\delta)\in\mathbb{R}^{n+1}$ be such that $\theta_i(\delta)=\theta_{\varepsilon,i}+\delta$ for all $i=0,1,\ldots,n$.
Then, obviously, $\tilde{\bm\theta}(\delta)\in\Theta_g$.
Furthermore, there exists a sufficiently small $\delta>0$ such that $\sum_{i=1}^n\tilde\theta_{i}(\delta)^2/\sigma_i^2\le \varepsilon^2/4$.
For any such $\delta>0$, $\tilde\theta_0(\delta)>\theta_{\varepsilon,0}$.
This contradicts the assumption that $\bm\theta_\epsilon$ solves (\ref{modulus}).

Next, we prove (ii) by contradiction. Suppose there exists $i\ge 1$ such that $\theta_{\varepsilon,i}<0$.
Let $\tilde{\bm \theta}(\delta)\in\mathbb{R}^{n+1}$ be such that $\tilde\theta_i(\delta)=\max\{0,\theta_{\varepsilon,i}\}+\delta$ for all $i=0,1,\ldots,n$.
For any $\delta\ge 0$, $\tilde{\bm \theta}(\delta)\in \Theta_g$, since for all $i$ and $j$,
$$
|\tilde\theta_i(\delta)-\tilde\theta_j(\delta)|=|\max\{0,\theta_{\varepsilon,i}\}-\max\{0,\theta_{\varepsilon,j}\}|\le |\theta_{\varepsilon,i}-\theta_{\varepsilon,j}|\le C\|R_i-R_j\|.
$$
Furthermore, we have $\tilde\theta_i(0)^2=\theta_{\varepsilon,i}^2$ if $\theta_{\varepsilon,i}\ge 0$, and $\tilde\theta_i(0)^2=0<\theta_{\varepsilon,i}^2$ if $\theta_{\varepsilon,i}< 0$.
Since $\theta_{\varepsilon,i}<0$ for some $i\ge 1$, it follows that
$\sum_{i=1}^n\tilde\theta_i(0)^2/\sigma_i^2<\sum_{i=1}^n\theta_{\varepsilon,i}^2/\sigma_i^2\le \varepsilon^2/4$.
As a result, there exists a sufficiently small $\delta>0$ such that $\sum_{i=1}^n\tilde\theta_{i}(\delta)^2/\sigma_i^2\le \varepsilon^2/4$.
For any such $\delta>0$, $\tilde\theta_0(\delta)>\theta_{\varepsilon,0}$.
This contradicts the assumption that $\bm\theta_\epsilon$ solves (\ref{modulus}).

To prove (iii), we apply Lemma D.1 in Supplemental Appendix D of \cite{Armstrong.Kolesar2018}.
Our setting falls into their framework where $f=\bm\theta$, $\mathcal{F}=\mathcal{G}=\Theta_g$, $Kf=(\theta_1/\sigma_1,\ldots,\theta_n/\sigma_n)'$, and $Lf=\theta_0$ in their notation.
To apply their Lemma D.1, let $\iota\in\mathbb{R}^{n+1}$ denote the vector of ones. Then, we have $\iota\in\Theta_g$, $L\iota=1$, and $\bm\theta_{\varepsilon}+c\iota\in\Theta_g$ for all $c\in\mathbb{R}$.
By their Lemma D.1, $\omega(\cdot)$ is differentiable at $\varepsilon>0$ with
$$
\omega'(\varepsilon)=\frac{\varepsilon}{(K\iota)'(K\bm\theta_\varepsilon-K(-\bm\theta_\varepsilon))}=\frac{\varepsilon}{2\sum_{i=1}^n\theta_{\varepsilon,i}/\sigma_i^2}.
$$
\end{proof}
\newpage

\begin{center}
\Large
    \textit{Online Appendix}
    \normalsize
\end{center}

\section{Additional Tables for the empirical application}

\begin{table}[ht]
\centering
\begin{tabular}{llll}
  \hline
estimator & C & point & CI \\ 
  \hline
rdrobust &  & 0.138 & [-0.410, 0.686] \\ 
  rdbinary & C=0.5*Crot & 0.097 & [-0.185, 0.385] \\ 
  rdbinary & C=Crot & 0.103 & [-0.271, 0.469] \\ 
  rdbinary & C=1.5*Crot & 0.107 & [-0.323, 0.529] \\ 
   \hline
\end{tabular}
\caption{Narrow corruption at the cutoff 1 (N = 385)} 
\label{tab:application_narrow0}
\end{table}

\begin{table}[ht]
\centering
\begin{tabular}{llll}
  \hline
estimator & C & point & CI \\ 
  \hline
rdrobust &  & 0.534 & [ 0.168, 0.900] \\ 
  rdbinary & C=0.5*Crot & 0.070 & [-0.208, 0.345] \\ 
  rdbinary & C=Crot & 0.106 & [-0.246, 0.447] \\ 
  rdbinary & C=1.5*Crot & 0.087 & [-0.315, 0.471] \\ 
   \hline
\end{tabular}
\caption{Narrow corruption at the cutoff 2 (N = 218)} 
\label{tab:application_narrow1}
\end{table}

\begin{table}[ht]
\centering
\begin{tabular}{llll}
  \hline
estimator & C & point & CI \\ 
  \hline
rdrobust &  & -0.419 & [-1.133, 0.295] \\ 
  rdbinary & C=0.5*Crot &  0.293 & [-0.017, 0.606] \\ 
  rdbinary & C=Crot &  0.270 & [-0.129, 0.671] \\ 
  rdbinary & C=1.5*Crot &  0.215 & [-0.251, 0.671] \\ 
   \hline
\end{tabular}
\caption{Narrow corruption at the cutoff 3 (N = 225)} 
\label{tab:application_narrow2}
\end{table}

\begin{table}[ht]
\centering
\begin{tabular}{llll}
  \hline
estimator & C & point & CI \\ 
  \hline
rdrobust &  & -0.637 & [-1.382, 0.108] \\ 
  rdbinary & C=0.5*Crot & -0.131 & [-0.494, 0.228] \\ 
  rdbinary & C=Crot & -0.128 & [-0.564, 0.301] \\ 
  rdbinary & C=1.5*Crot & -0.092 & [-0.591, 0.386] \\ 
   \hline
\end{tabular}
\caption{Narrow corruption at the cutoff 4 (N = 139)} 
\label{tab:application_narrow3}
\end{table}

\begin{table}[ht]
\centering
\begin{tabular}{llll}
  \hline
estimator & C & point & CI \\ 
  \hline
rdrobust &  & 0.755 & [-0.641, 2.150] \\ 
  rdbinary & C=0.5*Crot & 0.142 & [-0.293, 0.576] \\ 
  rdbinary & C=Crot & 0.221 & [-0.351, 0.795] \\ 
  rdbinary & C=1.5*Crot & 0.300 & [-0.376, 0.940] \\ 
   \hline
\end{tabular}
\caption{Narrow corruption at the cutoff 5 (N = 116)} 
\label{tab:application_narrow4}
\end{table}

\begin{table}[ht]
\centering
\begin{tabular}{llll}
  \hline
estimator & C & point & CI \\ 
  \hline
rdrobust &  &  0.738 & [-0.016, 1.492] \\ 
  rdbinary & C=0.5*Crot & -0.004 & [-0.315, 0.306] \\ 
  rdbinary & C=Crot &  0.031 & [-0.339, 0.408] \\ 
  rdbinary & C=1.5*Crot &  0.080 & [-0.332, 0.494] \\ 
   \hline
\end{tabular}
\caption{Narrow corruption at the cutoff 6 (N = 73)} 
\label{tab:application_narrow5}
\end{table}

\begin{table}[ht]
\centering
\begin{tabular}{llll}
  \hline
estimator & C & point & CI \\ 
  \hline
rdrobust &  & 1.954 & [-0.238, 4.146] \\ 
  rdbinary & C=0.5*Crot & 0.306 & [-0.339, 0.903] \\ 
  rdbinary & C=Crot & 0.360 & [-0.446, 1.000] \\ 
  rdbinary & C=1.5*Crot & 0.395 & [-0.524, 1.000] \\ 
   \hline
\end{tabular}
\caption{Narrow corruption at the cutoff 7 (N = 46)} 
\label{tab:application_narrow6}
\end{table}


\end{document}